\documentclass[11pt]{article}
\usepackage{graphicx} % Required for inserting images
\usepackage[utf8]{inputenc}
\usepackage[T1]{fontenc}
\usepackage{amsmath,amssymb,amsthm}
\usepackage{fullpage}
\usepackage{hyperref}
\usepackage{enumitem}
\usepackage{dsfont}
\usepackage{mathtools}
\usepackage{cleveref}
\crefname{figure}{Fig.}{Figs.}
\Crefname{figure}{Fig.}{Figs.}
\crefname{section}{Sec.}{Sec.}
\Crefname{section}{Sec.}{Sec.}
\crefname{table}{Table}{Tables}
\Crefname{table}{Table}{Tables}
\crefname{appendix}{Appendix}{Appendices}
\Crefname{appendix}{Appendix}{Appendices}
\crefname{equation}{Eq.}{Eqs.}
\Crefname{equation}{Eq.}{Eqs.}
\crefname{theorem}{Theorem}{Theorems}
\Crefname{theorem}{Theorem}{Theorems}
\crefname{definition}{Definition}{Definitions}
\crefname{prop}{Proposition}{Propositions}
\crefname{conjecture}{Conjecture}{Conjectures}
\crefname{corollary}{Corollary}{Corollaries}
\usepackage{tikz-cd}
\usepackage{authblk}
\usepackage{booktabs}
\usepackage{braket}
\usepackage{caption}
\usepackage{tikz-cd}
\usepackage{bbm}
\usepackage{longtable}
\usepackage{array}
\usepackage{siunitx}
\DeclareMathOperator{\im}{\text{im}}
\begingroup\lccode`~=`& \lowercase{\endgroup
\newcommand{\spmat}[1]{%
  \left(
  \let~=&
  \begin{smallmatrix}#1\end{smallmatrix}
  \right)
}}

\newtheorem{theorem}{Theorem}[section]
\newtheorem{prop}{Proposition}

\newtheorem{definition}{Definition}
\newtheorem*{conjecture*}{Conjecture}
\newtheorem{corollary}{Corollary}[theorem]
\newtheorem{lemma}[theorem]{Lemma}

\newcommand{\F}{\mathbb{F}_2}
\newcommand{\FG}{\mathbb{F}_2[G]}

\newcommand{\te}{\text{E}}
\newcommand{\ts}{\text{S}}
\newcommand{\ba}{\mathbf{A}}
\newcommand{\bb}{\mathbf{B}}
\newcommand{\bc}{\mathbf{C}}
\newcommand{\bi}{\mathbf{1}}
\title{Magic tricycles: Efficient magic state generation with finite block-length quantum LDPC codes}
\author{Varun Menon\textsuperscript{*}\textsuperscript{†}}
\author{J. Pablo Bonilla Ataides\textsuperscript{*}}
\author{Rohan Mehta}
\author{Andi Gu}
\author{Daniel Bochen Tan}
\author{Mikhail D. Lukin}
\affil{Department of Physics, Harvard University, Cambridge, MA 02138, USA}
\date{\today}

\begin{document}
\maketitle

\begingroup
\renewcommand\thefootnote{}\footnotetext{\textsuperscript{*}These authors contributed equally to this work.}
\renewcommand\thefootnote{}\footnotetext{\textsuperscript{†}varunmenon@g.harvard.edu}
\endgroup

\begin{abstract}

The preparation of high-fidelity non-Clifford (magic) states is an essential subroutine for universal quantum computation, but imposes substantial space-time overhead. Magic state factories based on high rate and distance quantum low-density parity check (LDPC) codes equipped with transversal non-Clifford gates can potentially reduce these overheads significantly, by circumventing the need for multiple rounds of distillation and by producing a large number of magic states in a single code-block. As a step towards realizing efficient, fault-tolerant magic state production, we introduce a class of finite block-length quantum LDPC codes which we name tricycle codes, generalizing the well-known bicycle codes to three homological dimensions. These codes can support constant-depth physical circuits that implement logical $CCZ$ gates between three code blocks. To construct these constant-depth $CCZ$ circuits, we develop new analytical and numerical techniques that apply to a broad class of three-dimensional homological and balanced product codes. We further show that tricycle codes enable single-shot state-preparation and error correction, leading to a highly efficient magic-state generation protocol. Numerical simulations of specific codes confirm robust performance under circuit-level noise, demonstrating a high circuit-noise threshold of $>0.5\%$. With modest post-selection, certain tricycle codes of block-lengths of only $50-100$ qubits are shown to achieve logical error-rates of $6\times 10^{-10}$ or lower. Finally, we construct optimal depth syndrome extraction circuits for tricycle codes and present a protocol for implementing them efficiently on a reconfigurable neutral atom platform.
\end{abstract}

\newpage
\tableofcontents

\newpage
\section{Introduction}
\subsection{Background}
Universal fault-tolerant quantum computation requires the implementation of both Clifford and non-Clifford operations on the logical qubits of a quantum error correcting code \cite{gottesman1998theory}. Logical Clifford gates can often be implemented \textit{transversally} in many codes -- that is, by a physical circuit with a depth that is constant in the size of the code. Transversal gate architectures are naturally fault tolerant as the spread of errors between physical qubits of the code is inherently bounded. However, a seminal result of Eastin and Knill proves that it is not possible to implement a universal gate set transversally in a quantum error correcting code \cite{eastin2009restrictions}. Since it is desirable to implement as many gates transversally as possible, many architectures for fault-tolerant quantum computation utilize codes with a transversal Clifford gate set, while non-Clifford gates, also known as \textit{magic gates}, are implemented by teleporting specially prepared non-Clifford resource states into the codespace. One class of protocols for the fault-toleration preparation of such high-fidelity resource states, known as magic state distillation (MSD), relies on quantum error correcting codes with a transversal magic gate to refine a large number of noisy input magic states into a smaller number of output states with improved fidelity. To achieve the target output fidelity, typical MSD schemes such as the 15-to-1 distillation (15-1-3) protocol \cite{bravyi2005universal} require multiple levels of repeated distillation, as well as concatenation with an inner code of sufficiently large distance that supports transversal Clifford gates, contributing significantly to the space-time overhead of the protocol. As such, in state-of-the-art fault-tolerant architectures, the overhead associated with magic state distillation has been estimated to dominate total qubit and gate counts -- for example, magic state generation dominates the latest resource estimates of factoring RSA integers with Shor's algorithm despite significant optimizations \cite{gidney2025factor, zhou2025resource}. 

Over the last decade, remarkable progress has been made towards significantly reducing the overhead associated with magic state distillation by designing various codes with transversal non-Clifford gates and favorable code parameters \cite{bravyi2012magic, krishna2019towards, hastings2018distillation, campbell2012magic}. In fact, it has very recently been shown that families of quantum codes with transversal non-Clifford gates and asymptotically constant rate and relative distance exist, leading in principle to constant spatial overhead MSD protocols \cite{nguyen2024good, wills2024constant}. However, these proposals come with two important caveats. First, the proposed codes in such low-overhead distillation schemes have high-weight parity checks with check weights that are typically extensive in the block length of the code. As such these codes may not be obviously fault-tolerant on their own in the presence of circuit-level noise as they require syndrome extraction circuits with depths that grow with the block-length, and thus may not support a finite threshold or exponential sub-threshold error suppression with bare-ancilla based syndrome extraction. Second, these schemes implicitly assume that Clifford operations are noiseless to execute the Clifford encoding and unencoding circuits of the codes \cite{bravyi2005universal} -- an assumption that is justified when concatenating the distillation code with  an inner code that has transversal Clifford operations such as the two-dimensional color code. Practically, the overhead of these protocols is therefore exaggerated by the requirement of a sufficiently high distance inner code so that the distillation protocol is not bottlenecked by the error-rate of Clifford operations. Moreover, the Clifford encoding and unencoding circuits can have depths scaling linearly in the number of physical qubits of the code, leading to significant temporal overhead of such protocols.

Quantum low-density parity-check (qLDPC) codes, which simultaneously admit sparse stabilizer checks and favorable rate and distance scaling, have been proposed as a pathway toward reducing the cost of quantum error correction. In the context of quantum memory, qLDPC codes with asymptotically constant rates and relative distances have recently been shown to exist \cite{panteleev2022asymptotically, leverrier2022quantum}, and several architectures for implementing qLDPC memories on superconducting qubit and neutral-atom architectures have been proposed \cite{bluvstein2025architectural, xu2024constant, xu2025fast, bravyi2024high}. However, the application of general qLDPC codes to magic state generation protocols has only recently begun to be explored, in part because of the difficulty of constructing codes that host transversal non-Clifford gates. Utilizing appropriate qLDPC codes could dramatically reduce the overhead of MSD protocols for a number of reasons. Codes with high relative distances ensure effective error suppression during distillation, while high rate codes allow multiple magic states to be distilled in parallel within a single block, improving throughput. In contrast to traditional MSD schemes that rely on small codes concatenated many times, LDPC codes can potentially achieve the same target fidelity in a single round, greatly reducing space-time overhead. While this is also true for the aforementioned non-LDPC constant-overhead and low-overhead MSD protocols, crucially, qLDPC codes enable these properties while maintaining fault-tolerance with constant-depth syndrome extraction circuits. Moreover, for qLDPC codes with strongly-transversal non-Clifford gates (i.e a constant depth physical circuit of non-Clifford gates without a Clifford correction), there is no need for concatenation with an inner code with transversal Clifford gates. Constructions of practical LDPC magic state factories are thus highly desirable for fault-tolerant quantum computing architectures.

Motivated by these considerations, this work introduces a class of finite block-length, high-rate and distance quantum LDPC codes with a transversal logical $CCZ$ circuit action for magic state distillation, which we call `tricycle codes'. These codes are constructed as the three-dimensional \textit{balanced product} \cite{breuckmann2021balanced}, a particular generalization of the homological product, of classical binary linear codes over Abelian group algebras \cite{huffman2021concise}. This construction extends the two-block group algebra codes of   \cite{kovalev2013quantum, lin2024quantum} (also known as bicycle codes \cite{panteleev2021degenerate}) into three dimensions, allowing in principle for transversal gates from the third level of the Clifford hierarchy \cite{gottesman1999demonstrating}. We may thus also refer to the tricycle codes as `three-block Abelian group-algebra codes'. We highlight that for the task of magic state distillation, high-rate and high-distance finite block-length codes are sufficient for essentially any large-scale quantum computation. For example, codes with distances large enough to achieve a logical error rate of $~10^{-12}$ at attainable physical error rates are expected to be sufficient to produce magic states to execute large-scale quantum algorithms such as Shor's algorithm \cite{shor1994algorithms} fault-tolerantly. 

    The theoretical formalism we utilize for constructing transversal non-Clifford gates in three-dimensional product codes comes from endowing the codes with a \textit{cup-product}, an operation from algebraic topology and homological algebra, related to the triple-intersection of logical operators of the code, which we discuss further in \cref{sec::transversal_CCZ} and \cref{sec::cup}. At a high level, this formalism allows general product quantum CSS codes to be equipped with the necessary structure to host transversal gates ascending the Clifford hierarchy. In fact, the transversal non-Clifford gates of three-dimensional color codes can also be interpreted in terms of cup-products on the associated three-dimensional cellular complexes \cite{bombin2007topological, kubica2015unfolding}. Recent work by Breuckmann et al. introduced a general formalism to endow the cochain complexes associated with homological product and balanced product quantum codes with cup-product operations and classified the associated gates they give rise to \cite{breuckmann2024cups}. In this work, we first present a  modification of this formalism, leading to a set of conditions that high rate and distance tricycle codes with non-trivial $CCZ$ action must satisfy (see \cref{sec::cup}). The conditions resulting from this modified cup-product construction  are essential to finding tricycle codes with both high rates and relative distances. We then study the transversal non-Clifford circuits hosted by codes which satisfy these conditions in detail. More generally, we also introduce a new numerical method that is inspired by the cup-product construction for finding short-depth code-space preserving non-Clifford circuits on iterated balanced-product quantum codes, which we also apply to certain tricycle codes. We present several concrete examples of such high rate and distance tricycle codes, study their performance under realistic circuit level noise, and explore their application to practical magic state distillation protocols. 
    
\subsection{Summary of results}

The tricycle codes we construct achieve favorable parameters at relatively small block lengths, while hosting constant-depth circuits that implement a logical $CCZ$ action. We introduce several classes of tricycle codes with check weights $w_x \leq 12$, $w_z\leq 8$, where $w_x$ and $w_z$ denote the weights of the $X$ and $Z$ checks of the code respectively.  We use the notation $[[N,K,D]]$ to report the block-length $N$, number of encoded logical qubits $K$, and minimum distance $D$ of quantum codes. Some examples of codes we find have parameters $[[48,6,4]]$, $[[108,21,6]]$, $[[270,24,8]]$, $[[270,12,12]]$, $[[480,12,\leq 17]]$, exhibiting highly competitive distances and rates that are compatible with a transversal non-Clifford action. A more complete list of codes along with the depths of the non-Clifford transversal circuits they host can be found in \cref{table::all_codes}. 

We show that tricycle codes host constant depth physical circuits that preserve the code-space, consisting solely of $CCZ$ gates between three code blocks. The logical action of the $CCZ$ circuits depends on a choice of a representative logical operator basis in each code-block, but is generically a circuit of logical $CCZ$ gates. When such a circuit acts on the $\overline{\ket{+,+,+}}^{\otimes K}$ logical state of three code-blocks it produces a high magic state known as a \textit{hypergraph} magic state \cite{zhu2025topological, chen2024magic}. Such a hypergraph magic state embeds $K_{CCZ}\leq K$ many disjoint $\overline{CCZ} \cdot \overline{\ket{+,+,+}}$ magic states which can be extracted using gate-teleportation techniques. The physical $CCZ$ circuit, factors determining their depths, and the logical action they induce for particular codes are discussed in \cref{sec::transversal_CCZ} and \cref{sec::cup}.

We further show that the tricycle codes are single-shot \cite{bombin2015single} in the $Z$- measurement basis. We show in \cref{sec::single_shot} that this property enables a distillation protocol that requires only a constant number of rounds of syndrome extraction during error correction cycles instead of $O(d)$ rounds, where $d$ is the code distance. Moreover, we present numerical evidence that tricycle codes exhibit the stronger notion of single-shot state-preparation, allowing the initial logical state of a magic state preparation circuit to be prepared in constant depth. To the best of our knowledge, this is the first example of such a single-shot magic state distillation protocol. 

We analyze the performance of select tricycle codes under a realistic circuit level noise model. By comparing the performance of codes across a range of block-lengths, we are able to ascertain a threshold of $0.5\%$ under circuit-level noise with an MLE decoder for a certain family of tricycle codes. We also study how post-selection and full error-detection can improve the performance of these codes. To do so, we utilize a recently developed efficient cluster-based post-selection method for quantum LDPC codes \cite{lee2025efficient} combined with Belief-Propagation + Localized Statistics Decoding (BP+LSD)~\cite{hillmann2025localized}. With full error detection, we show that even the smallest tricycle code we study in detail (with parameters $[[48,6,4]]$) can achieve logical error rates of approximately $6\times 10^{-10}$ at an acceptance fraction of roughly $30\%$. For a larger code, we estimate logical error rates $<10^{-12}$ with full error detection at a $\approx 9\%$ acceptance fraction, compatible with full-scale fault tolerant quantum computation. Note that these logical error rates are achieved in a numerical simulation of the code as a memory under multiple rounds of error correction (with post-selection). In particular, they do not account for the error introduced by the $CCZ$ circuit when determining the fidelity of the final magic-state. However, since the $CCZ$ circuits are relatively short (depth $8$ for the codes discussed above), we anticipate that these estimates from the logical memory performance are a close proxy to the fidelity of the target magic-state -- we discuss this perspective in more detail in \cref{sec::noise_sim}.

We note that although the codes we present are finite-block length codes, the construction allows us to define a family of codes by choosing a single high-performing tricycle code from a given block-length. Although such a code family is unlikely to be asymptotically good, we expect that a larger block-length code with a larger distance can always be found -- a claim that is supported by the distance lower bound we prove in \cref{sec::balanced}. Defining a threshold with respect to such a constructed code family allows one to make  guarantees of sub-threshold noise suppression by selecting codes of arbitrarily large block-lengths with larger distances. 

We then show how to construct compact syndrome extraction circuits of optimal depth for all tricycle codes. Finally, we present a concrete protocol for efficiently implementing the preparation and syndrome extraction of these codes on a reconfigurable neutral-atom array platform -- a video illustrating the atom-moving scheme for syndrome extraction of tricycle codes is provided in the ancillary files that accompany this manuscript.

We also clarify some subtle points about the application of the cup-product gate formalism of \cite{breuckmann2024cups}. In particular, the success of this formalism for producing transversal non-Clifford circuits on tricycle codes with high rate and distance requires a generalization of the conditions presented in section $5$ of \cite{breuckmann2024cups}. These resulting conditions will likely also be useful for other homological product and balanced product code families. We discuss these points on the theory of cup-product based gates in \cref{sec::cup}. In \cref{sec::numerical_ccz}, we also present a complementary numerical method that generalizes the cup-product construction for finding code-space preserving $CCZ$ circuits on three-dimensional balanced-product codes.

Our work serves as a step toward a new class of distillation protocols that leverage the structural properties of quantum LDPC codes to reduce resource overhead in universal fault-tolerant computation. Such magic-state generation schemes are naturally suited to fault-tolerant quantum computing architectures that utilize LDPC quantum memories along with universal adapters for magic state injection, such as in Ref.~\cite{xu2025fast, yoder2025tour} or using efficient code-switching protocols with lower-dimensional product codes \cite{li2025transversal,tan2025single,golowich2025constant}. 

\subsection{Related work}

Several recent works introduce binary quantum LDPC code families with asymptotically high rates and relative distances with transversal logical $CCZ$ action. Golowich and Lin \cite{golowich2024quantum} constructed a family of codes based on homological products of classical  Sipser-Spielman expander codes \cite{sipser2002expander} with local Reed-Muller codes \cite{reed1953class, muller1954application} with near-constant rate and relative distance that supports transversal logical non-Clifford gates. These codes however have check-weights that are not constant but grow logarithmically in the block-length. Lin also introduced a general formalism for constructing LDPC quantum code families with transversal non-Clifford gates based on algebraic sheaves \cite{lin2024transversal}. Similarly, Zhu introduced a family of topological codes with constant rate and polynomial distance with transversal $CCZ$ gates  \cite{zhu2025topological, zhu2025transversal}. Quantum rainbow codes were introduced as generalizations of color-codes with transversal non-Clifford gates to higher-dimensional simplicial complexes \cite{scruby2024quantum}, achieving asymptotically constant rate and logarithmic distance. Finally, Breuckmann et al. \cite{breuckmann2024cups} discuss how to construct constant rate and polynomial distance families with transversal gates that ascend the Clifford hierarchy using iterated homological and balanced products of chain complexes and the cup-product gate formalism they develop. 

Adding to these results, this work focuses on a particular class of high-performance finite block-length codes. It is unclear precisely how large the smallest elements of the code families in the aforementioned works are, but constructions based on three-fold homological products of classical codes likely involve at-least several thousand and potentially tens of thousands of physical qubits for the smallest attainable block-lengths with good parameters, due to the cubic growth of the size of the quantum code with respect to the classical codes. The significantly smaller block-lengths of the tricycle codes presented in this work is one of the primary reasons we chose to study them. A direct comparison of the overhead and performance of the LDPC codes from the above constructions with each other and the tricycle codes we present would be an interesting direction for future work.

\section{Main results}
\subsection{Tricycle codes}\label{sec::tricycle_main}

Throughout this work, we denote the binary field by $\mathbb{F}_2 \equiv \{0,1\}$ where arithmetic is modulo $2$. The tricycle codes are quantum CSS codes \cite{calderbank1996good, steane1996multiple} defined by the binary block parity check matrices 

\begin{equation} \label{eq::parity_check_mats}
H_X = \begin{bmatrix}
\mathbf{A}^T & \mathbf{B}^T & \mathbf{C}^T 
\end{bmatrix} \in \mathbb{F}_2^{n_G\times 3n_G}
\quad \quad
H_Z = \begin{bmatrix}
\mathbf{C} & \mathbf{0} & \mathbf{A} \\
\mathbf{0} & \mathbf{C} & \mathbf{B} \\
\mathbf{B} & \mathbf{A} & \mathbf{0}
\end{bmatrix}\in\mathbb{F}_2^{3n_G\times 3n_G}
\end{equation}

where $\mathbf{A}=\sum_{i=1}^{w_a}\mathbf{A_i} \in \mathbb{F}_2^{n_G\times n_G}$, $\mathbf{B}=\sum_{i=1}^{w_b}\mathbf{B_i} \in \mathbb{F}_2^{n_G\times n_G}$, $\mathbf{C}=\sum_{i=1}^{w_c}\mathbf{C_i} \in \mathbb{F}_2^{n_G\times n_G}$ for constants $w_a, w_b, w_c$, and $\mathbf{A_i}, \mathbf{B_i}, \mathbf{C_i}$ are permutation matrices -- i.e matrices with a single 1 in every row and column. The parameter $n_G$ is the linear size of each block and is related to the size of a finite group (hence the subscript $G$) -- the details of this connection are discussed in \cref{sec::balanced}. These parity check matrices thus define codes on $N = 3n_G$ qubits. Modulo two arithmetic implies immediately that the CSS orthogonality condition $H_Z H_X^T = 0$ is satisfied. Moreover, the matrices $\mathbf{A}$, $\mathbf{B}$, $\mathbf{C}$ commute pair-wise. There are $n_G$ many $X$-checks of weight $w_a+w_b+w_c$ and $3n_G$ many $Z$-checks which can be partitioned into three sets of $n_G$ checks each with weights $w_c + w_a$, $w_c + w_b$, and $w_b + w_a$ respectively. By choosing $w_a, w_b, w_c \leq w$ for a constant $w$, the resulting codes are $(\leq w, \leq 3w)$-LDPC codes in the $X$ basis and $(\leq 2w, \leq 2w)$-LDPC in the $Z$ basis, where $(p,q)-$LDPC means that every qubit is involved in at most $p$ checks and each check involves at most $q$ qubits. Examples of particular tricycle codes in \cref{table::all_codes} demonstrate that the checks are typically highly redundant -- often leading to favorable code parameters and ease of decoding -- since the total number of checks $4n_G \gg N-K = 3n_G-K$. This observation is encapsulated by the fact that tricycle codes  possess \textit{metachecks} in the $Z$ basis -- these are relations between the $Z$ checks encoded in a matrix $H_{meta}$ such that $H_{meta}\cdot H_z = 0$. We discuss the metacheck structure and its implications in \cref{sec::single_shot}.

 To ensure practical feasibility of implementing these codes in early fault-tolerant architectures below pseudo-threshold error rates, we restrict our attention to tricycle codes with a maximum check weight at most $12$. More specifically, we consider three types of tricycle codes. The first is constructed with $w_a=4$, $w_b=w_c=2$, producing codes with weight $8$ $X$-checks, and $Z$-checks of weight $4$ and $6$ -- we call these codes $\mathbf{4-2-2}$ codes, referring to the weights $w_a,w_b,w_c$. The second class of codes we consider has $w_a=w_b=4$, $w_c=2$, and therefore has $X$-checks of weight $10$ and $Z$-checks of weight $8$ and $6$ -- we refer to these as $\mathbf{4-4-2}$ codes. Finally, we consider codes with $w_a=w_b=w_c=4$, for which the $X$-checks have weight $12$ and the $Z$-checks have weight $8$ -- we call these $\mathbf{4-4-4}$ codes. One may similarly consider $\mathbf{2-2-2}$ codes with weight $6$ $X-$checks and weight $4$ $Z-$checks, but we empirically found that the rates and distances of such codes with a non-trivial transversal $CCZ$ action are not as favorable. In particular, we could not find any such codes with $K>3$ encoded logical qubits and distance $D>7$ -- such $2-2-2$ codes with $K=3$ encoded logical qubits and distance $D\leq 7$ have been studied in other recent works \cite{jacob2025singleshot, li2025transversal}. 
 
 The reason we consider $w_a,w_b,w_c$ to be $2$ or $4$ is that the specific conditions for tricycle codes to host transversal $CCZ$ circuits (see \cref{sec::cup}) cannot be satisfied when any of $w_{a,b,c} = 3$, and $w_{a,b,c}=1$ leads to a trivial code. One may consider $w_{a,b,c} > 4$ to search for codes with even more favorable parameters than the ones we present here, but we empirically found $6-2-2$ codes to have parameters no better than $4-2-2$ and $4-4-2$ codes, while going to even higher check-weight will likely lead to poor thresholds.

\begin{table}[h!]
\centering
\begin{tabular}{ccccccc} % 5 columns, all left-aligned; use c or r for center/right alignment
\toprule

$\mathbf{N}$ & $\mathbf{K}$ & $\mathbf{D_X}$ & $\mathbf{D_Z}$ & \textbf{Type} & $\mathbf{CCZ}$  \textbf{depth} & $\mathbf{CCZ}$ \textbf{method}\\
\midrule

$48$ & $6$ & $8$ & $4$ & $4-2-2$ & $8$ & STCP\\
$84$ & $6$ & $12$ & $5$ & $4-2-2$ & $8$ & STCP  \\
$108$ & $6$ & $12$ & $6$ & $4-2-2$ & $8$ & STCP \\
$240$ & $6$ & $\leq 22$ & $8$ & $4-2-2$ & $8$ & STCP \\
$480$ & $6$ & $\leq 36$ & $10$ & $4-2-2$ & $8$ & STCP \\

$108$ & $12$ & $11$ & $4$ & $4-4-2$ & $16$ & STCP\\
$180$ & $12$ & $15$ & $6$ & $4-4-2$ & $32$ & STCP \\

$108$ & $15$ & $12$ & $6$ & $4-4-4$ & $96$ & NLR\\
$270$ & $24$ & $15$ & $8$ & $4-4-4$ & $96$ & NLR\\

$324$ & $12$ & $\leq 32$ & $\leq 12$ & $4-4-4$ & $128$ & STCP \\
$480$ & $15$ & $\leq 48$ & $\leq 14$ & $4-4-4$ & $128$ & STCP\\

% Add more rows as needed
\bottomrule
\end{tabular}
\caption{\textbf{Examples of tricycle codes and their parameters for various block-lengths}. $N$ is the block-length, $K$ the number of encoded logical qubits, $D_X$ and $D_Z$ are the minimum weights of logical $X$-type and $Z$-type operators respectively. The distances of codes are found exactly by using either a SAT-solver to find the shortest error \cite{gidney2021stim} or a mixed integer program \cite{landahl2011fault} when feasible. Otherwise, distances are estimated -- sometimes, with strong statistical guarantees, in which case we report the distance as exact\protect\footnotemark  -- using the low-weight biased Monte-Carlo sampling method discussed in \cref{sec::distance_finding} with 100 million shots. The construction of these specific codes is described in \cref{sec::balanced} and \cref{table::code_polys}. The column named `Type' denotes the class of tricycle codes, labeled by the weights $w_a-w_b-w_c$ from \cref{eq::parity_check_mats}. The depths of the code-space preserving $CCZ$ circuits with non-trivial logical action are noted for each code. The final column named `$CCZ$ Method' indicates whether the code and non-trivial $CCZ$ circuit was constructed using the symmetric triple cup-product method (STCP) -- see \cref{sec::cup} -- or using the numerical Leibniz rule method (NLR) -- see \cref{sec::numerical_ccz}.}
\label{table::all_codes}
\end{table} 

The codes in \cref{table::all_codes} have imbalanced distances with $d_Z\leq d_X$. In \cref{sec::balanced}, we prove that this is the case for all tricycle codes. The imbalanced distances of these codes is a potentially useful feature for experimental platforms which are typically dominated by noise biased towards one basis. A detailed exploration of how noise bias can boost the experimental performance of tricycle codes is left for future work. 
\footnotetext{Specifically, the method we use returns a confidence interval on the distance conditioned on certain assumptions on the distribution of minimum-weight code-words. These assumptions can in turn be hypothesis tested using the empirical distribution of low-weight code-words found from sampling. We report the distance as exact if the confidence is $>99.9\%$ and if the $p$-value associated with the hypothesis test is $>10\%$. Details of the method are in \cref{sec::distance_finding}.}

Overall, the high distances of tricycle codes combined with the transversal $CCZ$ circuits they host (which we discuss in \cref{sec::transversal_CCZ}) make them ideal candidates for magic state factories for various target logical error rates. While the rates of these codes are lower than their two-dimensional counterpart bicycle codes (see the codes in Ref.~\cite{bravyi2024high} for example), they are much higher than the rates of three dimensional binary homological product codes of comparable block-length, including the $3D$ color code.

For the remainder of this work, we primarily focus on select $4-2-2$ type tricycle codes that are practically relevant due to their lower check-weights and depth-$8$ $CCZ$ circuits. The $4-4-4$ tricycle codes generally exhibit higher rates and distances, but this comes at the price of higher check-weights and a lower circuit-noise threshold of $\approx 0.4\%$, as well as significantly deeper non-trivial $CCZ$ circuits on average -- additional details on the $4-4-4$ codes
are presented in \cref{sec::cup}, \cref{sec::numerical_ccz}, and \cref{sec::numerics-methods}. We empirically found that the $4-4-2$ type codes did not offer substantial improvements in rates, distances, or thresholds compared to the $4-2-2$ and $4-4-4$ type codes to the extent of our code search -- two examples of $4-4-2$ codes with non-trivial $CCZ$ are nonetheless presented in \cref{table::all_codes}.

\subsection{Transversal $CCZ$ circuits}\label{sec::transversal_CCZ} 

Tricycle codes support constant-depth, code-space preserving circuits built from physical $CCZ$ gates acting across three code blocks. Each block contains $3n_G$ qubits, partitioned into three equal-size sectors labeled $I$, $II$, and $III$, each with $n_G$ qubits. We label each qubit as $q_i^j$, where $i \in \{1,2,3\}$ denotes the code block and $j \in \{I,II,III\}$ the sector. Within each sector, qubits are indexed arbitrarily but consistently across blocks.

The structure of a transversal $CCZ$ circuit is encoded in a binary function
\begin{equation}
    f_{CCZ} : Q \times Q \times Q \rightarrow \mathbb{F}_2,
\end{equation}
where $Q$ is the set of all physical qubits. A value $f_{CCZ}(q_1^{j_1}, q_2^{j_2}, q_3^{j_3}) = 1$ indicates a physical $CCZ$ gate between the specified qubits, one from each code block. Otherwise, no gate is applied. This is illustrated in \cref{fig:CCZ_connectivity}a). $f_{CCZ}$ is then linearly extended to be defined as a trilinear function on all computational basis states. To define a valid, code-space preserving circuit, $f_{CCZ}$ must satisfy two additional conditions:
\begin{enumerate}
    \item $f_{CCZ}$ must vanish when any one or more of the arguments correspond to a binary string $b$ defining an $X$-type stabilizer -- i.e $X^b=X^{b_1}\otimes X^{b_2}\cdots \otimes X^{b_{3n_G}}$ is an X-type stabilizer -- and the other arguments similarly correspond to non-trivial X-type logical operators. This condition ensures $f_{CCZ}$ descends to a well-defined function on the space of logical codewords.
    \item $f_{CCZ}$ must define a constant-depth circuit of physical $CCZ$ gates.
\end{enumerate}
We expand on both conditions and their homological interpretation in \cref{sec::cup}. See Proposition \ref{prop:fccz_def} for how $f_{CCZ}$ is defined for tricycle codes.

In recent work, Breuckmann et al.~\cite{breuckmann2024cups} introduced a formalism for constructing appropriate trilinear functions $f_{CCZ}$ that satisfy conditions $1$ and $2$ based on algebraic operations on homological product codes known as \textit{cup-products}. We develop a modification of this formalism and derive conditions that tricycle codes must satisfy to host such transversal $CCZ$ circuits -- see \cref{sec::cup} for details on the conditions. For reasons elaborated in \cref{sec::cup}, we refer to this framework and the associated conditions on the codes as the \textit{symmetric triple cup-product} formalism for transversal $CCZ$ gates, which applies to both three-dimensional homological product and balanced product \cite{breuckmann2021balanced} quantum LDPC codes -- the latter of which tricycle codes are an example of. We then show that there are tricycle codes with favorable rates and distances which satisfy the conditions of the symmetric cup-product framework. We note that tricycle codes constructed using the original cup-product conditions of \cite{breuckmann2024cups} have recently been studied in \cite{jacob2025singleshot, li2025transversal} -- the resulting $2-2-2$ codes always have $K=3$ while $4-2-2$ and other higher weight codes have $D=2$. We find that the new symmetric cup-product conditions we derive in this work are necessary to circumvent these limitations, and are more generally applicable to other classes of higher-dimensional balanced product quantum codes.

In the following discussion, as a slight abuse of terminology, we refer to the `depth' of the $CCZ$ circuit and the maximum degree of any physical qubit -- i.e the number of $CCZ$ gates a qubit is involved in -- interchangeably. Indeed the maximum degree is the relevant quantity that controls the spread of errors between physical qubits of the code-blocks. The minimal depth for scheduling the gates in the circuit is upper bounded by a small constant factor of the degree (see \cref{sec::cup}) \cite{obszarski2017edge}.

The $4-2-2$ type tricycle codes which satisfy the conditions of the symmetric triple cup-product generically host depth $8$ $CCZ$ circuits. Similarly, the $4-4-2$ codes have depth $16$ or depth $32$ circuits, and the $4-4-4$ codes have depth $12$, $64$, or $128$ depth circuits. We explain the choices made in the code-construction that lead to these circuit depths in \cref{sec::cup}.  

In addition to the symmetric triple cup-product framework, we introduce a more general formalism for constructing valid $f_{CCZ}$ functions on balanced-product quantum codes \cite{breuckmann2021balanced} that is amenable to an extensive numerical search method for codes with block-lenghts of a few hundred qubits -- see \cref{sec::numerical_ccz} for details. Using this numerical method, we are able to find additional $4-4-4$ tricycle codes with better rates and distances along with shorter-depth physical $CCZ$ circuits than those that result from the cup-product formalism -- we give two examples of such $4-4-4$ codes in \cref{table::all_codes}, while we leave a more thorough exploration of this method for future work.

The physical qubits of each  tricycle code-block can be naturally partitioned into three sectors in accordance with block structure of the parity-check matrices in \cref{eq::parity_check_mats}. The code-space preserving $CCZ$ circuits always act across sectors of the three code-blocks, as illustrated in \cref{fig:CCZ_connectivity}. The logical action of a transversal $CCZ$ circuit is derived by restricting $f_{CCZ}$ to the logical subspace, using representatives of logical $X$ operators for each block. This produces logical $\overline{CCZ}$ gates between triples of logical qubits whenever $f_{CCZ}(l_1^i, l_2^j, l_3^k) = 1$, where $\{l_1^i\}_{i=1\cdots k}$ are a representative basis set of logical $X$ operators of the first code-block, with $\{l_2^j\}$ and  $\{l_3^k\}$ defined similarly. The resulting state, obtained by applying the logical circuit to $\overline{\ket{+,+,+}}^{\otimes K}$, defines a \textit{hypergraph magic state}~\cite{zhu2025topological, chen2024magic}, with vertices as logical qubits and hyperedges as $\overline{CCZ}$ gates. Recent work studies such hypergraph magic states, arguing that they can contain a large amount of magic and can be classically hard to simulate \cite{chen2024magic}.

\begin{figure}[ht!]
    \centering
    \includegraphics[width=0.95\linewidth]{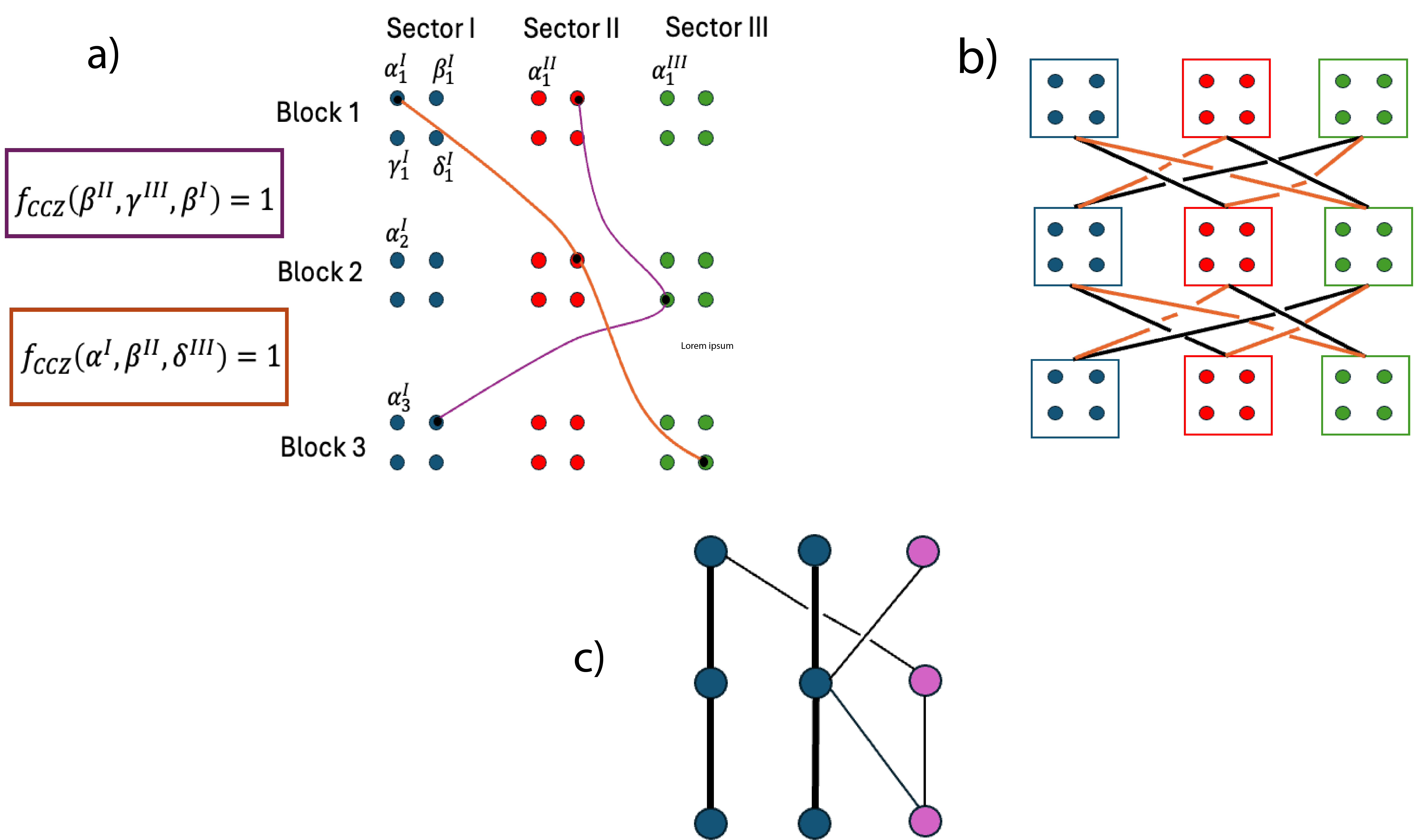}
    \caption{\textbf{Structure of transversal $CCZ$ gates of tricycle codes.} \textbf{a)} Schematic of a transversal $CCZ$ circuit on a 12-qubit code. Each code block is partitioned into three sectors of 4 qubits, labeled $\alpha, \beta, \gamma, \delta$ with subscripts and superscripts indicating the code block and sector. Colored curves denote $CCZ$ gates between triples of qubits where $f_{CCZ}$ is non-zero. \textbf{b)} Structure of transversal $CCZ$ circuits: all sectors participate via two disjoint sets of circuit layers denoted by orange and black edges that can individually be parallelized across qubits. Each qubit undergoes a maximum of $l$ black and a maximum of $m$ orange $CCZ$ gates, leading to a maximum degree of $l+m$. \textbf{c)} Logical $CCZ$ connectivity after basis optimization for a $K=3$ code . Circles denote logical qubits; rows correspond to separate code blocks. Thick black lines indicate usable $\overline{CCZ}$ gates for magic state distillation ($K_{CCZ}=2$ shown), while thin lines involve gauge qubits (pink), initialized in $\overline{\ket{0}}$. Blue circles represent logical qubits in disjoint triples connected only to other qubits in the triple or to gauge qubits.}
    \label{fig:CCZ_connectivity}
\end{figure}

While the produced hypergraph magic states are interesting and potentially computationally useful in their own right, we would also like to be able to extract smaller component magic gates from the logical circuit for quantum computation. One way to achieve this is to demand that the logical circuit consists of $K$ $\overline{CCZ}$ gates across code-blocks, each acting on disjoint triples of logical qubits. However, a choice of logical bases that produces such a logical circuit may not always exist (and is unlikely to). We may instead try to maximize the number of extractable $\overline{CCZ}$ gates on disjoint triples of logical qubits $K_{CCZ} \leq K$. Finding such a subset corresponds to computing the \textit{subrank} of the $f_{CCZ}$ function restricted to logical operators, a known problem in tensor rank theory~\cite{christandl2023gap, kopparty2020geometric, golowich2024quantum, lin2024transversal}. We use a mixed-integer programming method to find feasible solutions to this problem (\cref{sec::logical_opt}), and below we report the best values of $K_{CCZ}$ found for each code in \cref{table::all_codes}. For most codes, our solvers eventually time out within computational resource constraints without proving optimality, suggesting that larger $K_{CCZ}$ values than the ones we report are possible to find with tailored heuristic optimization strategies -- a promising avenue for future work. The resulting logical circuits enable the extraction of $K_{CCZ}$ usable $\overline{CCZ}$ gates. Logical qubits not used in these gates are treated as \textit{gauge qubits} and initialized in the $\overline{\ket{0}}$ state, which nullifies any $\overline{CCZ}$ gates they participate in (\cref{fig:CCZ_connectivity}d) -- see \cref{sec::cup} for details. 

Using the methods described in \cref{sec::logical_opt}, we are able to find at-least $K_{CCZ} \geq 2$ for all the codes in \cref{table::all_codes} using the associated $CCZ$ circuits. We attempted to find larger values of $K_{CCZ}$ for these codes, but the solver we used did not converge within a reasonable amount of time. We therefore emphasize that the $K_{CCZ}$ values found here are lower bounds returned by exact methods that quickly become infeasible for larger $K_{CCZ}$ -- this is primarily due to the current lack of well-performing efficient heuristics for the binary tensor-subrank problem, and we anticipate that progress towards such heuristics will lead to a higher yield of individually distillable $\overline{CCZ}$ type magic states for the codes in \cref{table::all_codes}. 

A subtle point that is understated in LDPC magic state distillation proposals is that extracting such $K_{CCZ}$ disjoint $\overline{CCZ}$ gates from the full logical circuit requires the ability to selectively initialize the gauge logical qubits in the $\overline{\ket{0}}$ state while other logical qubits are initialized in the $\overline{\ket{+}}$ state. For example, the proposals in Ref.~\cite{golowich2024quantum, lin2024transversal, zhu2025topological, zhu2025transversal, breuckmann2024cups} assume the ability to do this. However, selective state initialization is non-trivial for LDPC codes, and in general it is not clear that this is possible in constant depth. Extensions of the techniques introduced in Ref.~\cite{xu2025fast} for initializing arbitrary Pauli product states in homological product codes be useful for selective initialization. Recent work has also introduced batched Pauli-product state-preparation techniques that can achieve selective state-initialization on arbitrary qLDPC codes in constant overhead, as long as many code-blocks are initialized at once in the same state \cite{xu2025batched} -- we anticipate that these techniques will be very useful in a LDPC $\ket{\overline{CCZ}}$ magic state factory setting \cite{xu2025batched}. More generally however, we would like to emphasize an alternative perspective of viewing the global hypergraph magic state of the $K$ logical qubits as a high-magic resource state \cite{chen2024magic}, that can be used to perform useful quantum computations in a target computational code with transversal Clifford operations using appropriate circuit compilation techniques. We leave a detailed exploration of this perspective along with concrete protocols for efficient gate-teleportation between tricycle codes and target Clifford gate-set computation codes for future work.

\subsection{Single-shot distillation} \label{sec::single_shot}

We now describe how tricycle codes enable single-shot magic state generation. By initializing three code blocks fault-tolerantly in $\ket{\overline{+}}^{\otimes K}$ using a constant number of error correction rounds and subsequently applying the transversal $CCZ$ circuit described in the previous section, we obtain high-fidelity hypergraph magic states—all within constant circuit depth, as illustrated in \cref{fig:ccz-factory}.

\begin{figure}[ht!]
    \centering
    \includegraphics[width=0.65\linewidth]{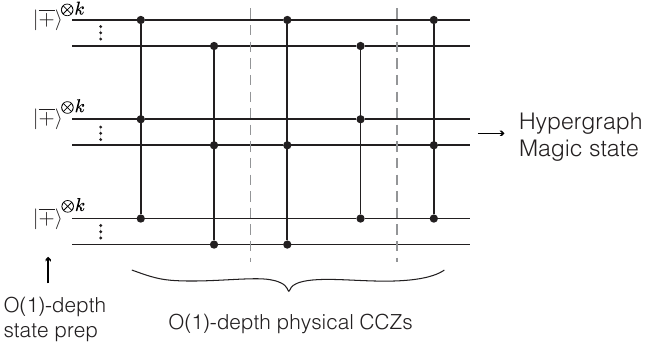}
    \caption{\textbf{Single-shot distillation with tricycle codes.} The logical $\ket{\overline{+}}^{\otimes K}$ state of the tricycle code can be prepared fault-tolerantly in constant depth by harnessing the code’s intrinsic resilience in one basis—namely, by preparing the physical qubits such that the associated stabilizer checks are deterministic—together with single-shot error correction in the complementary, non-deterministic, basis. The logical non-Clifford operation is implemented via a constant-depth circuit composed of physical $CCZ$ gates. The output is a logical hypergraph magic state which embeds $K_{CCZ} \leq K$ disjoint logical $\ket{\overline{CCZ}}$ resource states.}
    \label{fig:ccz-factory}
\end{figure}

Our approach to constant-depth preparation of logical $\ket{\overline{+}}^{\otimes K}$ hinges on two essential properties: intrinsic fault tolerance in one stabilizer basis during state initialization, and the presence of metachecks that enable single-shot error correction in the complementary basis. Specifically, we initialize the physical data qubits in the product state $\ket{+}^{\otimes n}$ and perform a single round of stabilizer measurements. As tricycle codes are CSS codes, the $X$-type checks are implemented as products of physical $X$ operators acting on the input $\ket{+}$ states. Thus, in the absence of errors, the measurement outcomes for the $X$ checks are deterministic, which allows $Z$ errors that occur during state preparation to be detected and corrected

However, initial $Z$-type stabilizer measurements produce non-deterministic $\pm 1$ outcomes, which must be reliably fixed to $+1$ on hardware before applying the $CCZ$ gate~\cite{bluvstein2025architectural}. In contrast to codes like the surface code, which lack single-shot state preparation, the $Z$ parity check matrices for tricycle codes contain a large number of redundant checks. These subsequently form a robust set of metachecks, which resolve a classical code on the syndrome bits, allowing a decoder to identify and correct syndrome measurement errors. In the $Z$ basis, these take the form

\begin{equation}\label{eq:metachecks}
    H_{meta} = \begin{bmatrix}
        \mathbf{B} & \mathbf{A} & \mathbf{C}
    \end{bmatrix}
\end{equation}
where the $\mathbf{A},\mathbf{B},\mathbf{C}$ matrices are the same as in \cref{eq::parity_check_mats}. The tolerance to measurement qubit errors is determined by the single-shot distance $D_Z^{SS}$~\cite{campbell2019theory}, which is defined as the minimum weight of a faulty syndrome that passes all metachecks but is not a $Z$ syndrome. A large $D_Z^{SS}$ indicates considerable redundancy, facilitating the identification of sparse syndrome errors. For all codes examined (see \cref{table::all_codes}), we found $D^{SS}_Z = D_Z$, or at least their upper bounds when exact distances were intractable to compute (see \cref{sec::distance_finding}). Concurrently, $D^{SS}_Z = D_Z$ was proven for 
this $H_{meta}$ analytically in \cite{jacob2025singleshot}. It is also always possible to extend $H_{meta}$ into a full-rank matrix \cite{campbell2019theory} by forcing the decoder to repair the noisy syndrome into one that can be produced by a pattern of data qubit errors, resulting in $D^{SS} = \infty$. 

Metachecks combine with \textit{soundness} properties \cite{campbell2019theory} to enable fault-tolerant inference of the $Z$ stabilizer eigenvalues after only a constant number of error correction rounds under select error models. It is worth noting that this task of preparing the $Z$ stabilizers to $+1$ from a completely random initialization is a more challenging task than just single shot error correction \cite{hong2024single, bombin2015single}, where a decoder may assume that the $Z$ stabilizers were fault-tolerantly fixed to $+1$ prior to the observed noise. For example, many hypergraph product codes enable single-shot error correction but not single-shot state preparation~\cite{hong2024single}.

In \cref{sec:soundness}, we prove soundness properties for a certain sub-class of tricycle codes in the $Z$ basis and review the ensuing guarantee on single-shot $Z$ basis decodability, which applies to state preparation \cite{xu2025fast}. Moreover, we perform direct simulations of single-shot state preparation for three codes with good finite-size encoding rates and distances, showing that general tricycle codes appear to host single-shot state-preparation properties. Our results, shown in \cref{fig:state-preparation} demonstrate exponential suppression of the logical error with increased distance, consistent with fault-tolerant state preparation.

\begin{figure}[h!]
    \centering
    \includegraphics[width=0.65\linewidth]{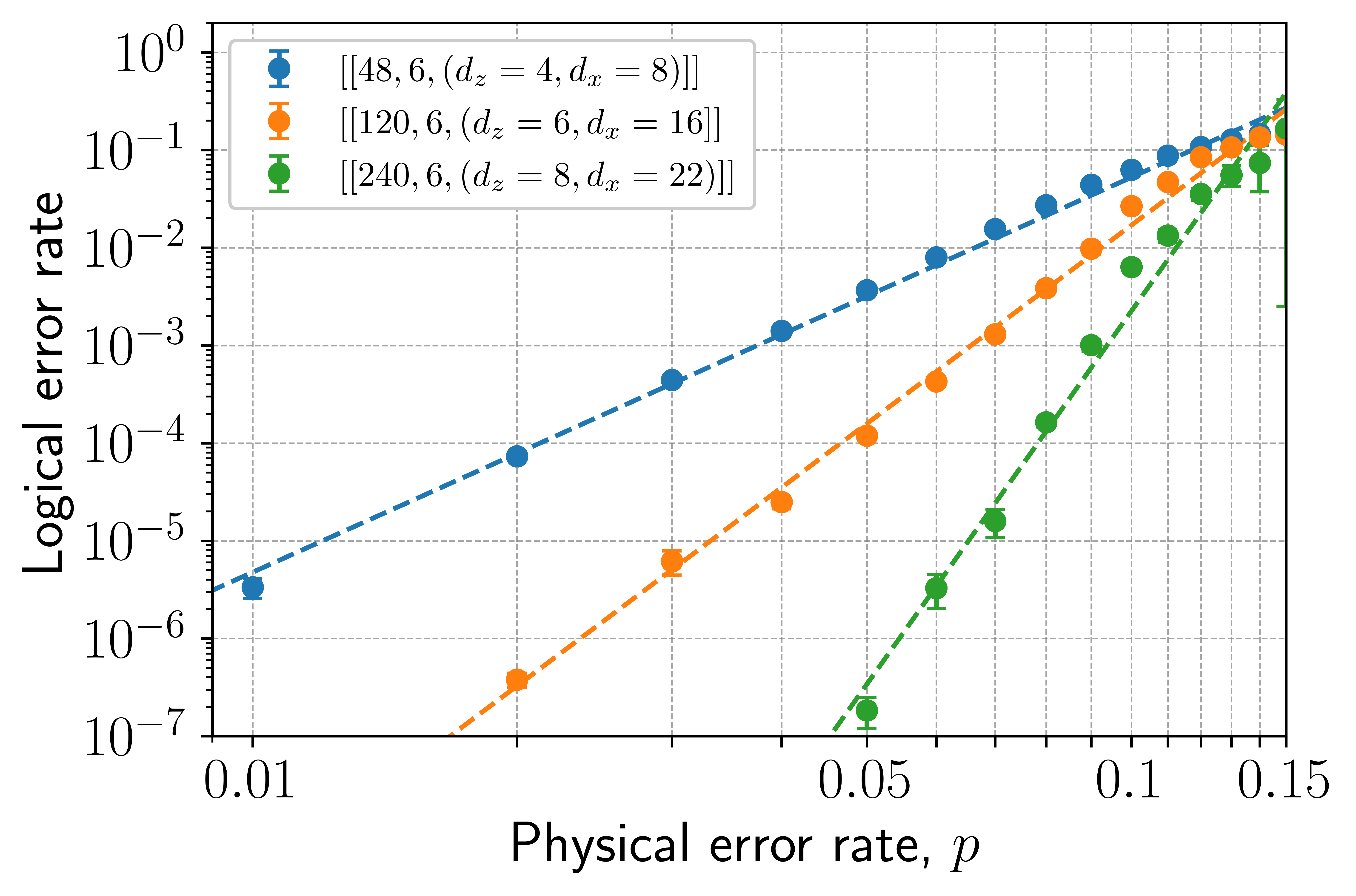}
    \caption{\textbf{Phenomenological noise simulation of single-shot state preparation in the $Z$ basis for $4-2-2$ tricycle codes.} Our method follows Ref.~\cite{hong2024single} and assumes that the initial $Z$ syndrome is trivial. For each code, we simulate one round of syndrome measurement in which measurement errors occur with probability $p$, though we expect performance to improve with a larger decoding window (see \cref{sec::noise_sim}). A most likely error (MLE) decoder applies a minimum weight correction to both the data and measurement qubits. Then, we simulate a noisy transversal $Z$ basis measurement of the data qubits, decode the reconstructed syndrome with the MLE decoder, and apply the corresponding correction. A logical failure is said to occur if the residual $X$ operator is a logical operator of the tricycle code, and the logical error rate is normalized per logical qubit. The observed phenomenological threshold is $\gtrsim 13\%$. Logical error rates are determined via Monte Carlo simulations; error bars indicate standard errors, computed as $\sqrt{p_L(1-p_L)/M}$, where $M$ is the number of samples.}
    \label{fig:state-preparation}
\end{figure}

Going back to the magic state factory, following initialization, the single-shot error correction capabilities of tricycle codes are necessary to prevent errors from spreading. Each $CCZ$ gate will propagate a pre-existing $X$ error into $CZ$ errors on the other two qubits, which are collapsed into nondeterministic patterns of $Z$ errors after syndrome measurement. To prevent such errors from building up, the $Z$ stabilizers must continuously be repaired to $+1$. In codes without single-shot properties, this usually requires `just-in-time decoding' \cite{bombin20182d}, in which error clusters are allowed to expand modestly before being corrected, and accompanies a decrease in performance and threshold of the code \cite{Scruby2022numerical, scruby2025withoutdistillation}. The single shot $Z$ basis of the tricycle code, in contrast, is able to immediately correct $X$ errors (up to a small residual error), limiting repeated propagations of $CZ$ errors. Although the $X$ basis is not single-shot, importantly, the $Z$ errors they detect commute with the $CCZ$ gates, so they do not need be corrected in real time. This is opportune given that the $CCZ$ circuit may disrupt the $X$ stabilizer group until the circuit is completed. As the $CCZ$ circuit is of constant depth, it suffices to measure the $X$ stabilizers only after completion; however, incorporating intermediate correction strategies that leverage stabilizer masking~\cite{mehta2025prep} could be advantageous.

Together, these results demonstrate that our codes enable fault-tolerant preparation of logical $\ket{\overline{+}}$ states in as little as a single round. As a result, magic state factories based on these codes are extremely compact: the initial state $\ket{\overline{+}}_1 \otimes \ket{\overline{+}}_2 \otimes \ket{\overline{+}}_3$ with all $Z$ stabilizers fault-tolerantly fixed to $+1$ can be prepared in constant depth, and the logical $\overline{CCZ}$ operation can also be implemented in constant depth, as described in the previous section, yielding high-fidelity hypergraph magic states in constant depth.

\subsection{Circuit-noise simulations}\label{sec::noise_sim}

We now evaluate the performance of our codes under realistic circuit-level noise models. Our focus is on the $4-2-2$ family of tricycle codes, characterized by modest check weights ($w_x=8$, $w_z=6$) and featuring the shortest-depth $CCZ$ circuits from the symmetric triple cup-product construction (\cref{sec::cup}). Circuit-level noise simulations of certain $4-4-4$ codes are presented in~\cref{sec::numerics-methods} for comparison.

To benchmark performance, we simulate circuit-level noise affecting the entangling gates in the syndrome extraction circuits. We adopt a standard two-qubit depolarizing noise model: for a given two-qubit gate error rate $p_{2q}$, each entangling operation is followed, with equal probability, by one of the fifteen nontrivial two-qubit Pauli errors from ${IX, IY, ..., ZZ}$, and otherwise experiences no error with probability $1-p_{2q}$.

Tricycle codes naturally exhibit higher distance in the $X$ basis compared to the $Z$ basis. Here, we prioritize simulating the more challenging $X$-basis. $Z$-basis results, demonstrating robust single-shot behavior, are included in Appendix G. In the $X$-basis, we use the conventional $d$-round protocol: the syndrome extraction sequence is repeated for $d$ cycles, and decoding is performed using the full record of accumulated syndrome information from all rounds. We utilize the depth-optimal syndrome extraction circuit structure outlined in~\cref{sec::implementation}.

For decoding, we map the syndrome extraction process to a spacetime code, where the checks correspond to parities of stabilizer measurement outcomes across consecutive time steps, and each code qubit is associated with a distinct fault location in the circuit~\cite{xu2024constant, ataides2025constant} (see~\cref{sec::numerics-methods}).

We perform the decoding for the $d$-rounds simulation of the $4-2-2$ codes using a most-likely error (MLE) decoder. The resulting logical error rates, as shown in~\cref{fig:results_new}(b), are defined as the probability of a logical qubit failure, normalized both by the total number of error correction rounds and by the number of logical qubits in the block. We observe a gate-noise threshold of approximately $0.5\%$, and achieve low logical error rates at physically relevant gate errors. For example, we achieve a logical error rate of $\sim 4\times 10^{-5}$ at a physical error rate of $0.1\%$ for the $[[108,6,(d_z=6,d_x=12)]]$ code.

For further performance gains, we consider post-selection, which is standard in the magic-state distillation factory setting \cite{bravyi2012magic}. Specifically, we analyze the $[[48,6,(4,8)]]$ code decoded using BP+LSD~\cite{hillmann2025localized} in conjunction with the clustering post-selection metric introduced in Ref.~\cite{lee2025efficient}. After $d_z=4$ rounds of error correction, the logical error rates as a function of abort rate using the ``$\alpha=2$ clusters llrs'' post-selection metric of Ref.~\cite{lee2025efficient} are depicted in~\cref{fig:results_new}(a). Additionally, we assess full error detection whereby we discard any samples that have any stabilizers flipped. We observe that for the $48$-qubit code, no stabilizer flip in roughly $32\%$ of samples, corresponding to an acceptance fraction of $32\%$ and yielding a very low logical error rate near $6\times 10^{-10}$. For the $84$-qubit code the acceptance fraction for full error detection over $d$ rounds becomes $\sim 8.5\%$  with a logical error rate $< 7\times 10^{-14}$.

In summary, the $4-2-2$ tricycle codes demonstrate strong performance in the $d$-round setting under circuit-level noise. It is worth emphasizing that in practical settings, $d$ rounds may not always be necessary --- e.g., when transversal teleportation can be employed --- so these results may be viewed as lower bounds on achievable QEC performance. We also note that the numerical results presented here and in the appendix pertain to a single code undergoing rounds of error correction and therefore do not directly capture the fidelity of the logical $CCZ$ gate. Nevertheless, since the $CCZ$ circuit has short depth, we expect these results to serve as a reasonable proxy. A detailed simulation of the full $CCZ$ circuit is left to future work.  For additional simulation details and further numerical results, see~\cref{sec::numerics-methods}.

\begin{figure}[ht!]
    \centering
    \includegraphics[width=\textwidth]{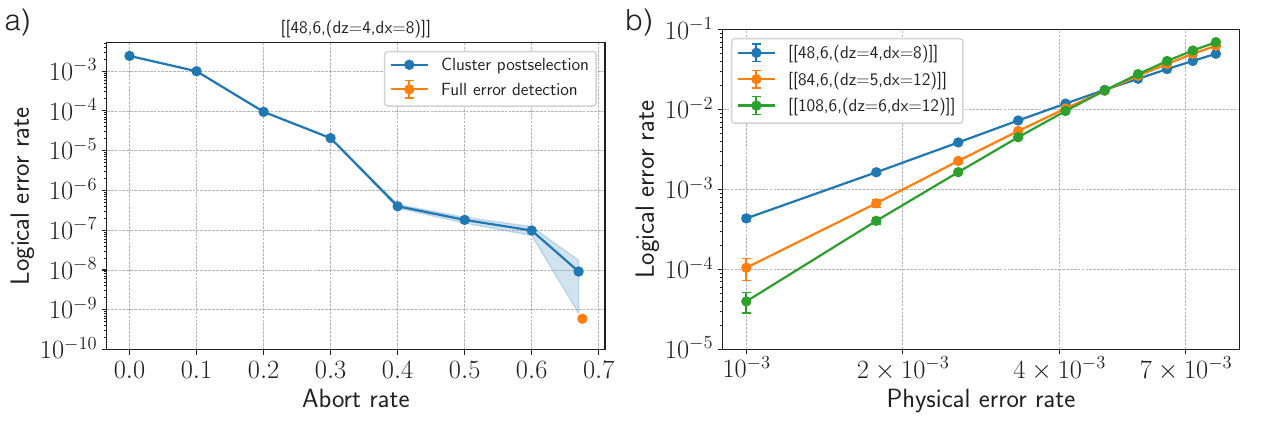}
    \caption{\textbf{Circuit-level noise simulation results for tricycle codes.}
(a) Logical error rate for the $[[48,6,(d_z=4,d_x=8)]]$ code as a function of abort rate under cluster postselection (blue) and full error detection (orange). The cluster postselection data show the tradeoff between logical error rate and postselection (abort) probability using a BP+LSD decoder, while full error detection corresponds to strictly accepting only trials with no detected stabilizer flips. (b) Logical error rate versus two-qubit physical gate error rate $p_{2q}$ for $d$-rounds, fault-tolerant error correction in the $X$ basis, for tricycle codes of increasing size and distance using a most-likely error (MLE) decoder. In both panels, errors are sampled according to a standard two-qubit depolarizing circuit-level noise model and the logical error rate corresponds to the total logical error rate normalized by the number of QEC rounds and by the number of logical qubits. Logical error rates are determined via Monte Carlo simulations, with each data point corresponding to $M$ samples; error bars indicate the standard error as $\sqrt{p_L(1-p_L)/M}$.}
    \label{fig:results_new}
\end{figure}

\begin{figure} [h!]
\centering
\includegraphics{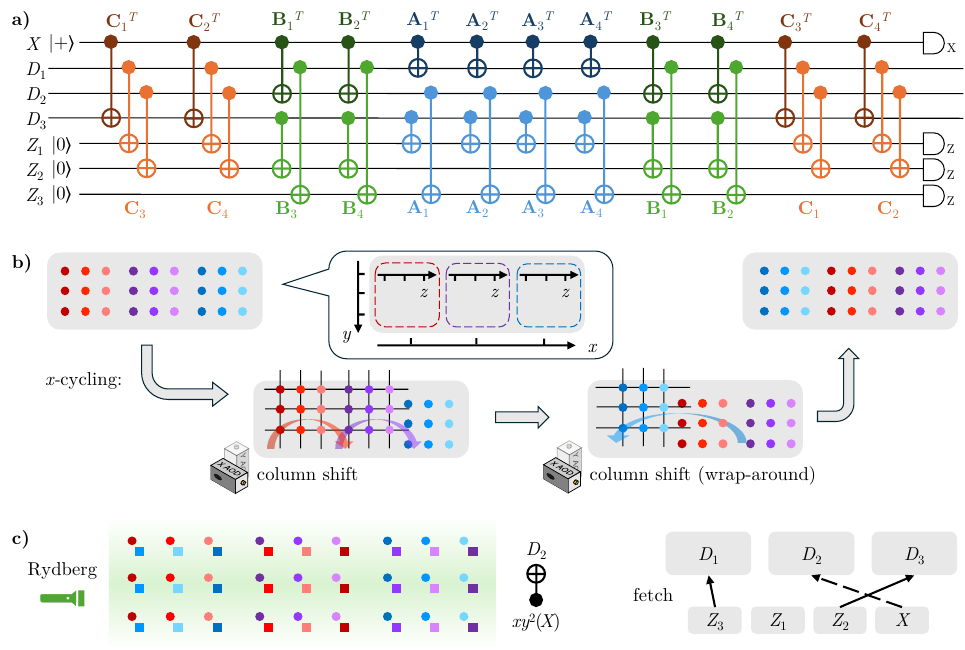}
\caption{\label{fig::implementation}\textbf{Implementation of tricycle codes on neutral atom arrays.}
(a) Syndrome extraction circuit.
Each line denotes a sector. CNOTs are applied on pairs of qubits across two sectors, with the pairing determined by the permutation matrices.
(b) Within each sector, physical qubits sharing the same $x$ index form tiles arranged in a row; within each tile, qubits are ordered by their $y$ and $z$ indices.
Two AOD movements can realize $x$-, $y$-, or $z$-cycling.
(c) Parallel CNOTs can be performed by Rydberg interaction on two overlaying sectors.
Data sectors reside in the entangling zone and have larger spacing to avoid intra-sector interactions.
After permutation in the workspace below, check sectors are fetched to overlay with the corresponding data sectors and perform CNOTs.
Then, check sectors are put back to their original positions, and perform permutation for the next layer.
}
\end{figure}

\subsection{Implementation of syndrome extraction circuits} \label{sec::implementation}

We construct optimal-depth syndrome extraction circuits for the tricycle codes presented.
Denote the data qubit sectors as $D_1$, $D_2$, $D_3$, the Z-check sectors as $Z_1$, $Z_2$, $Z_3$, and the X-check sector as $X$.
Our construction for $4-4-4$ codes is displayed in \cref{fig::implementation}a, where each CNOT symbol represents a parallel group of physical CNOTs between a check sector and a data sector, associated with a specific permutation.
For example, the first CNOT symbol corresponds to CNOTs from the $i$-th X-check to the $j$-th qubit in $D_3$ whenever the $(i,j)$ entry of $\mathbf{C}_1^T$ is 1.
Since $\mathbf{C}_1^T$ is a permutation matrix, there is one and only one $j$ for each $i$.
The circuit achieves the optimal CNOT depth of 12.
The construction consists of five stages represented by colors.
For the $4-4-2$ and $4-2-2$ codes, the overall stage structure is preserved, although some stages contain fewer layers.
The design allows certain degrees of freedom: specifically, how to assign the matrices $\mathbf{A}$, $\mathbf{B}$, and $\mathbf{C}$ to the stages, as well as the order of layers within each stage.
These degrees of freedom lead to 1728 possible circuits for the $4-2-2$ codes. An exhaustive search across these 1728 circuits was performed, from which we select the circuit with the highest noise-suppression for each code in the circuit-noise simulation results of \cref{sec::noise_sim}. For the general construction and a proof of its correctness, see \cref{sec::tricycle_scheduling_method}.

The syndrome extraction circuit can be implemented on either fixed architectures with all required couplings fabricated~\cite{wang2025demonstrationlowoverheadquantumerror}, or on reconfigurable architectures with physical qubit permutation~\cite{bluvstein2022quantum, evered2023high, bluvstein2024logical, sales2025experimental, bluvstein2025architectural, chiu2025continuousoperationcoherent3000qubit, scholl2021quantum, quantinuum23racetrack, ma2023high, gyger24continuous,  bedalov2024faulttolerantoperationmaterialsscience, reichardt2025faulttolerantquantumcomputationneutral}. Here, we focus on reconfigurable neutral atom arrays.

With slight abuse of notation, we label the qubits in each sector by indices $(x, y, z)$, where $1 \le x \le n_x$, $1 \le y \le n_y$, $1 \le z \le n_z$, and $n_G = n_x n_y n_z$.
As shown in the callout of \cref{fig::implementation}b, we place qubits with the same $x$ as tiles arranged in a row; within each tile, qubits are ordered by $y$ and $z$.
Qubits are held in place by traps generated by a spatial-light modulator (SLM);
we assume a sufficient number of SLM traps, and thus do not depict them explicitly.
Crossed 1D acousto-optic deflectors (AODs) create a 2D grid of mobile traps, enabling entire grids of qubits to be picked up and rearranged in the plane.
Our qubit layout leads to efficient implementation of permutation matrices of interest.
For example, \cref{fig::implementation}b illustrates the process of $x$-cycling, corresponding to cyclically shifting the tiles:
first, two tiles are transferred from the SLM to the AOD, shifted right in parallel, and deposited into ancillary SLM traps;
the second step applies a `wrap-around' in the opposite direction for the remaining tile.
Rows and columns in the AOD grid can otherwise move freely, though their order must be preserved, so the two steps for $x$-cycling cannot be combined.
The $y$-cycling and $z$-cycling similarly correspond to cyclically shifting rows and columns within each tile, respectively, as needed to implement polynomial terms like $\mathbf{C}^T_1$.

Entangling gates are performed by exciting adjacent pairs of qubits to Rydberg states.
Parallel CNOTs between corresponding qubits in two sectors are realized by overlaying those sectors and applying a global Rydberg laser illumination, as shown in \cref{fig::implementation}c (dots and squares denote data and check qubits, respectively.
As the colors suggests, the check sector has been permuted with $x$-cycling by one unit and $y$-cycling by two units.)
To prevent unintended interactions between qubits within the same sector, the pairs are separated by a sufficient distances determined by atom species and Rydberg state.

The overall layout (\cref{fig::implementation}c, right) places data sectors in SLM traps within the entangling zone.
For each CNOT layer, check sectors are permuted in a workspace below and fetched to the corresponding data sectors for parallel CNOTs.
In the workspace, smaller grid spacing is used to reduce travel distances for permutation; during fetching, check sectors are expanded to align with data sector spacing.
After entanglement, check sectors return to the workspace for the next permutation.
Due to crossing movement paths, fetching and put-back of Z-check and X-check sectors must be performed separately in this example to preserve AOD column order.

Executing one syndrome cycle requires a constant number of AOD movements.
Let $t_\text{wsp}$ and $t_\text{ent}$ denote the time for AOD to traverse a sector in the workspace and the entangling zone, respectively.
In the implementation above, permutation, fetch, and put-back for all $Z$ sectors can be done in parallel.
Each CNOT layer then involves two fetch and two put-back operations, each bounded by $3t_\text{ent}$ (say, from $X$ to $D_1$).
Permutations for $Z$ and $X$ sectors may require $x$-, $y$-, and $z$-cycling, each up to $2t_\text{wsp}$.
Thus, the per-layer duration is bounded by $12(t_\text{wsp} + t_\text{ent})$.
Further details and possible optimization are discussed in \cref{sec::detail_atom}.

\section{Discussion and Outlook}

We have introduced a class of quantum LDPC codes—\textit{tricycle codes}—that support constant-depth, single-shot preparation of logical magic states and show strong resilience to realistic circuit-level noise with a high threshold of $>0.5\%$. We showed that certain families of tricycle codes allow compact, constant depth preparation of high-magic states known as hypergraph magic states. These hypermagic states in turn embed individually distillable $\overline{CCZ}$-type magic states, enabling universal quantum computation by using tricycle codes as magic state factories. The promising performance of these codes can be improved further along several directions. The asymmetry between the $X$ and $Z$ sectors of tricycle codes make them particularly advantageous for leveraging noise bias in entangling gate operations—a feature pervasive in leading experimental platforms~\cite{aliferis2009fault, bluvstein2025architectural, nigg2014quantum, shulman2012demonstration,bonillaataides2021thexzzx}. Additionally, loss and leakage errors—prevalent in many hardware architectures—represent another avenue for enhancing performance; recent work has shown that appropriate handling of such errors can be leveraged to further improve logical error rates~\cite{baranes2025leveraging,bluvstein2025architectural}. While our present analysis focuses on unbiased depolarizing noise and neglects loss errors, future studies integrating tailored noise bias exploitation and explicit loss-/leakage-leveraging strategies with these codes are likely to yield further gains and represent a compelling direction for subsequent research. Moreover, improved decoding methods would be beneficial as we observe over an order-of-magnitude gap in logical error rates between standard BP+OSD and BP+LSD decoders and an exact but exponentially expensive most-likely-error (MLE) decoder. Encouragingly, recent work has proposed new decoders that enhance performance while maintaining practical inference times~\cite{muller2025improved,bonillaataides2025neural,maan2025decoding,maurer2025realtime}.

Beyond their standalone utility as distillation factories, integrating tricycle codes into broader quantum architectures presents several important challenges. A key open problem is the seamless injection or teleportation of distilled magic states into computational code blocks, such as those based on high-performance bicycle codes. While lattice surgery offers a robust, code-agnostic method for logical state transfer—and has seen significant theoretical and experimental development~\cite{cohen2022low,xu2024constant,williamson2024lowoverhead,ide2024faulttolerant,swaroop2024universal,cross2024linear,cross2024improved,baspin2025fast,zheng2025high,alex2025fast}—a particularly promising alternative is transversal teleportation based on natural isomorphisms between tricycle and bicycle codes. Similar to teleportation protocols between 3D and 2D color codes~\cite{sullivan2024code, daguerre2025experimental}, such schemes using transversal one-way CNOT gates between $3D$ and $2D$ product codes \cite{heussen2024efficient} may enable direct, fault-tolerant transfer of magic states. In particular, we believe this should be possible between bivariate-bicycle codes and bivariate-tricycle codes or between triviariate-bicycle codes and trivariate-tricycle codes, as they share the same product structure. Using extensions of the analytical tools in Ref.~\cite{eberhardt2024logical}, we believe it should be possible to show that the logical operators of pairs of some such $2D$ and $3D$ group-algebra codes split into isomorphic sectors as in hypergraph product codes \cite{tillich2013quantum} -- a natural setting for efficient gate teleportation with transversal one-way CNOTs. A version of this idea has recently been explored in \cite{li2025transversal}. If realized, these protocols could further reduce space-time overhead and enable highly modular architectures with efficiently integrated magic state distillation. We will explore a specific magic-state factory architecture based on high-rate tricycle codes with transversal teleportation onto a target high-rate bicycle code in a follow-up work.

Another important direction is improving the yield of disjoint $\overline{CCZ}$ gates that can be extracted from the hypergraph magic states produced by these circuits (\cref{sec::logical_opt}). This corresponds to the problem of finding the subrank of a binary tensor~\cite{christandl2023gap, kopparty2020geometric, golowich2024quantum, lin2024transversal}, which is computationally hard. While we used mixed-integer programming, new heuristic strategies could uncover larger values of $K_{CCZ}$ for the codes we consider. In parallel, a deeper understanding of the structure of these high-magic hypergraph states—produced by the transversal $CCZ$ circuits described in \cref{sec::transversal_CCZ}—may allow one to compile them directly into useful quantum circuits, offering an alternative to extracting only disjoint $\overline{CCZ}$ gates.

In summary, tricycle codes represent a significant step toward scalable, low-overhead, fault-tolerant magic state distillation with quantum LDPC codes. Ongoing work on new code constructions, transversal gate optimization, and efficient code-switching protocols will continue to strengthen the outlook for practical, high-performance quantum computing.

\section{Acknowledgments}
We thank Qian Xu, Nazli Ugur Koyluoglu, Rohith Sajith, Nishad Maskara, Shayan Majidy, Hengyun Zhou, Harry Putterman,  Brenden Roberts, Jin Ming Koh, Andrei Diaconu, and Jason Cong for valuable discussions. We particularly thank Qian Xu for detailed feedback on our results and manuscript. We acknowledge financial support from IARPA and the Army Research Office, under the Entangled Logical Qubits program (Cooperative Agreement Number W911NF-23-2-0219), the DARPA MeasQuIT program (grant number HR0011-24-9-0359), the Center for Ultracold Atoms (a NSF Physics Frontiers Center, PHY-1734011), the National Science Foundation (grant numbers PHY-2012023 and  CCF-2313084),  the Wellcome Leap Quantum for Bio program, Harvard Quantum Initiative Postdoctoral Fellowship (DBT) and QuEra Computing. After the completion of this project, we became aware of related work studying tricycle codes, albeit with less focus on transversal magic gates~\cite{jacob2025singleshot}. Another related work that appeared around the same time also introduced a general theory for multi-cyclic Abelian group-algebra codes \cite{lin2025abelian}.

\appendix

\section*{Appendix}

\section{Code construction and properties}\label{sec::balanced}

We describe the construction of the tricycle code parity check matrices in more detail, along with general properties of the codes and their connection to balanced products of classical group algebra codes.

Let $G$ be a finite Abelian group, expressed as $G = \mathbb{Z}_{m_1} \times \mathbb{Z}_{m_2} \times \cdots \times \mathbb{Z}_{m_k}$ for some finite $k$. Let $n_G = |G| = \prod_{i=1}^k m_i$. The codes we construct are defined via matrices over the ring $R \coloneqq \FG$, the group algebra of $G$ over $\F$. Elements of $\FG$ are formal sums
\[
\sum_{g \in G} a_g g, \quad \text{where } a_g \in \F,
\]
with componentwise addition and multiplication induced by extending the group operation bilinearly. Since $G$ is Abelian, $\FG$ is a commutative, associative algebra over $\F$ with identity $1_G \in G$.

For $a \in \FG$, we define its \emph{weight} as $|a| = |\{g \in G \mid a_g = 1\}|$. Each $a \in \FG$ determines a binary matrix $\mathbf{A} \in \F^{n_G \times n_G}$ via
\begin{equation} \label{eq::bin_mat_rep}
    \mathbf{A}_{\alpha, \beta} = \mathbb{B}_G(a) \equiv \sum_{g \in G} a_g \delta_{\alpha, g\beta},
\end{equation}
where $\alpha, \beta \in G$ and $\delta$ is the indicator for $\alpha = g\beta$. The map $\mathbb{B}_G : \FG \to \F^{n_G \times n_G}$ gives the regular representation of $a$, following \cite{panteleev2022asymptotically}.

Alternatively, elements of $\FG$ can be identified with polynomials via the isomorphism
\begin{equation}
    \FG \cong \F[x_1, \ldots, x_k] / \langle x_1^{m_1} - 1, \ldots, x_k^{m_k} - 1 \rangle,
\end{equation}
where $x_i$ corresponds to a generator of $\mathbb{Z}_{m_i}$. Under this identification, the weight of $a$ is the number of monomials with nonzero coefficient in the corresponding polynomial.

To compute $\mathbf{A}$, let $S_l$ denote the $l \times l$ binary right cyclic shift matrix, and define $\hat{S}_{m_i} = I_{m_1} \otimes \cdots \otimes S_{m_i} \otimes \cdots \otimes I_{m_k}$. These matrices commute for different $i$. Given $p_a(x_1, \ldots, x_k)$, the polynomial corresponding to $a$, the binary matrix representation is
\begin{equation}
    \mathbf{A} = p_a(\hat{S}_{m_1}, \hat{S}_{m_2}, \ldots, \hat{S}_{m_k}).
\end{equation}
If $|a| = w$, then each row and column of $\mathbf{A}$ has Hamming weight $w$. We can define tricycle codes formally as follows.

\begin{definition}[Tricycle codes]
A three-block Abelian group algebra code (tricycle code) is a CSS code defined by elements $a, b, c \in \FG$ with weights $w_a, w_b, w_c$, and binary matrix representations $\mathbf{A}, \mathbf{B}, \mathbf{C} \in \F^{n_G \times n_G}$. The $X$ and $Z$ parity check matrices are defined as in \cref{eq::parity_check_mats}, with component matrices formed from binarizing group algebra sums: for example, if $a = a_1 + a_2 + a_3 + a_4$, where $a_i \in G$, then $\mathbf{A} = \sum_i \mathbb{B}(a_i)$.
\end{definition}

We focus on codes with $G$ equal to a product of three cyclic groups ($k=3$) as we found these codes to generally have better parameters than those with $k=1,2$ or higher. By analogy to the nomenclature for bicycle codes \cite{kovalev2013quantum, bravyi2024high, panteleev2021degenerate}, these may be referred to as  trivariate-tricycle codes. The $k=1,2$ cases may be referred to as generalized tricycle and bivariate-tricycle respectively. We denote the generators by monomials $x_1 \equiv x$, $x_2 \equiv y$, and $x_3 \equiv z$. The polynomials and groups used in the examples of \cref{table::all_codes}  are listed in \cref{table::code_polys}.

\begin{longtable}{|l|>{\raggedright\arraybackslash}p{1.3cm}|>{\raggedright\arraybackslash}p{3cm}|>{\raggedright\arraybackslash}p{3cm}|>{\raggedright\arraybackslash}p{3cm}|}
\hline
\toprule
$\mathbf{[[N,K,(D_X,D_Z)]]}$ & $\mathbf{(l,m,n)}$ & \textbf{a} & \textbf{b} & \textbf{c} \\
\midrule
$[[48,6,(8,4)]]$ & $(2,2,4)$& $y+z+xz+xyz^2$ & $yz^2+yz^3$ & $y+xyz$\\
\midrule
$[[84,6,(12,5)]]$ & $(2,2,7)$& $y+z+xz+xyz^2$ & $z^3+xz^4$ & $y+yz^4$ \\
\midrule
$[[108,6,(12,6)]]$ & $(3,3,4)$& $x+z^2+yz+x^2yz^3$ & $y^2z+x^2yz^3$ & $x^2+x^2yz^2$ \\
\midrule
$[[240,6,(\leq 22,8)]]$ & $(4,4,5)$ & $xy^2z^3+xy^3z^4+x^2y^2z+x^2y^3z^2$ & $y^3+x^2yz^2$ & $xz^4+x^3y^3z$ \\
\midrule
$[[480,6,(\leq 36,10)]]$ & $(4,5,8)$ & $x^2z^7+x^3y^2+y^4z^4+xy^3z^3$ & $z^6+x^2y^2z^2$ & $x^2+x^3y^4z$ \\

\midrule
$[[108,12,(6,4)]]$ & $(3,3,4)$ & $z+xz^3+xyz^2+x^2y$ & $y^2+y^2z^3+xy^2z+xy^2z^2$ & $z+xyz^3$ \\
\midrule 

$[[180,12,(15,6)]]$ & $(3,4,5)$ & $yz^3+y^3+x^2yz^3+x^2y^3z$ & $xyz^4+xy^2z^2+x^2yz+x^2y^2z^4$ & $z^4+x^2z$ \\
\midrule

$[[108,15,(12,6)]]$ & $(3,3,4)$ & $y+y^2z+xyz^3+x^2y^2z^2$ & $z^2+xy+xy^2z+x^2z^3$ & $yz^3+y^2z+x^2+x^2y^2z^2$ \\
\midrule

$[[270,24,(15,8)]]$ & $(3,5,6)$ & $z^4+y^3+xy^2+x^2yz^4$ & $y^3+y^3z+xy^4+x^2y^2z$ & $yz^4+y^2z+xy^2z+xy^3z^4$ \\
\midrule

$[[324,12,(\leq 32,\leq 12)]]$ & $(3,4,9)$ & $y^2z+xyz+x^2z^6+x^2y^3z^5$ & $z^4+z^5+xyz^7+x^2y^3z^2$ & $yz^7+xz^7+xy^3+x^2y^2$ \\
\midrule
$[[480,15,(\leq 48,\leq 14)]]$ & $(4,5,8)$ & $x^2z^5+x^2yz^4+x^3y^3z^4+x^3y^4z^3$ & $x^2z^3+x^2y^4z^6+x^3z^5+x^3yz^2$ & $yz^5+xz^4+x^2y^4z^4+x^3y^2z^5$ \\

\hline
\caption{Polynomials used to construct trivariate-tricycle codes in \cref{table::all_codes} corresponding to elements of the group algebra of $G=\mathbb{Z}_l\times \mathbb{Z}_m\times\mathbb{Z}_n$. If there are fewer than three cyclic group factors, the corresponding element in the tuple of group orders $(l,m,n)$ is left empty. }
\label{table::code_polys}
\end{longtable}

Next, we describe how tricycle codes can be viewed as three-dimensional hypergraph products (HGPs) \cite{tillich2013quantum} over non-binary rings, and derive a conditional distance lower bound. We will first need the following definitions.

\begin{definition}[Support and intersection subgroups] \label{def::intersection_subgroup}
Let $G$ be a finite Abelian group, and $a = \sum_{g \in G} a_g g \in \FG$ with $a_g \in \{0,1\}$. Define the support subgroup $G_a \coloneqq \langle \{ g : a_g \neq 0 \} \rangle$. For elements $a_1, \ldots, a_k \in \FG$, the intersection subgroup is $N = \bigcap_{i=1}^k G_{a_i}$. Since $G$ is Abelian, $N$ is normal in each $G_{a_i}$.
\end{definition}

\begin{definition}[Connected codes]
Let $a,b,c \in \FG$ with support subgroups $G_a, G_b, G_c$. The code defined by $a,b,c$ is \emph{connected} if $G_a G_b G_c = G$. Otherwise, the code decomposes into subcodes supported on cosets of the subgroup $G_a G_b G_c \subset G$, which has order $|G_a||G_b||G_c|/|N|^2$, where $N$ is the intersection subgroup.
\end{definition}

This notion of connectivity extends the Abelian case of \cite{lin2024quantum}, which showed that connected two-block group algebra codes can be written as hypergraph products (HGP) over non-binary rings. The same argument applies here, yielding the following:

\begin{prop}[Non-binary 3D HGP]
Let $a,b,c \in \FG$ define a connected tricycle code, and let $N = G_a \cap G_b \cap G_c$ with $|N| = c$. Let $l_a = [G_a : N]$ be the index of $G_a$ in $N$, and define $l_b$, $l_c$ analogously. Define $R = \mathbb{F}_2[N]$, and consider $R$-valued matrices:
\begin{align}
A &= A_1 \otimes I_{l_b} \otimes I_{l_c}, \\
B &= I_{l_a} \otimes B_1 \otimes I_{l_c}, \\
C &= I_{l_a} \otimes I_{l_b} \otimes C_1,
\end{align}
for appropriately chosen $A_1 \in R^{l_a \times l_a}$, $B_1 \in R^{l_b \times l_b}$, and $C_1 \in R^{l_c \times l_c}$. The binary matrices $\mathbf{A}, \mathbf{B}, \mathbf{C}$ are then obtained by applying $\mathbb{B}_N$ element-wise: $\mathbf{A}_{i,j} = \mathbb{B}_N(A_{i,j}) \in \mathbb{F}_2^{c \times c}$, etc.

The parity check matrices correspond to a three-fold $R$-linear homological product code \cite{zeng2019higher}:
\begin{align}\label{eq:hgp3nonbinary}
H_X^T &= \begin{bmatrix}
    A_1 \otimes I_{l_b} \otimes I_{l_c} \\
    I_{l_a} \otimes B_1 \otimes I_{l_c} \\
    I_{l_a} \otimes I_{l_b} \otimes C_1
\end{bmatrix}, \\
H_Z &= \begin{bmatrix}
    I_{l_a} \otimes I_{l_b} \otimes C_1 & \mathbf{0} & A_1 \otimes I_{l_b} \otimes I_{l_c} \\
    \mathbf{0} & I_{l_a} \otimes I_{l_b} \otimes C_1 & I_{l_a} \otimes B_1 \otimes I_{l_c} \\
    I_{l_a} \otimes B_1 \otimes I_{l_c} & A_1 \otimes I_{l_b} \otimes I_{l_c} & \mathbf{0}
\end{bmatrix}.
\end{align}
The binary versions are recovered via the binarization map $\mathbb{B}_N$ defined in \cref{eq::bin_mat_rep}.
\end{prop}

\begin{proof}
    This follows identically from the proofs of Statements 7,8,9 in the appendix of \cite{lin2024quantum}. In particular, by choosing a representative set of coset elements $p_i \in G_a/N$, $q_j \in G_b/N$, $r_k \in G_c/N$ where $i=1\cdots l_a$, $j=1\cdots l_b$, $k=1\cdots l_c$, the $R$-valued matrices $A_1$, $B_1$, $C_1$ are given by 
    \begin{align}
        (A_1)_{i_1,i_2} = \sum_{n \in N} a_{(p_{i_1}n p_{i_2}^{-1})}  n  \\ 
        (B_1)_{j_1,j_2} = \sum_{n \in N} b_{(q_{j_1}n q_{j_2}^{-1})}  n  \\
        (C_1)_{k_1,k_2} = \sum_{n \in N} c_{(r_{k_1}n r_{k_2}^{-1})}  n 
    \end{align}

    where $a_g$, $b_g$, $c_g$ are the coefficients in the group algebra elements $a=\sum_{g\in G}a_g g$ and similarly for $b,c$.
\end{proof}

This perspective of viewing connected tricycle codes as three-dimensional homological products over a ring is useful for proving the following lower bound on the code distances. 

\begin{theorem}\label{thm::dist_lower_bnd}
    Consider a connected three-block Abelian group algebra code constructed from group algebra elements $a,b,c\in \mathbb{F}_2[G]$, with corresponding binary matrices $\mathbf{A},\mathbf{B},\mathbf{C} \in \mathbb{F}_2^{n\times n}$ where $n=|G|$. Further suppose that the classical binary linear codes with parity check matrices $\mathbf{A^T},\mathbf{B^T},\mathbf{C^T}$ have minimum distances $d_A^T, d_B^T, d_C^T$. We use the convention that $d = \infty$ if the corresponding matrix is full rank. Let $N=G_a \cap G_b \cap G_c$ be the intersection subgroup as defined in \cref{def::intersection_subgroup}. The minimum distance of the quantum code is bounded below as
    \begin{align}\label{eq:dist_lb}
        d \geq \frac{1}{|N|} \min\{d_A^T, d_B^T, d_C^T\}
    \end{align}
    
\end{theorem}
\begin{proof}
The proof that $d_Z$ satisfies this inequality follows from the proof of Statement 10 in \cite{lin2024quantum}. In particular, Lemma 3 of \cite{lin2024quantum} also holds for three-dimensional homological products by decomposing them into a direct sum of three-dimensional homological products over cyclic group algebras (analogous to Appendix B. in \cite{panteleev2021quantum}). The matrices $\mathbf{A},\mathbf{B},\mathbf{C}$ are constructed identically, and following the proof of Lemma 3. in \cite{lin2024quantum}, each component code over a cyclic group algebra can be shown to have all-unit Smith Normal Forms under the assumptions of the lemma. The rest of the proof is identical, applied instead to the matrix $H_X$ defined in \cref{eq::parity_check_mats}.

Then, in \cref{thm::dzdx}, we prove that $d_Z \leq d_X$, making \cref{eq:dist_lb} a lower bound for the overall code distance. 
\end{proof}

For disconnected codes formed from $m$ connected subcodes, the bound applies to each component and can be summed. In practice, this lower bound is often loose: the codes we construct typically exceed it by a large margin. It may be possible to improve this bound using techniques from \cite{zeng2019higher} adapted to non-binary rings. 

While the above bound may not be tight, it does imply the existence of a threshold for tricycle codes. This is because classical quasi-Abelian group-algebra codes with parity check matrices of the form $\mathbf{A^T}$, $\mathbf{B^T}$, $\mathbf{C^T}$ are known to include codes that have distances that grow linearly in the block-length \cite{lin2024quantum}. Combined with the above bound, this implies the existence of tricycle codes with distances that are extensive in the block-length with at-least $\sqrt[3]{N}$ scaling, implying the existence of a threshold. 

In the following theorem, we show that tricycle codes satisfy $d_Z \leq d_X$s. This would make \cref{thm::dist_lower_bnd} a lower bound on both the $X$ and $Z$ distances. 

\begin{theorem}
    \label{thm::dzdx}
All connected three-block Abelian group algebra codes defined by \cref{eq::parity_check_mats} satisfy $d_Z \leq d_X$.
\end{theorem}

\begin{proof}
        We will assume that none of $A_1$, $B_1$, $C_1$ are full rank, else the tricycle code is trivial with $k = 0$ and $d = \infty$ by convention. Noting $R \cong \bigoplus_{g \in N} \F$, the map $\mathbbm b_N: R \rightarrow \mathbb F_2{|N|}$ defined by $\mathbbm b_N: a = \sum_g a_g g \mapsto (a_{g_1} \ \cdots \ a_{g_|N|})$ is a vector space isomorphism. With action defined element-wise, it can be extended to $\mathbbm b_N^m: R^m \rightarrow \mathbb F_2^m$; we omit the superscript hereafter. We define the norm $|\cdot|_R$ on $\mathbb R^m$ via $|v|_R = \sum_{i = 1}^{m} |\{g \in G: (v_i)_g = 1\}|$ for any $v \in R^m$. Importantly, $|v|_R = |\mathbbm b_N(v)|$ in general. Let $\mathbb B_N: R \rightarrow \mathbb F_2^{|N| \times |N|}$ be the regular representation of $N$ introduced in \cref{eq::bin_mat_rep}, extended element-wise and linearly to $R^{m \times m}$. For any $A \in R^{m \times m}$ and $v \in R^m$, $\mathbbm b_N(Av) = \mathbb B_N(A) \mathbbm b_N(v)$. 

    Let $x = (x_1 \ x_2 \ x_3) \in R^{l_a l_b l_c}$ denote a minimum weight logical $X$ operator. Define 
        \begin{equation}\label{eq::BB code}
        (\delta^{-1})^T := \begin{pmatrix}
            A_1^T \otimes_R I_{l_b} & I_{l_a} \otimes_R B_1^T
        \end{pmatrix} \hspace{20pt} \delta^0 := \begin{pmatrix}
            I_{l_a} \otimes_R B_1 & A_1 \otimes_R I_{l_b}
        \end{pmatrix}.
    \end{equation}
    The CSS code with $R$-valued $X/Z$ parity check matrices $(\delta^{-1})^T/\delta^0$ is a bicycle code. $H_Zx = 0$ implies $(x_1 \ x_2) \in \ker(\delta^0 \otimes_R I_{l_c}) = (\ker \delta^0) \otimes_{R} R^{l_c} \cong (\ker \delta^0)^{\oplus l_c}$. Thus, one can express $(x_1 \ x_2) = \bigoplus_{i = 1}^{l_c} \tilde x_i$, $\tilde x_i \in \ker \delta^0$. Suppose that for at least one $i^* \in [1, l_c]$, $\tilde x_{i^*} \notin \im \delta^{-1}$. Then, $\tilde x_{i^*}$ is a nontrivial logical $X$ operator of this bicycle code; there exists some $\tilde z \in \ker (\delta^{-1})^T \setminus \im \delta^{0T}$ with $|\mathbbm b_N(\tilde z)| \leq |\mathbbm b_N(\tilde x)|$ (Lemma 1 of Ref. \cite{bravyi2024high}). We will lift $\bar z$ to an equal-weight logical $Z$ operator of the tricycle code.

    Expressing $H_X/H_Z$ in terms of $(\delta^{-1})^T/\delta^0$,
    \begin{equation}
        H_X = \begin{pmatrix}
            (\delta^{-1})^T \otimes_R I_{l_c} &  I_{l_a l_b} \otimes_R C_1^T
        \end{pmatrix} \hspace{20pt} H_Z = \begin{pmatrix}
            \delta^0 \otimes_R I_{l_c} & 0 \\
            I_{2 l_a l_b} \otimes_R C_1 & \delta^{-1} \otimes_R I_{l_c}
        \end{pmatrix}.
    \end{equation}
    From the Künneth theorem for $R$-module cohomology (see section \ref{app:products} of the supplementary material and \cite{loormobius}),
\begin{equation}\label{eq: Z homology}
\frac{\ker H_X}{\im H_Z^T} \cong \left(\frac{R^{l_a}}{\im (\delta^{-1})^T} \otimes_R \ker C_1^T
\right) \oplus \left(\frac{\ker (\delta^{-1})^T}{\im \delta^{0T}} \otimes_R \frac{R^{l_c}}{\im C_1^T}\right) .
\end{equation}
Since $R^{l_c} = \im C_1^T + R^{l_c}/\im C_1^T$ is spanned by a basis of unit vectors, and $\im C_1^T \subsetneq R^{l_c}$ by assumption, there exists some unit vector $\hat e$, $|\hat e|_R = 1$, such that $z = (0 \ \tilde z \otimes_R \hat e)$ is a nontrivial logical $Z$ operator. For any $g \in G, \ r \in R$, $|r|_R = |gr|_R$. Therefore, assuming $\tilde x_{i^*} \in \ker \delta^0 \setminus \im \delta^{-1}$ exists, $d_Z \leq |\mathbbm b_N(z)| = |z|_R = |\tilde z|_R \leq |\tilde x|_R \leq |x|_R = \mathbbm b_N(x) = d_X$. 

Repeating the analysis for the other two ways to express the tricycle code as a $R$-homological product of a bicycle code and a classical group-algebra code, the above analysis is only incomplete in the case
\begin{equation} \label{eq: Bicycle code stabilizers}
    x \in \im \begin{pmatrix}
        T_A & 0 \\
        0 & I_{l_a l_b l_c} \\
        T_C & 0
    \end{pmatrix} \cap \im \begin{pmatrix}
        T_A & 0 \\
        T_B & 0\\
        0 & I_{l_a l_b l_c} 
    \end{pmatrix} \cap \im \begin{pmatrix}
        0 & I_{l_a l_b l_c} \\
        T_B & 0 \\
        T_C & 0
    \end{pmatrix},
\end{equation}
where $T_A = A_1 \otimes_R I_{l_b} \otimes_R I_{l_c}$, $T_B = I_{l_a} \otimes_R B_1 \otimes_R I_{l_c}$, and $T_C = I_{l_a} \otimes_R I_{l_b} \otimes_R C_1$. Such an $x$ exists iff there exist $u_1, u_2, u_3 \in R^{2 l_a l_b l_c}$ such that $x_1 = T_A u_1  = T_A u_2$, $x_2 = T_B u_2 = T_B u_3$, and $x_3 = T_C u_1 = T_C u_3$, yet
\begin{equation}\label{eq: Nontrivial on tricycle code}
    x \notin \im H_X^T = \im \begin{pmatrix}
        T_A \\ T_B \\ T_C
    \end{pmatrix}.
\end{equation}
Let $v_A = u_1 + u_2$, $v_B = u_2 + u_3$, so $u_1 + u_3 = v_A + v_B$. It follows that $v_A \in \ker T_A, v_B \in \ker T_B, v_A + v_B \in \ker T_C$. Expressing
\begin{equation}
x = \begin{pmatrix}
    T_A u_1 \\
    T_B u_2 \\
    T_C u_3
\end{pmatrix} = \begin{pmatrix}
    T_A u_1 \\
    T_B(u_1 + v_A) \\
    T_C(u_1 + v_A + v_B)
\end{pmatrix} = \begin{pmatrix}
    T_A u_1 \\
    T_B(u_1 + v_A) \\
    T_Cu_1 
\end{pmatrix} = H_X^Tu_1 + \begin{pmatrix}
    0 \\ T_B v_A \\
    0
\end{pmatrix},
\end{equation}
a solution for $x$ exists iff $(0 \ T_B v_A \ 0)$ is a nontrivial logical $X$ operator, which requires $v_A \notin \ker T_B$. Note $v_A \in \ker T_A \cap (\ker T_B+ \ker T_C)$. Now we use the following lemma.

\begin{lemma}
    In this setting, the intersection can be distributed over subspace addition.
        \begin{equation}
        \ker T_A \cap (\ker T_B + \ker T_C) = (\ker T_A \cap \ker T_B) + (\ker T_A \cap \ker T_C)
    \end{equation}
\end{lemma}
\begin{proof}
    The reverse inclusion is generally true. For the forward direction, define the canonical projection $\Pi: R^{l_a l_b l_c} \rightarrow \ker T_A$. Any $v \in \ker T_A \cap (\ker T_B + \ker T_C)$ satisfies $\Pi v = v$ and can be written as $v= t_B + t_C$, where $t_B \in \ker T_B$ and $t_c \in \ker T_C$. Since $\Pi$ factorizes as $\Pi=  \pi \otimes_R I_{l_b} \otimes_R I_{l_c}$, where $\pi: R^{l_a} \rightarrow \ker A_1$, $\Pi t_B \in \ker T_B$ and $\Pi t_C \in \ker T_C$. By definition, $\Pi t_B, \Pi t_C \in \ker T_A$.
\end{proof}
Thus, $v_A \in (\ker T_A \cap \ker T_B) + (\ker T_A \cap \ker T_C) \subseteq \ker T_B + (\ker T_A \cap \ker T_C)$. If we write $v_A \in t_B + t_{AC}$, where $t_B \in \ker T_B$, $t_{AC} \in \ker T_A \cap \ker T_C$, then, observe that $(0 \ T_B v_A \ 0) = H_X^T t_{AC}$ and is thus logically trivial. This contradicts \cref{eq: Nontrivial on tricycle code}, so \cref{eq: Bicycle code stabilizers} cannot hold. \end{proof}

\begin{corollary}\label{cor: inequality for disconnected codes}
    \cref{thm::dzdx} extends to disconnected codes.
\end{corollary}
\begin{proof}
    Any disconnected tricycle code is a direct sum of connected tricycle codes. \cref{thm::dzdx} can then be applied to the component with the lowest weight logical $X$ operator. 
\end{proof}

Next, we explain how tricycle codes arise as the \textit{balanced product} of classical group-algebra codes. This construction builds on the homological formulation of CSS codes and the balanced product construction from Refs.~\cite{breuckmann2021balanced, panteleev2022asymptotically, panteleev2021quantum, breuckmann2024cups}, which we assume readers are familiar with in the following discussion -- a review is provided in the supplementary material.

A classical group-algebra code is defined by an element $a \in \FG$, with parity check matrix $H = \mathbb{B}_G(a)$, where $\mathbb{B}_G$ is the binarization map from \cref{eq::bin_mat_rep}. This corresponds to the 2-term cochain complex
\begin{equation}\label{eq:classical_complex}
\begin{tikzcd}
\FG \arrow[r,"a"] & \FG
\end{tikzcd}
\end{equation}
where the coboundary map is given by multiplication by $a$ in the algebra. The bits of the classical code are associated with the $1$-cochains (right) and the checks are associated with the $0$-cochains (left).

Given two such complexes defined by elements $a, b \in \FG$, their balanced product results in the 3-term complex
\begin{equation}\label{eq:bicycle_complex}
\begin{tikzcd}
R \arrow[r,"\spmat{a \\ b}"] & R^2 \arrow[r,"\spmat{b & a}"] & R
\end{tikzcd}
\end{equation}
where $R = \FG$ for notational convenience. Binarizing this complex yields the Abelian bicycle code, with $H_X^T = \begin{bmatrix} \mathbb{B}(a)^T & \mathbb{B}(b)^T \end{bmatrix}$ and $H_Z = \begin{bmatrix} \mathbb{B}(b) & \mathbb{B}(a) \end{bmatrix}$ \cite{lin2024quantum, kovalev2013quantum}, where the $0$-cochains are $X$-checks, $1$-cochains are qubits, and $2$-cochains are Z-checks (with 0,1,2 annotating the spaces in the complex from left to right). This complex arises naturally from the balanced product using the isomorphism $R \otimes_R R \cong R$; see \cite{eberhardt2024logical, breuckmann2021balanced, panteleev2022asymptotically} and the supplementary material.

Tricycle codes generalize this by iterating the balanced product once again with a third classical group-algebra code defined by $c \in \FG$, giving the 4-term cochain complex:
\begin{equation}\label{eq:tricycle_complex}
\begin{tikzcd}
\FG \arrow[r,"\spmat{a \\ b \\ c}"] & \FG^3 \arrow[r,"\spmat{c & 0 & a \\ 0 & c & b \\ b & a & 0}"] & \FG^3 \arrow[r,"\spmat{b & a & c}"] & \FG
\end{tikzcd}
\end{equation}
This again follows from expanding the terms of the balanced product complex and applying $R \otimes_R R \cong R$.

The quantum CSS code associated with the tricycle construction is defined by the first three terms of \cref{eq:tricycle_complex}, with consecutive spaces representing $X$-checks, qubits, and $Z$-checks respectively. Binarizing the coboundary maps yields the parity check matrices in \cref{eq::parity_check_mats}, and the final map defines the metachecks in \cref{eq:metachecks}. This complex appears in the appendix of \cite{eberhardt2024logical}, and tricycle codes of this form are also briefly noted in \cite{breuckmann2024cups}.

This homological viewpoint is especially useful for constructing transversal logical gates. In the next section, we apply the formalism of \cite{breuckmann2024cups} to equip the complex \cref{eq:tricycle_complex} with an algebraic structure that enables transversal $CCZ$ circuits.

\section{Single-shot guarantees}\label{sec:soundness}
We prove that tricycle codes possess soundness properties in one basis, which in turn guarantee some fault-tolerant single-shot state preparation and error correction capabilities under adversarial noise. Our proofs heavily reference Campbell's results on the soundness of 4D homological product codes in Ref. \cite{campbell2019theory}.

\begin{definition} \textnormal{(Soundness)}:
Let $t \in \mathbb{Z}_+$ and $f: \mathbb{Z} \rightarrow \mathbb R$ be a function with $f(0) = 0$. For some ring $R$ and norm $|\cdot|$ on $R^n$, we say that a parity check matrix $H \in R^{m \times n}$ is a $(t,f)$-sound map if, for any $v \in R^n$ such that $|Hv| < t$, there exists some $v' \in R^n$ such that $Hv = Hv'$ and $|v|' \leq f(|Hv|)$.
\end{definition}

To begin, we consider homological product codes formed from classical seed codes over $\mathbb F_{2^k}$ for some $k > 1$. Let $|\cdot|$ denote the Hamming weight: for $v \in \mathbb{F}_{2^k}^n$, $|v| = |\{i: v_i \neq 0\}|$. A similar Hamming weight applies over $R$. For a classical code over $\mathbb F_{2^k}$-vector spaces, we define its distance in the usual way with respect to the Hamming weight.

\begin{lemma} \label{lem: First soundness} \textnormal{(Extension to the first soundness lemma, \cite{campbell2019theory}):} Let $\mathcal{C} = C^0 \xrightarrow{\delta^0} C^1$ be a cochain complex over $\mathbb F_{2^k}$-vector spaces for some $k > 1$ corresponding to a code over $l$ dits. Denote the Hamming distance of $\mathcal C$ by $d_H(\delta^0) = d_0$ and the transpose code by $d_0^T$. Let $\tilde{\mathcal C} = \mathcal{C} \otimes_{\mathbb F_{2^k}}\mathcal{C}$, yielding the $3$-term complex $\tilde{\mathcal C} = \tilde C_{-1} \xrightarrow{\tilde \delta^{-1}} \tilde C^0 \xrightarrow{\tilde \delta^0} \tilde C^1$, where 
    \begin{equation}
        \tilde\delta^{0T} = \begin{pmatrix}
            \delta^{0T} \otimes_{\mathbb F_{2^k}} I \\
            I \otimes_{\mathbb F_{2^k}} \delta^0
        \end{pmatrix}, \qquad\ \tilde \delta^{-1} = \begin{pmatrix}
            I \otimes_{\mathbb F_{2^k}} \delta^{0T} \\
            \delta^0 \otimes_{\mathbb F_{2^k}} I
        \end{pmatrix}.
    \end{equation}
    Then, $\tilde \delta^{0T}, \ \tilde \delta^{-1}$ are $(t, f)$-sound, where $t = \min\{d_0, d_0^T\}$ and $f(x) = x^2/4$. 
\end{lemma}
    We provide a self-contained proof of this lemma, which mostly mirrors Appendix D in Ref. \cite{campbell2019theory}, for illustration. In later adaptations, we only discuss necessary modifications to results from Ref. \cite{campbell2019theory}.
\begin{proof}
We focus on $\tilde \delta^{0T}$; the proof sketch for $\tilde \delta^{-1}$ is similar. Let $s \in \im \tilde \delta^{0T}$ with $|s| < t$. $v \in \mathbb F_{2^k}^n$, $s = \tilde \delta^{0T} v$. Letting $s = (s_L \ s_R)$ and reshaping $\mathbb{F}_{2^k}^l \otimes_{\mathbb F_{2^k}} \mathbb F_{2^k}^l \rightarrow \mathbb F_{2^k}^{l \times l}$ yields the equations $S_L = \delta^{0T} V$ and $S_R = V \delta^{0T}$. For any $A \in \mathbb F_{2^k}^{l \times l}$, let $\textnormal{rsp}(A) \subseteq \{1, \dots, n\}$ denote the set of indices corresponding to nonzero rows of $A$, and likewise with $\textnormal{csp}(A)$ for columns.

\begin{lemma}\label{lem: Reduction}
    There exists some $w \in \mathbb  F_{2^k}^l$ such that $s = \tilde \delta^{0T} w$, $|\textnormal{csp}(S_L)| = |\textnormal{csp}(W)|$, and $|\textnormal{rsp}(S_R)| = |\textnormal{csp}(W)|$. 
\end{lemma}
\begin{proof}
    At least one $w$ satisfying the syndrome equation must exist. For any such $w$, $\textnormal{csp}(S_L) \subseteq \textnormal{csp}(W)$ and $\textnormal{rsp}(S_R) \subseteq \textnormal{rsp}(W)$, $|\textnormal{csp}(S_L)| \leq |\textnormal{csp}(W)|$ and $|\textnormal{rsp}(S_R)| \leq |\textnormal{rsp}(W)|$. The $\geq$ direction is more difficult to prove. For concreteness, assume $|\textnormal{csp}(S_L)| < |\textnormal{csp}(V)|$ and $|\textnormal{rsp}(S_R)| < |\textnormal{rsp}(V)|$. Then, there exist column and row vectors $a$ and $b^T$ in $V$ such that $a \in \ker \delta^{0T}$ and $b \in \ker \delta^0$. $b^T$ is the $i$th row of $V$, which we denote by $V_{i, \cdot} = b^T$, and $a$ is the $j$th column: $V_{\cdot, j} = a$. If $V_{ij} \neq 0$, then $V' = V + V_{ij}^{-1} ab^T$ has strictly smaller support size: $V'_{\cdot, j} = V_{\cdot, j} + V_{ij}^{-1} (a b^T)_{\cdot, j} = 0$ since $(a b^T)_{\cdot, j} = b^T_j a = V_{ij} a$. Identically, $V'_{i, \cdot} = 0$, and no new rows or columns can pick up support. Thus, we update $V \gets V'$ and repeat until the algorithm stalls, which is when there are no rows and columns $b^T$ and $a$ such that $V_{ij} \neq 0$. Up to row and column permutations, this leaves $V$ in the form
    \begin{equation}
        V = \begin{pmatrix}
            D & B & 0 \\
            A & C & 0  \\
            0 & 0 & 0
        \end{pmatrix}
    \end{equation}
    where columns in the left column block contain $a \in \ker \delta^{0T}$, and those in the middle column block satisfy $\delta^{0T}a \neq 0$. Similarly, rows $b^T$ in the top row block satisfy $b \in \ker \delta^0$, whereas for those in the middle block, $\delta^0b \neq 0$. Since no row-column pairs intersect, $D=0$. If $A$ is nonzero, any nonzero column $a \in \ker \delta^{0T}$ satisfies $|a|\geq d_0^T$, $|\text{rsp}(S_R)| \geq |c| \geq d_0^T$, which is a contradiction. Thus, $A = 0$, and similarly, $B = 0$.
\end{proof}

Since $|s|  = |S_L| + |S_R| \geq |\text{csp}(S_L)| + |\text{rsp}(S_R)|$, the general identity $(a + b)^2/4 \geq ab$ for integers $a, b$ yields
\begin{equation}
    |s|^2/4 \geq |\text{csp}(S_L)||\text{rsp}(S_R)| = |\text{csp}(V)||\text{rsp}(V)| \geq |V|.
\end{equation}
\end{proof}

Similarly, we have the following result.
\begin{lemma} \label{lem: Second soundness} \textnormal{(Extension to the second soundness lemma, \cite{campbell2019theory})}:
    Let $\tilde C^{-1} \xrightarrow{\tilde \delta^{-1}} \tilde C^0 \xrightarrow{\tilde \delta^0} \tilde C^1$ be a cochain complex over $\mathbb{F}_{2^{k_i}}$-vector spaces where $\tilde \delta^{0T}$, $\tilde \delta^{-1}$ are $(t, f)$-sound with $f(x) = x^2/4$. Then, $\breve{\mathcal C} = \tilde{\mathcal C} \otimes_{\mathbb F_{2^k}} \tilde{\mathcal C}$ is the 4-complex 
    \begin{equation} \label{eq: 4-term} \breve{\mathcal C} =\breve{C}^{-2} \xrightarrow{\breve{\delta}^{-2}} \breve{C}^{-1} \xrightarrow{\breve{\delta}^{-1}} \breve{C}^{0} \xrightarrow{\breve{\delta}^{0}} \breve{C}^{1} \xrightarrow{\breve{\delta}^{1}} \breve{C}^{2}
    \end{equation}
    where 
\vspace{-11pt}
\begin{gather}
\check{\delta}^{-2} = \begin{pmatrix}
I \otimes_{\mathbb F_{2^k}} \tilde{\delta}_0^T \\
\tilde{\delta}_{-1} \otimes_{\mathbb F_{2^k}} I
\end{pmatrix}, \qquad
\check{\delta}^{-1} = \begin{pmatrix}
I \otimes_{\mathbb F_{2^k}} \tilde{\delta}_{-1}^T & 0 \\
\tilde{\delta}_{-1} \otimes_{\mathbb F_{2^k}} I & I \otimes_{\mathbb F_{2^k}} \tilde{\delta}_0^T \\
0 & \tilde{\delta}_0 \otimes_{\mathbb F_{2^k}} I
\end{pmatrix}, \\
\check{\delta}^{0T} = \begin{pmatrix}
\tilde{\delta}_{-1}^T \otimes_{\mathbb F_{2^k}} I & 0 \\
I \otimes_{\mathbb F_{2^k}} \tilde{\delta}_{-1} & \tilde{\delta}_0^T \otimes_{\mathbb F_{2^k}} I \\
0 & I \otimes_{\mathbb F_{2^k}} \tilde{\delta}_0
\end{pmatrix}, \qquad
\check{\delta}^{1T} = \begin{pmatrix}
\tilde{\delta}_0^T \otimes_{\mathbb F_{2^k}} I \\
I \otimes_{\mathbb F_{2^k}} \tilde{\delta}_{-1}
\end{pmatrix}.
\end{gather}
The maps $\breve \delta^0$ and $(\breve \delta^{-1})^T$ are $(t, g)$-sound, with $g(x) = x^3/4$. 
\end{lemma}
\begin{proof}
    For $k=1$, the proof of this result is Lemma 6 in \cite{campbell2019theory}, which in turn relies on ``partial soundness'' and ``inheritance of soundness'' results, which are Lemma 7 and Claim 3 there, respectively. For $k > 2$, the proofs of all three statements can be adapted with a small modification, just like \cref{lem: Reduction}: during the reduction steps $V \gets V + ab^T$ is modified into $V \gets V + V_{ij}^{-1} ab^T$. The $\mathbb F_2$ case is simply the special case where $V_{ij}^{-1} = 1$ given $a$ and $b^T$ intersect. Still working within the Hamming norm, the remainder of the analysis extends. 
\end{proof}

We can invoke these to demonstrate soundness properties for some connected 4D balanced product codes. 
\begin{theorem}\label{thm: 4D RHPC Soundness}
    Let $\breve{\mathcal C}$ denote a 4-term cochain complex \cref{eq: 4-term} over $R$-modules, built from the classical seed code $\mathcal C = C^0 \xrightarrow{\delta^0} C^1$ by repeated applications of the $R$-valued homological product, where $R = \mathbb F_2[N]$. Assume $|N|$ is odd. Let $d_0$ denote the distance of the binarized classical code $\mathbbm b_N(C^0) \xrightarrow{\mathbb{B}_N(\delta^0)} \mathbbm b_N(C^1)$. Then, $\breve C$ has $(t, f)$-soundness, where $t = d_0/|N|$, and $g(x) = \frac{|N|+1}{8}x^3$
\end{theorem}
\begin{proof}
Since $|N|$ is odd, $|N| \nmid \text{char}(\F) = 2$, so by Maschke's theorem, $R$ is semisimple. From the Wedderburn-Artin theorem, it can be shown that $
    R \cong \mathbb{F}_{2^{k_1}} \times \cdots \times \mathbb{F}_{2^{k_m}}$ \cite{panteleev2021quantum}. The explicit form of the isomorphism is $\psi: r \mapsto (r \cdot e_1, r \cdot e_2, \dots, r \cdot e_m)$ where $\{e_1, \dots, e_m\}$ is the set of primitive orthogonal idempotents, which also satisfy a completeness relation $\sum_i e_i = 1_R = 1_G$. The isomorphism can be extended to matrix rings $R^{m \times n} \cong \bigoplus_i \mathbb F_{2^{k_i}}^{m \times n}$. Let $A \in R^{n \times m}$ and $B \in R^{m \times p}$. Then, $AB = \bigoplus_{i, j} Ae_i B e_j = \bigoplus_{i} A_{(k_i)} B_{(k_i)}$, and, identically, $A \otimes_R B = \bigoplus_{i} [A_{(k_i)} \otimes_{\mathbb{F}_{2^{k_i}}} B_{(k_i)}]$. Therefore, $\mathcal C \simeq \bigoplus_i \mathcal{C}_{(k_i)}$, where $\mathcal C_{k_i}$ is a $\mathbb F_{2^{k_i}}$-code, with the  parity check matrix decomposed as $\delta^0 \simeq \bigoplus_i \delta^0_{(k_i)}$. 

Note that for any $v \in R^m$, $|v|_R \leq |N||v|$. Thus, $d_H(\delta^{0}_{(k_i)}) \geq d_H(\delta^0) \geq d_R(\delta^0)/|N| = d_H(\mathbb B_N(\delta_0))/|N|$ provides a lower bound on the Hamming distance of each $\delta^0_{(k_i)}$. Since the tensor product over $R$ splits component-wise, one can see the $R$-homological product codes also decompose component-wise, e.g., $\breve{\mathcal C} = \bigoplus_i[\tilde{\mathcal{C}}_{(k_i)} \otimes_{\mathbb{F}_{2^{k_i}}} \tilde{\mathcal{C}}_{(k_i)}]$. We apply lemmas \ref{lem: First soundness} and \ref{lem: Second soundness} to each classical code component $\mathcal{C}_{(k_i)}$, obtaining that every $\breve \delta^0_{(k_i)}$ (and $(\breve \delta^{-1}_{(k_i)})^T$) are at least $(d_H(\delta^0_{(k_i)}), g)$-sound. Let $s \in \im \tilde \delta_0$ decompose as $s = \bigoplus_i s_{(k_i)}$. Since $\psi$ is any isomorphism, so $\ker \psi^{-1} = \{0\}$, $|s_{(k_i)}|_H \leq |s|_H$. If $|s|_R \leq d_H(\mathbb B_N(\delta_0))/|N|$, then it follows by soundness that for each $i$, there exists $v_{(k_i)} \in \mathbb{F}_{2^{k_i}}^{6l^4}$ such that $s_{(k_i)} = \breve \delta^0_{(k_i)} v_{(k_i)}$, and $|v_{(k_i)}| \leq |s_{(k_i)}|^3/4$. This gives $|v|_R \leq m |v_{(k_i)}|_H \leq m|s|_H^3/4 \leq  m|s|^3_R/4 = m|\mathbbm b_N(s)|^3/4$. Finally, since $m$ is the number of equivalence classes of integer pairs $(x, y)$ such that $y = 2^{k} x \bmod |N|$ for some integer $k$ (formally known as the number of 2-cyclotomic cosets modulo $|N|$), it has the upper bound $m \leq (|N|+1)/2$. 
\end{proof}

\begin{corollary} \label{thm: tricycle soundness} \textnormal{(Soundness of tricycle codes)}: Consider any connected tricycle code with parameters $[[n, k, (d_Z, d_X)]]$, which can be expressed as the balanced product of binary classical codes $\mathcal C_1, \mathcal C_2$, and $\mathcal C_3$. There exists a decoder such that for any pattern of measurement errors $u \in \mathbb{F}_2^n$ satisfying $|u| \leq \frac{1}{2|N|}\min_i d_i$ and data errors $v \in \mathbb F_2^n$ satisfying $(|N|+1)|u| + |v| < d_Z/2$, the decoder takes in the syndrome $H_Z v + u$ and outputs $\hat v$ such that, modulo stabilizers, $|(\hat v + v) \bmod 2|$ leaves a residual error of size at most $(|N|+1)|u|$. 
\end{corollary}
\begin{proof}
    First, note that it is easy to extend \cref{thm: 4D RHPC Soundness} to when the classical seed codes are different, though this generally results in the bases having different soundness parameters. Then, recognizing that a tricycle code is just a 4D $R$-homological product code in which one of the seed codes is a degree-0 chain complex \cite{hong2024single}, one can check that the unaffected basis remains sound with the claimed $t = \min_i d_i$ and the claimed $g$. As mentioned in \cref{sec::single_shot}, one can modify a minimum-weight decoder to decode the metacheck matrix $H_{\text{meta}}$ in such a way that $d_{ss} = \infty$. Then, combining Theorem 1 of \cite{campbell2019theory}, \cref{thm::dzdx}, and \cref{thm: 4D RHPC Soundness}, we obtain the corollary.
\end{proof}

As with \cref{cor: inequality for disconnected codes}, this result can be readily extended to disconnected codes. Even though \cref{thm: tricycle soundness} only applies to codes where $|N|$ is odd, we expect, based on our numerics in \cref{sec::single_shot} and \cref{sec::noise_sim} that codes with even $|N|$ are sound. In fact, many of our codes in  \cref{table::code_polys} are connected with even-order intersection subgroups. 

\section{Determining code distances}\label{sec::distance_finding}

We use three distinct methods to either find the exact distances or determine high statistical confidence estimates of the distances of tricycle codes. 

The first method is a Mixed-Integer-Program (MIP). This involves first using Gaussian row-reduction over $\F$ to find a basis of logical $X$ and $Z$ operators $L_X$ and $L_Z$ independently for a given code. One then formulates the low-weight logical operator search as a mixed-integer-program by iterating over the basis of logical operators. For example, to find $d_X$, for the $i^{th}$ logical operator $l_i \in L_Z$, we minimize the hamming weight of candidate bitstring solutions $b$ under the constraints $H_Z b = 0 \ \mathrm{mod}(2)$ and $b^T l_i = 1 \ \mathrm{mod}(2)$ where the latter constraint enforces that $b$ corresponds to an $X$-type logical operator that anticommutes with the $i^{th}$ logical $Z$ operator. The minimum weight logical operator is then found after iterating over all $l_i$ if the solver converges within time constraints. We used the Gurobi \cite{gurobi} package in Python for this task. 

We also sometimes use the shortest-error search method provided in the STIM package \cite{gidney2021stim} combined with the PySat package \cite{imms-sat18} to directly formulate the problem of finding code distances as a Boolean constraint satisfaction (SAT) problem. For certain codes, we found that this method converged while the previous MIP method times out within time constraints, or vice-versa.

The MIP and SAT methods are not readily useful for the larger codes we consider in this work, and sometimes do not converge in one basis even for smaller codes, as in general, finding the exact distance of a linear code is known to be $NP$-Hard. In these cases, we estimate the minimum distance of a CSS quantum code using a Monte Carlo algorithm that performs structured sampling over the space of low-weight logical operators. The algorithm is based on the information-set algorithm of \cite{leon1988probabilistic} for classical linear codes, with the additional step that one has to remove low-weight codewords belonging to the row-space of the opposite parity check matrix. We note that the algorithm we use is essentially identical to that of the QDistRnd package \cite{pryadko2023qdistrnd} that is provided in the GAP computational algebra programming language. We developed an independent implementation in the python programming language. We briefly summarize the algorithm as well as the empirical probability of success estimates derived for the QDistRnd algorithm.

We discuss how to estimate $d_Z$, with the analysis for $d_X$ being identical with the parity check matrices swapped. The objective is to find a nonzero vector \( c \in \ker(H_X) \) such that \( c \notin \text{rowspace}(H_Z) \), with minimal Hamming weight. The algorithm proceeds by first computing a basis \( G \) for the nullspace \( \ker(H_X) \) using row-reduction over $\F$. In each Monte Carlo trial, a random column permutation \( \pi \) is applied to the columns of \( G \), followed by row reduction (forward elimination) to produce a low-weight generating set for \( \ker(H_X) \). The inverse permutation \( \pi^{-1} \) is then applied to return to the original coordinate system. Rows of the resulting matrix that are linearly independent of the rows of \( H_Z \) are retained as candidate logical operators -- this is once again determined using efficient row-reduction over $\F$. The smallest observed weight among all valid logical operators encountered over many trials is reported as the estimated distance.

To assess the reliability of the distance estimate, we track the number of rediscoveries of each distinct minimum-weight logical operator over \( N \) trials. Let \( m \) denote the number of distinct minimum-weight vectors found, and let \( n_i \) be the number of times the \( i \)-th one is observed. We define the average rediscovery count as
\[
\langle n \rangle = \frac{1}{m} \sum_{i=1}^m n_i.
\]
Assuming that each trial samples a minimum-weight codeword independently and with equal probability \( \lambda > 0 \), the number of times each such codeword is found follows a Poisson distribution with mean \( N \lambda \). Under this model, the probability that \emph{no} minimum-weight codeword is discovered in \( N \) trials is at most
\[
P_{\mathrm{fail}} \leq \exp(-\langle n \rangle),
\]
where \( \langle n \rangle \approx N \lambda \). This gives a conservative upper bound on the probability of missing the true minimum distance. In particular, when \( \langle n \rangle \gg 1 \), the failure probability becomes exponentially small.

To further evaluate sampling quality, we test whether the distribution of rediscovery counts \( \{n_i\} \) is consistent with a uniform multinomial distribution. If each minimum-weight codeword is equally likely to be sampled, the frequencies \( \{n_i\} \) should follow a multinomial distribution with equal probabilities. Deviations from this uniformity are detected using a Pearson chi-squared test. A small associated p-value indicates uneven sampling, suggesting either bias in the heuristic or incomplete exploration of the codeword space. In the main text, we only report the distance found from this method as exact if the returned confidence $1-P_{fail} > 0.99$ and if the p-value associated with the chi-squared test is $>0.1$. To ensure reliability of the hypothesis test, we also only infer a p-value from the chi-squared test if at least $2$ distinct minimum weight code-words are found, and if at least $5$ occurrences of each codeword are found. In all other cases, we report the estimated distances as an upper bound, although the noise simulations in \cref{sec::noise_sim} show that noise suppression is compatible with code distances that are equal or near these estimates.

\section{Symmetric triple cup-product and transversal CCZ action}\label{sec::cup}

Now we discuss how we derive the transversal $CCZ$ circuits discussed in \cref{sec::transversal_CCZ} from the cup-product logical gate formalism of \cite{breuckmann2024cups}. 

First, we introduce the necessary formalism at a high level. A cup-product is an algebraic structure that can be defined on some cochain complexes, such as the one in \cref{eq:tricycle_complex}. Let us label the cochain spaces of the complex in \cref{eq:tricycle_complex} as $0,1,2,3$-cochains from left to right, indicated with the notation $C^0$, $C^1$, $C^2$, and $C^3$ respectively. Furthermore, we are interested in the case where the $C^i$ are finite-dimensional vector spaces over $\F$ with bases $X^i$. We define the weight $|c|$ of a cochain $c \in C^i$ as the number of non-zero terms in its basis expansion over $X^i$. For example, for the tricycle code complex in \cref{eq:tricycle_complex}, the bases $X^i$ are defined by $G$, and the cochain weight coincides with the weight of group-algebra elements. We will always associate qubits with the $1-$cochains of a complex, with $0-$ and $2-$ cochains corresponding to spaces of $X$ and $Z$ checks respectively. Thus, we will label the elements of the basis of $1-$cochains $X^1$ by qubit labels $q_1 \cdots q_N$ for the $N$ qubits of a CSS quantum code. 

A cup-product is then a bilinear map 

\begin{equation}\label{eq:cup_def}
    \cup: C^i \times C^j \rightarrow C^{i+j}
\end{equation}

with the convention that $C^{i+j} = 0$ if $i+j > 3$ in the case of the tricycle complex. Not every such bilinear map defines a cup-product; the map $\cup$ must satisfy an additional technical condition known as a Leibniz rule -- see \cite{MR1867354, breuckmann2024cups, zhu2025topological}. However, given such a cup-product, we can construct a circuit of diagonal physical gates that preserves the code-space of the CSS code associated with the cochain complex. In general, a cup-product on an $m-$term cochain complex will yield a physical circuit of diagonal $C^{m-2}Z$ gates from the $m-1$-th level of the Clifford hierarchy \cite{breuckmann2024cups}.  In the situation of interest, we have the $4-$term cochain complex in \cref{eq:tricycle_complex}, leading to a circuit of $CCZ$ gates. As discussed in \cref{sec::transversal_CCZ}, such a circuit can be specified by a trilinear function $f$ acting on the space of physical qubits -- the $1-$cochains $C^1$ in our convention. The circuit associated with the cup-product is then defined by the function

\begin{equation}\label{eq:cup_circuit}
    f_{\cup}(q_i, q_j, q_k) = |(q_i\cup q_j) \cup q_k| \ \mathrm{mod} \ 2\in \F
\end{equation}

where non-zero values indicate a physical $CCZ$ gate acting between the qubits of three distinct code-blocks labeled by those arguments of $f_\cup$. The $|\cdot|$ here is the size of the support of the argument when considered as a binary vector with respect to the $X^i$ bases. Circuits defined by such functions constructed from cup-products of finitely many $1-$cochains (three, in this case) lead to sparse, finite depth circuits \cite{breuckmann2024cups, lin2024transversal}. Moreover, such functions $f_\cup$ naturally have the property that they vanish on $1-$coboundaries: $f_{\cup}$ vanishes if any of its three arguments is of the form $d^0(c)$ for $c\in C^0$, where $d^0$ is the coboundary map from $C^0$ to $C^1$ -- equivalent to the matrix $H_X^T$. This condition is thus equivalent to $f_{\cup}$ vanishing when any one of its arguments corresponds to to an $X-$type stabilizer while the other two arguments correspond to non-trivial $X-$type logical operators (see \cref{sec::transversal_CCZ})  and is called \textit{coboundary} invariance. The coboundary invariance property is essential to ensure that the circuit defined by $f_{\cup}$ descends to a well-defined logical operation \cite{lin2024transversal, breuckmann2024cups} -- intuitively, this is so that the action of $f_{CCZ}$ on logical operators does not change upon adding stabilizers to the logical operators. 

The formalism of \cite{breuckmann2024cups} constructs cup-products on classical $2$-term cochain complexes which satisfy certain sufficient conditions such that homological and balanced product complexes constructed from them (see supplementary material) are also equipped with an appropriate cup-product. This yields a trilinear function of the form in \cref{eq:cup_circuit} on the product complexes associated with quantum CSS codes. Since tricycle codes are balanced products of classical Abelian group-algebra codes, these methods readily apply, which we review here. We will now review the key elements of this formalism.

To equip a two-term cochain complex $\begin{tikzcd}
C^\bullet: C^0 \arrow[r, "\delta"] & C^1
\end{tikzcd}$ corresponding to a classical linear code with a cup-product, one must construct partitions of the $1-$coboundaries $B^1$ (see supplementary material), known as \textit{pre-orientations}. For any $a \in C^0$, we consider partitions 

\begin{equation}\label{eq:preorientations}
    \delta(a)=\delta_{in}(a) \sqcup \delta_{out}(a) \sqcup \delta_{free}(a)
\end{equation}

where $\sqcup$ denotes a disjoint union and we identify an element $x \in C^\bullet$ with its support as binary vector. The corresponding partitions are called 'in', 'out' and 'free' preorientations of the element $a\in C^0$ respectively. A choice of pre-orientation induces a cup-product on the classical code by the following definition.

\begin{definition}[Classical complex cup-product]
The cup-product on a 2-term cochain complex induced by the preorientations of \cref{eq:preorientations} is defined as 
    \label{def:classical_cup_def}
    \begin{align}
        & a \cup a = a \quad \forall a\in C^0 \\
        & a \cup x = x \quad \text{for } a\in C^0, x\in C^1 \ \quad \text{if }  x \in \delta_{out}(a) \\
        & x \cup a = x \quad \text{for } a\in C^0, x\in C^1 \ \quad \text{if }  x \in \delta_{in}(a) \\
        & 0 \quad \text{otherwise}
    \end{align}
\end{definition}

Note that this cup-product is not associative in general: $(a\cup b)\cup c \neq a\cup(b\cup c)$. 

Different choices of partitions induce different cup-products on the classical code according to the rules in \cref{def:classical_cup_def}. Given three such classical group algebra codes defined by elements $a,b,c\in \FG$, each with their respective in, out, and free partitions, the three-fold balanced product tricycle code can be equipped with a natural cup-product that is inherited from the cup-products on the classical codes -- see Ref. \cite{breuckmann2024cups} for details on the inherited cup-product on quantum codes.  Additional conditions on the classical preorientations are required for the inherited cup-product on quantum codes -- and hence the function $f_\cup(q_i, q_j,q_k)$ as defined in \cref{eq:cup_circuit} -- to induce a code-space preserving circuit. Ref. \cite{breuckmann2024cups} shows that a sufficient condition on the preorientations is the so called \textit{integrated Leibniz rule}.

\begin{prop}[Integrated Leibniz Rule]\label{prop:integrated_leibniz}
    The cup-product on classical codes from \cref{def:classical_cup_def} is said to satisfy an integrated Leibniz rule if for all $a_1, a_2, a_3 \in C^\bullet$
    \begin{equation}\label{eq:integrated_leibniz}
        |(\delta(a_1)\cup a_2) \cup a_3| + |(a_1 \cup \delta(a_2) \cup a_3| + |(a_1\cup a_2) \cup \delta(a_3)| = 0  \ \mathrm{mod} \ 2
    \end{equation}
The integrated Leibniz rule on three classical codes is a sufficient condition for the $CCZ$ circuit defined by the inherited cup-product on the product quantum code, $f_{\cup}(q_1, q_2, q_3) = |(q_1 \cup q_2) \cup q_3| \ \mathrm{mod} \ 2$, to be a coboundary invariant operation -- i.e a code-space preserving unitary operator.
\end{prop}

\begin{proof}
    See section $5$ of Ref. \cite{breuckmann2024cups}.
\end{proof}

From \cref{eq:integrated_leibniz}, further equivalent sub-conditions can be derived for the `in', `out', and `free' partitions. These conditions are listed in proposition 5.2 of \cite{breuckmann2024cups}. However, we note that proposition $5.2$ as stated in \cite{breuckmann2024cups} is technically incorrect unless the cup-product is associative which is not generically the case. While the cup-product is indeed associative when the 'non-overlapping bits' condition -- $\delta_{in}(a_1)\cap \delta_{in}(a_2) = \varnothing$ for $a_1\neq a_2$ (and similarly for $\delta_{out}$) -- holds, this is no longer true when bits of the $in$ and $out$ partitions are allowed to overlap, which is the case for tricycle codes when one or more of the group-algebra elements have weight $>2$. The correct set of conditions can be derived by applying the cup-products in sequence as $(a\cup b) \cup c$ in the derivation of proposition 5.2 of \cite{breuckmann2024cups}. However, we find that the resulting conditions are too restrictive on the parameters of the code -- for instance, these conditions can be found in Refs. \cite{jacob2025singleshot, li2025transversal} in which they were used to study $2-2-2$ codes with $K=3$ and $4-2-2$ and $4-4-4$ codes with $D=2$ with non-trivial $CCZ$ action.

We instead derive a new set of conditions on the pre-orientiations that follow from mixing the order of cup-product applications for different orders of arguments to define a different trilinear operation -- we call this operation a \textit{symmetric triple cup-product}, as illustrated by the following definition.

\begin{definition}[Symmetric triple cup-product] \label{def:stcp}

The symmetric triple cup-product is a trilinear function on a $2-$term cochain complex 
\begin{equation*}
   \_\cup\_\cup\_: C^{i} \times C^{j} \times C^{k} \rightarrow C^l \quad i,j,k,l \in \{0,1\}
\end{equation*}

that is non-zero only if $i+j+k=l$ and is defined as

\begin{align}
a_1\cup a_2 \cup a_3 \equiv & \nonumber\\
    &(a_1 \cup a_2) \cup a_3 \quad \mathrm{if} \ a_1\in C^1, a_2\in C^0, a_3 \in C^0 \label{eq:sym1}\\
    & (a_1 \cup a_2) \cup a_3 \quad \mathrm{if} \ a_1\in C^0, a_2\in C^1, a_3 \in C^0 \label{eq:sym2} \\
    & a_1 \cup (a_2 \cup a_3) \quad \mathrm{if} \ a_1\in C^0, a_2\in C^0, a_3 \in C^1 \label{eq:sym3}\\
    & (a_1 \cup a_2) \cup a_3 \quad \mathrm{if} \ a_1\in C^0, a_2\in C^0, a_3 \in C^0 \label{eq:sym4}\\
    & 0 \quad \text{otherwise}
\end{align}

where the cup-products in parentheses are the standard ones defined in \cref{def:classical_cup_def}.
\end{definition}

Henceforth, whenever a triple cup-product is written with no parentheses, it is to be understood as the symmetric version defined above. Note that the order of parentheses for the specific cases in \cref{eq:sym2} and \cref{eq:sym4} does not matter as for such arguments the standard cup-product is associative. The product quantum code (homological or balanced products) of three classical codes equipped with symmetric triple cup-products naturally inherits a symmetric triple cup-product operation defined on the $4$-term cochain complex corresponding to the quantum code. The symmetric triple cup-product on product quantum codes is defined essentially identically as in the original cup-product construction of Ref. \cite{breuckmann2024cups}, except with the order of operations taken as in \cref{def:stcp}. Similarly, one can show that the appropriate version of the integrated Leibniz rule is sufficient for the trilinear function defined by the symmetric triple cup-product on quantum codes to produce a code-space preserving $CCZ$ circuit.

\begin{prop}[Symmetric Integrated Leibniz Rule]\label{prop:integrated_leibniz}
    The symmetric triple cup-product on classical codes is said to satisfy an integrated Leibniz rule if for all $a_1, a_2, a_3 \in C^\bullet$
    \begin{equation}\label{eq:integrated_leibniz}
        |\delta(a_1)\cup a_2 \cup a_3| + |a_1 \cup \delta(a_2) \cup a_3| + |a_1\cup a_2 \cup \delta(a_3)| = 0  \ \mathrm{mod} \ 2
    \end{equation}
The integrated Leibniz rule on the symmetric triple cup-product is a sufficient condition for the $CCZ$ circuit defined by the inherited symmetric triple cup-product on the quantum code, $f^{sym}_{\cup}(q_1, q_2, q_3) = |q_1 \cup q_2 \cup q_3| \ \mathrm{mod} \ 2$, to be a coboundary invariant operation -- i.e a code-space preserving unitary operator.
\end{prop}

\begin{proof}
    The proof is analogous to that of Lemma 5.1 of Ref. \cite{breuckmann2024cups}.
\end{proof}

The symmetric integrated Leibniz rule can then be expressed in terms of a number of sub-conditions on the coboundary partitions.

\begin{prop}

A preorientiation on a classical code induces a symmetric triple cup-product that satisfies the integrated Leibniz rule if the 'in', 'out', and 'free' partitions satisfy the following conditions:

\begin{align}
    & |\delta_{in}(a_1)|+|\delta_{out}(a_1)| = 0 \ \mathrm{mod}(2) \quad \forall a_1 \in X^0 \ \label{eq:cup_cond1}\\
    & |\delta_{in}(a_1) \cap \delta_{free}(a_2)| = 0 \ \mathrm{mod}(2)\quad \forall a_1,a_2\in X^0: \ a_1\neq a_2 \label{eq:cup_cond2}\\
    & |\delta_{out}(a_1) \cap \delta_{free}(a_2)| = 0 \ \mathrm{mod}(2) \quad \forall a_1,a_2\in X^0: \ a_1\neq a_2 \label{eq:cup_cond3}\\
    & |\delta_{in}(a_1) \cap \delta_{in}(a_2)| + |\delta_{out}(a_1) \cap \delta_{out}(a_2)|  = 0 \ \mathrm{mod}(2)\quad \forall a_1,a_2\in X^0: \ a_1\neq a_2 \label{eq:cup_cond4}\\
    & |\delta_{in}(a_1) \cap \delta_{in}(a_2) \cap \delta_{in}(a_3)| + |\delta_{out}(a_1) \cap \delta_{out}(a_2) \cap \delta_{out}(a_3)| + |\delta_{free}(a_1) \cap \delta_{in}(a_2) \cap \delta_{in}(a_3)| \label{eq:cup_cond5} \\ &+ |\delta_{out}(a_1) \cap \delta_{out}(a_2) \cap \delta_{free}(a_3)| + |\delta_{out}(a_1) \cap \delta_{free}(a_2) \cap \delta_{in}(a_3)| = 0 \ \mathrm{mod}(2) \quad \nonumber \\ &\forall a_1,a_2,a_3\in X^0: \ a_1\neq a_2\neq a_3 \nonumber
\end{align}

where $X^0$ is the basis set for $C^0$.

\end{prop}

\begin{proof}
    By direct case-by-case computation using \cref{def:classical_cup_def}, \cref{def:stcp}, and the Leibniz rule.
\end{proof}

Note that these conditions apply to the classical codes from which any three-dimensional homological product and balanced product quantum codes are constructed. Compared to the original conditions of Ref. \cite{breuckmann2021balanced} these conditions are less restrictive and may yield product codes with $CCZ$ gates with better parameters for other code families besides tricycle codes.

We now apply these conditions to the classical Abelian group Algebra codes which are used to construct tricycle codes. Below, we will use $|a|$ to denote the weight of the group-algebra element $a$ -- i.e the number of group elements that appear in its expansion. We also write $a_1\cap a_2$ for $a_1,a_2\in \FG$ to denote the set of group elements that are common to both group-algebra elements. For a classical Abelian group-algebra code associated with the two-term cochain complex of \cref{eq:classical_complex} defined by an element $a\in \FG$, defining a preorientation amounts to choosing a partition of the group algebra element into smaller disjoint group-algebra elements

\begin{equation}\label{eq:ga_partition_classical}
    a = a_{in}+a_{out}+a_{free}
\end{equation}

where $a_{in}, a_{out}, a_{free}\in \FG$. The preorientations are then defined as $\delta_{in}(\alpha) = a_{in}\cdot \alpha$, $\delta_{out}(\alpha) = a_{out}\cdot \alpha$ and $\delta_{free}(\alpha) = a_{free}\cdot \alpha$ for $\alpha\in \FG$ where the multiplications are within the group-algebra. Note that for balanced product codes, the formalism of Ref. \cite{breuckmann2021balanced} requires an additional condition that preorientations are preserved by group actions. The preorientations in \cref{eq:ga_partition_classical} naturally satisfy this condition -- for example, $g\cdot \delta_{in}(\alpha) = g \cdot a_{in} \cdot \alpha = \delta_{in}(g\cdot \alpha)$ and similarly for the other partitions.

The conditions in \cref{eq:cup_cond1}-\cref{eq:cup_cond5} take a particularly simple form when the free partition is empty. In this case, the conditions on the classical preorientation of the code defined by $a\in \FG$ are then (with the conditions on the $b$ and $c$ classical code preorientations defined identically):

\begin{align} 
    & |a_{in}|+|a_{out}| = 0 \ \mathrm{mod}(2) \label{eq:empty_free1}\\
    & |a_{in}\cdot g \cap a_{in}\cdot h| + |a_{out}\cdot g \cap a_{out}\cdot h|  = 0 \ \mathrm{mod}(2)\quad \forall g,h\in G: \ g\neq h \label{eq:empty_free2}\\
    & |a_{in}\cdot f \cap a_{in}\cdot g \cap a_{in}\cdot h| + |a_{out}\cdot f \cap a_{out}\cdot g \cap a_{out}\cdot h| = 0 \ \mathrm{mod}(2) \quad \label{eq:empty_free3} \\ &\forall f,g,h\in G: \ f\neq g\neq h\nonumber
\end{align}

To construct high-rate and distance tricycle codes from weight-$4$ group-algebra elements which satisfy the above conditions, it turns out to be necessary to set the free partition to be empty (see the discussion after \cref{thm:weight4_partitions}), and thus we will work with this version of the Leibniz rule conditions. 
For three classical group-algebra codes with preorientiations that satisfy the above conditions, the balanced product yields a symmetric triple cup-product on the quantum code's cochain complex, and consequently produces a constant depth $CCZ$ circuit defined by the induced symmetric triple cup-product on the quantum code $f_{CCZ}(q_1, q_2, q_3) = |q_1 \cup q_2 \cup q_3| \ \mathrm{mod} \ 2$ (the precise definition of $f_{CCZ}$ is discussed in \cref{prop:fccz_def} below).

Note that for weight $2$ group-algebra elements, these conditions are trivially satisfied automatically by choosing $a_{in}$ and $a_{out}$ to be singletons. All codes in \cref{table::all_codes}  are constructed by finding classical group-algebra elements which satisfy these conditions with a numerical search. A simple rule to generate weight $4$ group-algebra elements that satisfy these conditions is to choose $a = a_1 + a_1\cdot s + a_2 + a_2\cdot s$ for a fixed element $s \in G$ such that none of the $4$ terms are equal. As long as $s^2 \neq e$, we are able to find tricycle codes constructed from such weight-$4$ elements with non-trivial $CCZ$ circuits and $k > 3, d > 2$, circumventing the limitations of the original conditions in \cite{breuckmann2024cups}. We formalize this with the following theorem.

\begin{theorem}\label{thm:weight4_partitions}
    There exist weight-$4$ elements $a = a_1+a_2+a_3+a_4 \in \FG$ that satisfy conditions \cref{eq:empty_free1} - \cref{eq:empty_free3} for any $G$ with $|G|\geq 4$.
\end{theorem}

\begin{proof}
We specify one prescription for constructing preorientations which satisfy \cref{eq:empty_free1} - \cref{eq:empty_free3}. First note that $|\alpha \cdot g \cap \beta \cdot h| = |\alpha \cap \beta \cdot hg^{-1}|$ for any $\alpha,\beta \in \FG$ and $g,h\in G$.

Thus we can simplify the conditions to (where $e$ denotes the group identity element below)

\begin{align} 
    & |a_{in}|+|a_{out}| = 0 \ \mathrm{mod}(2) \label{eq:cup_proof1}\\
    & |a_{in}\cap a_{in}\cdot w| + |a_{out} \cap a_{out}\cdot w|  = 0 \ \mathrm{mod}(2)\quad \forall w \in G: \ w\neq e \label{eq:cup_proof2}\\
    & |a_{in} \cap a_{in}\cdot v \cap a_{in}\cdot w| + |a_{out} \cap a_{out}\cdot v \cap a_{out}\cdot w| = 0 \ \mathrm{mod}(2) \quad \label{eq:cup_proof3} \\ &\forall v,w\in G: \ v\neq w\neq e, \nonumber
\end{align}

Now we choose any $a_1, a_3\in G$ with $a_1 \neq a_3$ and choose another fixed element $s \in G$ with $s\neq e$ such that: $a_1 s \neq a_3s$ and $a_1  s \neq a_3$ and $a_1 \neq a_3s $. Then we define $a_2 \equiv a_1s$ and $a_4\equiv a_3 s$ and set the partitions $a_{in}=a_1+a_2$, $a_{out}=a_3+a_4$ and $a_{free}=\varnothing$. \cref{eq:cup_proof1} is clearly satisfied since both partitions have weight $2$. Next, notice that $|a_{in}\cap a_{in}w| = \mathds{1}_{w=s}$ and similarly $|a_{out}\cap a_{out}w| = \mathds{1}_{w=s}$. Thus $|a_{in}\cap a_{in}w|+|a_{out}\cap a_{out}w| = 2\cdot \mathds{1}_{w=s} = 0 \ \mathrm{mod} \ 2$ so \cref{eq:cup_proof2} is also satisfied. Finally, since $a_{in}$, $a_{out}$ both are weight $2$, it is easy to see that $a_{in} \cap a_{in}\cdot v \cap a_{in}\cdot w = \varnothing$ and similarly for $a_{out}$ (this is true for any weight $2$ group-algebra elements and non-identity $v\neq w$). Thus, \cref{eq:cup_proof3} is also satisfied.
\end{proof}

The prescription in the proof of \cref{thm:weight4_partitions} is not strictly the only way to satisfy conditions \cref{eq:cup_cond1} - \cref{eq:cup_cond5} for weight-$4$ elements. Another possibility involves working with groups $G$ that have an \textit{involution} $t\in G$ such that $t^2=e$. One can then construct preorientations with non-empty free partitions and one of the in and out partitions being empty -- for example, by choosing $a_{in}=a_1+a_1t$ and $a_{free}=a_3+a_3t$. However, we empirically find that tricycle codes which use such group algebra elements always have distance $D=2$. More generally, we find that when the offset $s$ used to construct the partitions of \cref{thm:weight4_partitions} is an involution, $D=2$ always, although we are unable to prove it. This appears to be the same structure that causes all the $4-2-2$ and $4-4-4$ codes studied in Ref. \cite{jacob2025singleshot} using the original cup-product conditions to have $D=2$, as their construction can also be interpreted in terms of using an involutive group element to satisfy the associated conditions. The symmetric triple cup-product conditions we introduce instead allow us to choose offsets $s$ that are not involutions, leading to tricycle codes with better parameters.

Next, we define the $f_{CCZ}$ function for the quantum tricycle codes that arises from the cup-product specified by a choice of preorientations on the constituent classical group-algebra codes. The structure of the cochain complex in \cref{eq:tricycle_complex} implies that we can partition the qubits of the code into three sectors of $|G|$ qubits each -- see \cref{sec::balanced}. Within each sector, a qubit can be labelled by a group element. Given some order of the group elements, the notation $g^i$ will thus be used to denote qubit $g$ of the $i$-th sector ($i=I,II,III$). For the following, denote $\alpha_I \equiv a$, $\alpha_{II} \equiv b$, $\alpha_{III} \equiv c$ as the three group-algebra elements that specify the classical codes (this notation makes the expression below compact). 

\begin{prop}\label{prop:fccz_def}
Consider a tricycle code defined by group-algebra elements $\alpha_I, \alpha_{II}, \alpha_{III} \in \FG$ with respective `in' and `out' pre-orientations. The function $f_{CCZ}$ specifying the transversal $CCZ$ circuits in \cref{sec::transversal_CCZ}   is defined as 

\begin{equation}\label{eq:fccz_def}
    f_{CCZ}(p^i, q^j, r^k) = |r \alpha_i^{in} \alpha_j^{in} \cap q \alpha_i^{in} \alpha_k^{out} \cap p \alpha_j^{out} \alpha_k^{out} |\cdot \delta_{i\neq j\neq k}\quad \mathrm{mod}(2) 
\end{equation}

where $i,j,k \in \{I,II,III\}$ label the sectors of the qubits $p,q,r$ in the three code-blocks respectively, and $\delta_{i\neq j\neq k}$ indicates that $f_{CCZ}=0$ if any two or more qubits from its arguments are in the same sector. If any partition in the above expression is empty, $f_{CCZ}=0$ on these arguments.
\end{prop}
\begin{proof}[Proof (informal)]
    This definition follows from a straightforward but tedious computation of the cup-product on the balanced-product quantum code induced by the classical code cup-products, with the new mixed-orders we define to attain the conditions in \cref{eq:cup_cond1}--\cref{eq:cup_cond4}. 
\end{proof}

Note that the expression in \cref{prop:fccz_def} holds for any tricycle code equipped with pre-orientations that satisfy the conditions of \cref{eq:cup_cond1}-\cref{eq:cup_cond5}, not just those constructed from weight-$4$ group-algebra elements. To be clear, the resulting $CCZ$ circuit has a $CCZ$ gate acting on qubits $p^i, q^j, r^k$ if and only if $f_{CCZ}(p^i, q^j, r^k)=1$. The $\delta_{i\neq j \neq k}$ structure is manifest in the illustrations of the circuits in \cref{fig:CCZ_connectivity}, where no $CCZ$ acts between qubits in the same sector across code-blocks. One must check (numerically) whether the resulting circuits have a non-trivial logical action or not, as has been done for all the codes listed in \cref{table::all_codes}.

From the $f_{CCZ}$ function of \cref{prop:fccz_def}, it is not difficult to ascertain the maximum degrees (the maximum number of $CCZ$ gates a particular qubit is involved in) of the resulting $CCZ$ circuits which use the construction in \cref{thm:weight4_partitions}to construct weight-$4$ group-algebra elements with valid preorientations for the $4-2-2$, $4-4-2$, and $4-4-4$ type tricycle codes. Below, we let $s_i$ denote the offset element in the proof of \cref{thm:weight4_partitions} chosen for $\alpha_i$, $i\in \{I,II,III\}$, and we list weight-$4$ elements before weight-$2$ elements -- for example, for $4-4-2$ codes we let $\alpha_I$ and $\alpha_{II}$ be weight-$4$ and $\alpha_{III}$ be weight-$2$.

A straightforward analysis shows that the possible maximum circuit degrees are:

\begin{enumerate}
    \item $4-2-2$ codes: degree $8$
    \item $4-4-2$ codes:
    \begin{itemize}
        \item $s_1 = s_2$: degree $16$
        \item If $s_1 \neq s_2$: degree $32$
    \end{itemize}
    \item $4-4-4$ codes:
    \begin{itemize}
        \item $s_1=s_2=s_3$: degree $12$
        \item $s_i \neq s_j = s_k$ for $i,j,k \in \{I,II,III\}$ and $i\neq j \neq k$:
        degree $64$
        \item $s_1\neq s_2 \neq s_3$: degree 128
    \end{itemize}
\end{enumerate}

We found empirically that the $4-4-4$ codes with degree $64$ and degree $12$ circuits from the above constructions had relatively poor rate and distance parameters, which is why they are not included \cref{table::all_codes}.

In the main text, we abuse terminology and refer to the maximum degree of a circuit and the minimal depth of the circuit interchangeably. For the $CCZ$ circuits considered in this work, the true minimal circuit depth is related to the chromatic index of the corresponding tripartite $3-$uniform hypergraph. A standard bound \cite{obszarski2017edge} relates the minimal depth  and maximum degree $\Delta$ as

\begin{equation}
    \text{depth} \leq 3(\Delta-1)
\end{equation}

Another widely believed conjecture \cite{alon1997degree} states thst 

\begin{equation}
    \text{depth} \leq \frac{3\Delta}{2}
\end{equation}

for $3-$uniform hypergraphs where any two hyperedges share at-most $2$ vertices, as is the case for the $CCZ$ circuits we consider in our work. For example, using an integer programming method, we found that the minimal depth of the $CCZ$ circuit for the main example $[[48,6,(8,4)]]$ $4-2-2$ code considered in the main-text (\cref{sec::noise_sim}) is $10$, while the maximum degree of the circuit is $8$. In practice, it is the degree of the circuit that controls the spread of errors between qubits, while extra circuit layers from scheduling conflicts simply add a relatively small idling error.

\section{Numerical Leibniz rule for $CCZ$ gates on $3D$ balanced-product codes}\label{sec::numerical_ccz}

We now present a numerical method that generalizes the cup-product based gate approach to construct code-space preserving $CCZ$ circuits on three-dimensional balanced product codes. We applied this method to two $4-4-4$ tricycle code in \cref{table::all_codes} to find circuits with depth shorter than that from the symmetric triple cup-product (STCP) conditions in \cref{eq:cup_cond1}-\cref{eq:cup_cond5} (which typically have depth $128$ for $4-4-4$ codes). We refer to the method described in the remainder of this section as the 'Numerical Leibniz Rule' (NLR). This method yields code-space presercing circuits for any three-dimensional balanced product code. However, unlike the cup-product based construction, there is no guarantee on the depths of the resulting circuits. We describe a procedure that nevertheless allows some empirical control of the resulting circuit depths for certain codes, which are typically shorter than or comparable to the depths of the cup-product based $CCZ$ circuits with appropriate choices of method hyperparameters. We present the key ideas in this section, while we intend to provide a more detailed exposition of this method in future work.

The approach involves defining an Ansatz for a code-space preserving trilinear function $f_{CCZ}$ on $3D$ balanced-product codes in terms of trilinear functions on the component classical codes. We then impose a generalization of the integrated Leibniz rule in \cref{prop:integrated_leibniz} on each of the classical trilinear functions, which we will show is a sufficient condition for the trilinear function on the quantum code to be a coboundary invariant operation corresponding to a code-space preserving $CCZ$ circuit. We assume familiarity with the homological formulation of balanced product codes -- see the supplementary material for a brief introduction. For simplicity of terminology, we describe the method as it applies to tricycle codes, but it can be readily extended to other $3D$ balanced product codes as well as to code-space preserving $C^kZ$ circuits on $k$-dimensional balanced product codes. 

We consider the cochain complex of \cref{eq:tricycle_complex} with the $1-$cochains $\tilde{C}^1$ denoting the space of qubits. A basis for $\tilde{C}^1$ is given by elements \begin{equation}\label{eq:tricycle_1cochain_basis}
    p_1 \coloneq[p,1,1]_G \quad q_2 \coloneq[1,q,1]_G \quad r_3 \coloneq [1,1,r]_G \quad \forall p,q,r \in G
\end{equation}

where $\alpha_1, \alpha_2, \in G$ are $1$ cochains of the classical group-algebra code complex in \cref{eq:classical_complex}, which we denote as $C^1_i$ ($i=1,2,3$ for each classical code), and $1 \in G$ denotes the group identity element which is to be understood as an element of the $0-$cochains of the classical complex denoted by $C^0_i$ (see section $5.3$ of Ref. \cite{breuckmann2024cups} for more details). The subscripts $1,2,3$ indicate which index of the vector is an element of $C^1_i$, while the others are elements of $C^0_i$.

The notation $[\cdot,\cdot,\cdot]_G$ is to be understood as a vector in $(\FG)^{\oplus 3}$ subject to the equivalence relation 
\begin{equation}\label{eq:balanced_relation}
    [g^{-1} \cdot \alpha, \ gh^{-1} \cdot\beta \ , h\cdot \gamma]_G \sim [\alpha, \beta, \gamma]_G \quad \forall g,h \in G \quad \alpha,\beta,\gamma \in \FG
\end{equation}

This relation reflects the group-action quotient space structure of the balanced product code \cite{breuckmann2021balanced} (see supplementary material). We now construct a trilinear operation on the quantum code that shares some structural properties with the cup-product based trilinear function in \cref{eq:cup_circuit}. We start by defining appropriate functions on the classical code complexes.

\begin{definition}[Group-equivariant functions on classical group-algebra codes]\label{def:group_equiv_funcs}

Let $C_i$ denote the cochain complex of the $i^{th}$ classical group algebra code ($i=1,2,3$) as in \cref{eq:classical_complex}, with $C^0_i \equiv \FG$ and $C^1_i \equiv \FG$ denoting the $0$ and $1$
 cochain spaces respectively Define a set of functions \begin{equation}\label{eq:NLR_classical_func}
     f_i^j: C \times C \times C \rightarrow \F
 \end{equation}

for $i,j \in \{1,2,3\}$ with the properties 

\begin{enumerate}
    \item $f_i^j(x_1,x_2,x_3)$ is non-zero only if $x_j \in C_i^1$ and $x_k \in C_i^0$ for $k \neq j$.
    \item Group equivariance: $f_i^j(x_1, x_2, x_3) = f_i^j(g\cdot x_1, g\cdot x_2, g \cdot x_3) \quad \forall g\in G$
\end{enumerate}

\end{definition}

The $f_i^j$ functions can be viewed as a generalization of the cup-product operation on the classical codes. Now we construct a function $f_{CCZ}: \tilde{C^1}\times\tilde{C^1} \times \tilde{C^1} \rightarrow \F$ on the quantum code that is defined in terms of these functions, such that $f_{CCZ}$ respects the group-action balancing relation in \cref{eq:balanced_relation}. To this end, we consider the following structure for $f_{CCZ}$.

\begin{definition}[Product Ansatz for trilinear function on tricycle codes] \label{def:product_ansatz}

Let $f_i^j$ with $i,j\in \{1,2,3\}$ be group-equivariant trilinear functions as defined in \cref{def:group_equiv_funcs}. Let $p^i,q^j,r^k$ for $i,j,k \in \{1,2,3\}$ be basis elements of $\tilde{C}^1$ as in \cref{eq:tricycle_1cochain_basis}. Then, $f_{CCZ}$ is a trilinear function on $\tilde{C}^1$ defined on the basis elements as 

\begin{gather} \label{eq:f_ccz_nlr}
 f_{CCZ}(p_i, q_j, r_k) = \\    
\begin{cases}
\sum_{\substack{g_1, g_2, g_3 \\ h_1, h_2, h_3} \in G} f_1^1(g_1 \cdot p, g_2, g_3)\cdot f_2^2(g_1^{-1}h_1^{-1}, g_2^{-1}h_2^{-1} \cdot q, g_3^{-1}h_3^{-1})\cdot f_3^3(h_1, h_2, h_3 \cdot r) \ \mathrm{mod} \ 2 \quad  i=1,j=2,k=3 \\
\sum_{\substack{g_1, g_2, g_3 \\ h_1, h_2, h_3} \in G} f_1^1(g_1 \cdot p, g_2, g_3)\cdot f_2^3(g_1^{-1}h_1^{-1}, g_2^{-1}h_2^{-1}, g_3^{-1}h_3^{-1} \cdot r)\cdot f_3^2(h_1, h_2 \cdot q, h_3) \ \mathrm{mod} \ 2\quad  i=1,j=3,k=2 \\
\sum_{\substack{g_1, g_2, g_3 \\ h_1, h_2, h_3} \in G} f_1^2(g_1, g_2 \cdot q, g_3)\cdot f_2^1(g_1^{-1}h_1^{-1} \cdot p, g_2^{-1}h_2^{-1}, g_3^{-1}h_3^{-1})\cdot f_3^3(h_1, h_2, h_3 \cdot r) \ \mathrm{mod} \ 2\quad  i=2,j=1,k=3 \\
\sum_{\substack{g_1, g_2, g_3 \\ h_1, h_2, h_3} \in G} f_1^3(g_1, g_2, g_3 \cdot r)\cdot f_2^1(g_1^{-1}h_1^{-1} \cdot p, g_2^{-1}h_2^{-1}, g_3^{-1}h_3^{-1})\cdot f_3^2(h_1, h_2\cdot q, h_3) \ \mathrm{mod} \ 2\quad  i=2,j=3,k=1 \\
\sum_{\substack{g_1, g_2, g_3 \\ h_1, h_2, h_3} \in G} f_1^2(g_1, g_2\cdot q, g_3)\cdot f_2^3(g_1^{-1}h_1^{-1}, g_2^{-1}h_2^{-1}, g_3^{-1}h_3^{-1} \cdot r)\cdot f_3^1(h_1\cdot p, h_2, h_3) \ \mathrm{mod} \ 2\quad  i=3,j=1,k=2 \\
\sum_{\substack{g_1, g_2, g_3 \\ h_1, h_2, h_3} \in G} f_1^3(g_1, g_2\cdot q, g_3)\cdot f_2^2(g_1^{-1}h_1^{-1}, g_2^{-1}h_2^{-1}\cdot q, g_3^{-1}h_3^{-1})\cdot f_3^1(h_1\cdot p, h_2 , h_3)\ \mathrm{mod} \ 2 \quad  i=3,j=2,k=1 \\
0 \quad \text{otherwise}
\end{cases} \nonumber
\end{gather}
\end{definition}

The above definition may appear involved, but each of the $6$ nonzero cases simply defines the trilinear function on a triple of sectors of the three-blocks of the tricycle codes, producing a $CCZ$ circuit structure as in \cref{fig:CCZ_connectivity}, analogous to the cup-product based circuits from \cref{sec::cup}. The $6-$fold sum over group elements combined with group-equivariance of the component functions ensures that $f_{CCZ}$ is well-defined on the equivalence classes of the balanced product in \cref{eq:balanced_relation}. 

To ensure that the $f_{CCZ}$ function in \cref{eq:f_ccz_nlr} is a coboundary invariant operation that produces a code-space preserving $CCZ$ circuit, we impose a version of the Leibniz rule similar to that in \cref{prop:integrated_leibniz} on the $f_i^j$ functions on classical codes.

\begin{prop}[Generalized Leibniz rule]
    Suppose the $i^{th}$ classical group-algebra code is defined by an element $\alpha_i \in \FG$, as in \cref{eq:classical_complex} ($i\in \{1,2,3\}$). The functions $f_i^j$ in \cref{def:group_equiv_funcs} are said to satisfy a generalized Leibniz rule if,
    \begin{equation}\label{eq:generalized_leibniz}
        f_i^1(\alpha_i \cdot g_1, g_2, g_3) + f_i^2(g_1, \alpha_i \cdot g_2, g_3) + f_i^3(g_1, g_2, \alpha_i \cdot g_3) = 0 \ \mathrm{mod}\ 2 \quad \forall g_1,g_2,g_3 \in G
    \end{equation}
    A sufficient condition for $f_{CCZ}$ from \cref{def:product_ansatz} to correspond to a code-space preserving circuit is for all $f_i^j$ to satisfy the generalized Leibniz rule.
\end{prop}
\begin{proof}
    The proof is structurally identical to that of Lemma 5.1 of Ref. \cite{breuckmann2024cups}. In particular, one can show that the $f_{CCZ}$ function inherits a similarly defined generalized Leibniz rule from the component $f_i^j$ functions due to the product structure of \cref{eq:f_ccz_nlr}. The generalized Leibniz rule is then in turn a sufficient condition for $f_{CCZ}$ to be well-defined on the first cohomology space of the complex \cref{eq:tricycle_complex}.
\end{proof}

The NLR method can then be summarized as:

\begin{enumerate}
    \item Construct a basis-set of possible group-equivariant trilinear functions $f_i^j$ that satisfy \cref{eq:generalized_leibniz} on fixed classical codes defined by $\alpha_i\in \FG$ for $i=1,2,3$.
    \item Search randomly over candidates $f_i^j$: for each set of candidates, construct $f_{CCZ}$ as in \cref{eq:f_ccz_nlr} and iterate until the maximum degree of $f_{CCZ}$ is low.
\end{enumerate}

Step $1$ can be formulated efficiently. In particular, for each $i=1,2,3$, the set of functions $f_i^1, f_i^2, f_i^3$ are codewords of a classical binary parity check code with parity check matrix $H_{leibniz}$ of shape $|G|^2 \times 3|G|^2$. Each row of $H_{leibniz}$ corresponds to one constraint of the form \cref{eq:generalized_leibniz} -- there are $|G|^2$ such constraints instead of $|G|^3$ due to the group equivariance property of the $f_i^j$ functions. The columns of $H_{leibniz}$ have weight $3|\alpha_i|$, corresponding to the non-zero the constraint \cref{eq:generalized_leibniz}. Consequently, a basis of possible functions $f_i^j$ that satisfy the generalized Leibniz rule can be found efficiently by finding a basis for the nullspace of $H_{leibniz}$.

Step $2$. can be optimized with various heuristics. In practice, we found that the shortest depth circuits were found by generating a large set of low-weight codewords of $H_{leibniz}$ using a variant of the algorithm described in \cref{sec::distance_finding} and randomly searching over such a set. This is the method used to find the circuits for the $4-4-4$ type NLR codes in \cref{table::all_codes}. We numerically verified that the resulting circuits indeed preserve the code-space, by explicitly checking the coboundary invariance condition. In general, we find that sparsity of the functions (low weight of the codewords) leads to lower-degree $f_{CCZ}$ functions. Other strategies using simulated annealing or structured searches over a basis of codewords for $H_{leibniz}$ could potentially lead to finding lower degree circuits for a particular code. In particular, it is currently unclear if this method can efficiently yield circuits for certain codes with depths $\leq 20$ to be feasible for practical magic-state generation.

One caveat with this method is that it aims to construct a code-space preserving circuit for a fixed code. Consequently, for certain codes, there may be no such circuits with low degree. Therefore, a code-search must be performed in conjunction with this method to find good codes with short-depth $CCZ$ circuits. We will study this method and its application to other codes in more detail and generality in future work.

\section{Logical circuit optimization}\label{sec::logical_opt}

In this section we discuss how the logical $\overline{CCZ}$ circuits that produce the hypergraph magic states described in \cref{sec::transversal_CCZ} can be optimized to maximize the yield of disjoint $\overline{CCZ}$ magic-gates. 

The logical connectivity is determined by the cohomology invariant $f_{CCZ}$ function restricted to a basis of logical $X$-type operators for the three code-blocks. In cohomological terms, this is a basis for the first cohomology group $H^1(C)$ of the associated cochain complex. Denote a choice of these bases as $\{[l^1_i]\}_{i=1\cdots K}, \{[l^2_j]\}_{j=1\cdots K}, \{[l^3_k]\}_{k=1\cdots K}$ where the $l_i$ are representative elements of $\mathrm{ker}(H_Z)$ and $[\cdot]$ denotes an equivalence class modulo elements of $\mathrm{im}(H_Z^T)$. We then construct a $K\times K\times K$ binary $3$-tensor corresponding to the restriction of $f_{CCZ}$ to the logical operators:

\begin{equation}\label{eq:logical_tensor}
    T^{log}_{ijk} = f_{CCZ}(l^1_i, l^2_j,l^3_k) \in \F
\end{equation}

The tensor $T^{log}$ encodes all the structure of the resulting logical circuit:  $T^{log}_{ijk} = 1$ implies a $\overline{CCZ}$ gate between the $i$-th logical qubit of the first code block, $j$-th logical qubit of the second block, and $k$-th logical qubit of the third block. Naturally, changing the logical basis of each block will change the connectivity of the resulting circuit and the associated hypergraph magic state, corresponding to a change of basis on each leg of the $T^{log}$ tensor. We would like to be able to extract some number of disjoint $\overline{CCZ}$ gates from the circuit defined by $f_{CCZ}$ by choosing an appropriate basis of logical $X$ operators for each code block. The maximum such number of disjoint $\overline{CCZ}$, $K_{CCZ}$, that are extractable from the circuit is equal to the \textit{subrank} of $T^{log}$. Intuitively, the subrank of a binary tensor measures the size of the largest identity-like diagonal subtensor that can be embedded within the tensor \cite{christandl2023gap, kopparty2020geometric}. More formally, the subrank of the tensor $T^{log}$ is the largest integer $r\geq 0$ such that there exist matrices $M^1, M^2, M^3 \in \F^{r \times K}$ that satisfy 

\begin{equation}\label{eq:subrank}
    (M^1 \otimes M^2 \otimes M^3)\cdot T^{log} = \mathbb{I}^{r\times r\times r} \ \mathrm{mod} \ 2
\end{equation}

where $\mathbb{I}^{r\times r\times r}$ is the $r\times r \times r$ tensor with $1$s on the diagonal and $0$ everywhere else. Equivalently, we can express this in components as 

\begin{equation}\label{eq:logical_tensor_comp}
    \sum_{i,j,k} T^{log}_{ijk}M^1_{a,i} M^2_{b,j} M^3_{c,k} = \mathbbm{1}_{a=b=c} \ \mathrm{mod}(2) \quad a,b,c\in [1\cdots r]
\end{equation}

The $M$ matrices can be viewed as change of basis matrices such that rows of $M$ define new logical operators in a different basis that produces $r$ disjoint $\overline{CCZ}$ gates. In particular, the optimal choice of $r$ logical operators for the $I$-th code block is given by $\{\sum_j  M^I_{a,j} (l^I_a)_j\}_{a=1\cdots r}$. The remaining $K-r$ logicals are found by choosing any set of logical operators that are linearly independent from the first $r$. These $K-r$ logicals are treated as gauge logical qubits, to be initialized in the $\ket{\overline{0}}$ state -- see the discussion at the end of \cref{sec::transversal_CCZ}.

Unfortunately, finding the subrank $r$ and the associated bases $M^1, M^2, M^3$ is a problem that is widely believed to be NP-Hard, even over the real numbers \cite{kopparty2020geometric}. In this work, we find sub-optimal solutions using a Mixed-Integer-Programming (MIP) method -- we use the Gurobi Python package for this task \cite{gurobi}. We first find a set of initial logical operators numerically using row-reduction over $\F$. For fixed $r$ we then treat the $3rK$ entries of $M^1, M^2, M^3$ as binary variables for the integer program, and enforce the condition in \cref{eq:logical_tensor_comp} as a set of integer constraints which involve another $2r^3$ many slack variables. The complexity of the integer program therefore grows dramatically with $r$. We iterate from $r=0$ to increasing values of $r$ until the MIP solver converges or is unable to find a feasible solution within time constraints. This leads to the reported values of $K_{CCZ}$ in \cref{sec::transversal_CCZ}. An important direction for future work is to find tailored optimization heuristics for the binary tensor subrank problem to extract larger values of $K_{CCZ}$, signifcantly improving the throughput of $\overline{CCZ}$ type magic states from the hypergraph magic states produced by our circuits.

\section{Numerical details of noise simulations}\label{sec::numerics-methods}

Here, we provide further numerical simulations of the tricycle codes to demonstrate robust single-shot error correction in the $Z$ basis and assess fault-tolerant behavior in the $X$ basis using repeated rounds of syndrome extraction~\cite{bravyi2024high}.

For both types of simulation, we simulate syndrome extraction circuits for both $X$ and $Z$ checks of the tricycle codes. Each simulation begins by initializing the data qubits in the appropriate basis, followed by multiple cycles of syndrome extraction. At each cycle, fresh ancilla qubits are initialized, entangled with the data via two-qubit gates, and subsequently measured to extract the stabilizer values.

For both single-shot $Z$-basis and $d$-round $X$-basis simulations, we employ a coloration circuit for syndrome extraction together with the standard two-qubit depolarizing circuit-noise error model also used in the main text simulations.

For the single-shot scenario, we utilize a windowed decoding protocol, wherein each window comprises a fixed number of syndrome extraction rounds followed by decoding and correction. This window is then repeated multiple times, providing a concrete probe of the code’s sustainable logical error suppression~\cite{brown2016fault}. Specifically, our simulations employ windows of three rounds, each repeated fourteen times for a total of forty-two rounds---a protocol previously demonstrated to saturate sustainable performance in high-rate codes~\cite{xu2024constant}. With fixed window size, only codes possessing genuine single-shot capability exhibit logical error suppression that scales with code distance; in contrast, codes lacking such features see their performance limited by the window size. Although a window size of one is theoretically sufficient, empirical studies indicate that size three yields enhanced robustness to measurement errors and generally improved logical error rates~\cite{xu2024constant, ataides2025constant}. For conventional $d$-round simulations in the $X$ basis, the syndrome extraction is repeated for $d$ cycles, and decoding is performed using the entire record of accumulated syndrome information.

Here, we perform decoding using a belief-propagation with ordered statistics decoding (BP+OSD) scheme as presented in Ref.~\cite{panteleev2021degenerate}. To enable circuit-level decoding, we map each syndrome extraction circuit to a spacetime code~\cite{xu2024constant, ataides2025constant}: checks of the spacetime code correspond to parities of stabilizer measurements across consecutive rounds, and the code qubits represent distinct fault locations (e.g., a two-qubit gate error at a specific location in space and time). Formally, for each round, the value of a stabilizer at time $t$ is $m_t \in {0,1}$; the detector value is defined as $d_t \coloneq m_t + m_{t-1}$ (mod $2$). In single-shot simulations, belief propagation is applied to each window, except for the final readout which uses BP+OSD; for $d$-round simulations, BP+OSD is used for decoding the full syndrome record.

All circuits were constructed using Stim~\cite{gidney2021stim}, and BP+OSD decoding was implemented using the tools from Ref.~\cite{roffe2020decoding, Roffe_LDPC_Python_tools_2022}. For the $X$-basis $d$-rounds simulations, we apply a BP+OSD decoder to the full spacetime code with 10000 maximum iterations of min-sum BP and min-sum scaling factor of 0. We used OSD-CS of order 10. For the $Z$-basis single-shot simulations, we decode each window of three cycles using a BP decoder with maximum iterations set to $\#$ qubits / 5 and min-sum scaling factor of 0.9. To decode the final transversal measurement we use BP+OSD with maximum iterations set to $\#$ qubits / 5, min-sum scaling factor of 0.9 and OSD-E of order 10. Although we did not attempt extensive hyperparameter optimization, preliminary trials suggest that targeted tuning of BP+OSD parameters could produce significant gains in logical error rate.

We simulate various $4-4-4$ tricycle code and our simulation results are summarized in \cref{fig:results}. Logical error rate is defined as the probability that any logical qubit within the code block fails under the simulated noise model. We observe a gate-noise threshold greater than $0.4\%$ for each of the $X$ and $Z$ basis results (see the discussion after \cref{thm::dist_lower_bnd} in for further details). By fitting the measured logical error rates to an exponential decay model, $p_L = \alpha (p/p_{\text{th}})^{\beta n^{\gamma}}$, we estimate the achievable error rates for both $X$ and $Z$ bases. \cref{table::simulations} presents logical error rates for two key physical error benchmarks: $p_{2q} = 10^{-3}$, representative of near-term hardware, and $10^{-4}$, representative of projected advances, together with relevant code parameters.

We observe robust single-shot performance across several tricycle code instances. $d$-round performance is weaker, reflecting the higher check weight and reduced distance in the $X$ basis.

For completeness, we also simulate the single-shot performance of the $4-2-2$ tricycle codes explored in the main text. We used the same BP+OSD decoding setup described above. The results can be found in~\cref{fig:single-shot-422}, which demonstrate impressive single-shot error correction suppression for this class of codes. In practice, however, we anticipate that the overall magic state factory performance will be limited by the $d$-rounds $X$-basis simulated in the main text, owing to its higher check weight and lower distance compared to the $Z$ basis.

\begin{figure}[h!]
    \centering
    \includegraphics[width=\textwidth]{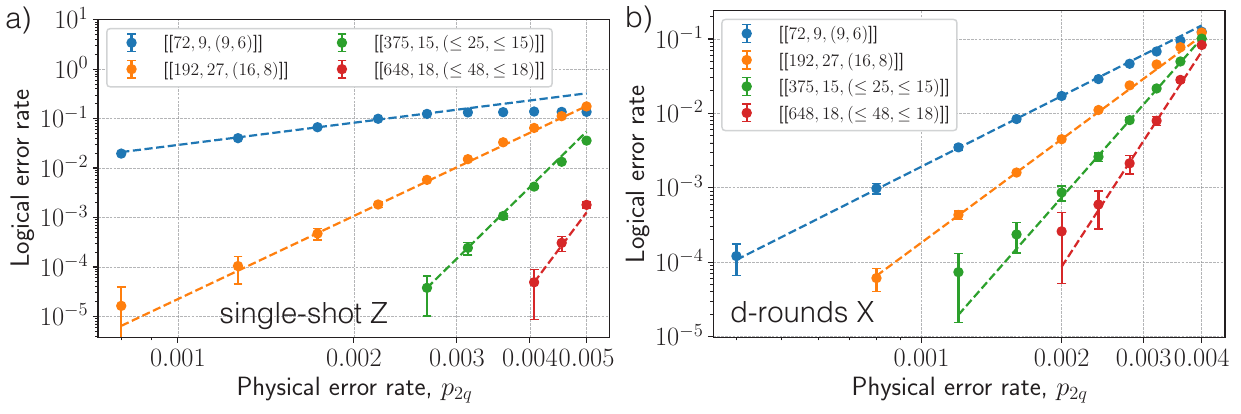}
    \caption{\textbf{Circuit-level noise simulation results for $4-4-4$ tricycle codes.}
(a) Logical error rate (as measured by block failure probability) versus two-qubit physical gate error rate ($p_{2q}$) for single-shot error correction in the $Z$ basis. Single-shot performance is evaluated using a windowed decoding protocol: three rounds of syndrome extraction followed by decoding and correction, with the window repeated 14 times (for 42 total rounds) to probe sustainable suppression of logical errors. (b) Logical error rate versus $p_{2q}$ for fault-tolerant, $d$-round error correction in the $X$ basis. In both panels, errors are sampled under a standard two-qubit depolarizing circuit-level noise model, and results are shown for various tricycle code. Logical error rates are determined via Monte Carlo simulations using a BP+OSD decoder, with each data point corresponding to $M$ samples; error bars indicate standard errors, computed as $\sqrt{p_L(1-p_L)/M}$.}
    \label{fig:results}
\end{figure}

\begin{table}[h!]
\centering
\begin{tabular}{lllll} % 5 columns, all left-aligned; use c or r for center/right alignment
\toprule
\toprule
$\mathbf{[[N,K,(D_X,D_Z)]]}$ & $\mathbf{p_L^{(X)}(0.1\%)}$ & $\mathbf{p_L^{(Z)}(0.1\%)}$ & $\mathbf{p_L^{(X)}(0.01\%)}$ & $\mathbf{p_L^{(Z)}(0.01\%)}$ \\
\midrule
$[[72,9,(9,6)]]$ & $2\times 10^{-3}$ & $3\times 10^{-2}$ & $1\times 10^{-6}$ & $9 \times 10^{-4}$ \\
$[[192,27,(16,8)]]$ & $2\times 10^{-4}$ & $2 \times 10^{-5}$ & $4 \times 10^{-9}$ & $6 \times 10^{-11}$ \\
$[[375,15,(\leq 25, \leq15)]]$ & $5\times 10^{-6}$ & $4\times 10^{-10}$ & $4\times 10^{-13}$ & $1\times 10^{-21}$  \\
$[[648,18,(\leq 48, \leq18)]]$ & $1 \times 10^{-7}$ & $1\times 10^{-14}$ & $3 \times 10^{-17}$ & $2\times 10^{-30}$  \\
% Add more rows as needed
\bottomrule
\bottomrule
\end{tabular}
\caption{\textbf{Example $4-4-4$ tricycle codes and their logical error rates under circuit-level noise.}
Logical error rates $p_L^{(X)}$ and $p_L^{(Z)}$ are shown for both $p_{2q} = 10^{-3}$ and $p_{2q} = 10^{-4}$, for different [[N, K, D]] code parameters.}
\label{table::simulations}
\end{table}

\begin{figure}[h!]
    \centering
    \includegraphics[width=0.5\textwidth]{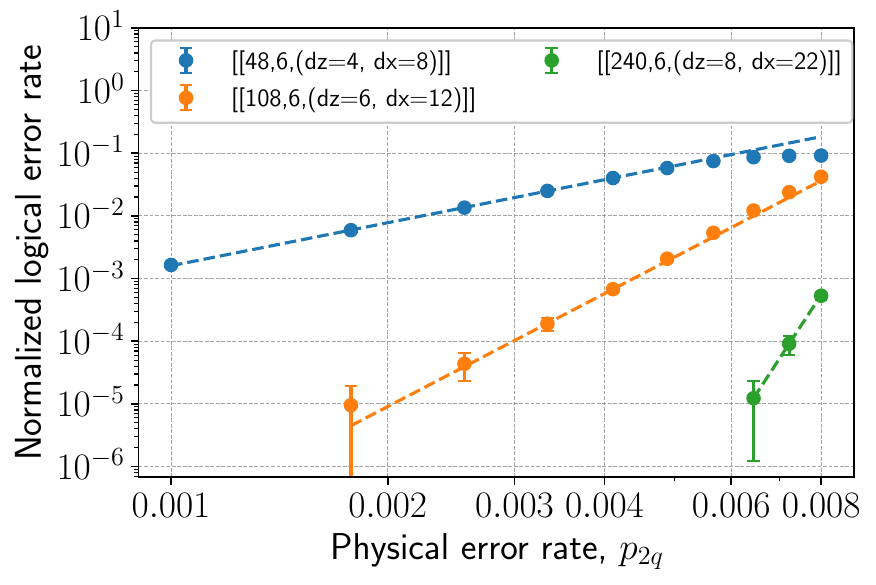}
    \caption{\textbf{Single-shot simulation results for $4-2-2$ tricycle codes.} Logical error rate normalized by number of QEC rounds and number of logical qubits versus two-qubit physical gate error rate ($p_{2q}$) for single-shot error correction in the $Z$ basis. Logical error rates are determined via Monte Carlo simulations using a BP+OSD decoder, with each data point corresponding to $M$ samples; error bars indicate standard errors, computed as $\sqrt{p_L(1-p_L)/M}$.}
    \label{fig:single-shot-422}
\end{figure}

\section{Optimal-depth syndrome extraction circuits}\label{sec::tricycle_scheduling_method}

We provide a constructive proof of the following theorem.
\begin{theorem}\label{thm::tricycle_scheduling}
For tricycle codes defined by $\mathbf{A}=\sum_{i=1}^{w_a}\mathbf{A}_i$, $\mathbf{B}=\sum_{i=1}^{w_b}\mathbf{B}_i$, and $\mathbf{C}=\sum_{i=1}^{w_c}\mathbf{C}_i$ where $w_a$, $w_b$, and $w_c$ are  even, there exists syndrome extraction circuits with CNOT depth $w_a+w_b+w_c$.
\end{theorem}
\begin{proof}
Let $X \to \mathbf{A}_i^T(D_1)$ denote CNOTs from check sector $X$ to data sector $D_1$ according to permutation $\mathbf{A}_i^T$.
The arrow indicates the CNOT direction: from each qubit $j$ in $X$ to qubit $k$ in $D_1$ whenever the $(j,k)$ entry of $\mathbf{A}_i^T$ is 1.
Since $\mathbf{A}_i^T$ are permutation matrices, the CNOTs are disjoint and can be executed in parallel.

We schedule all CNOTs in five stages, with $w_a + w_b + w_c$ layers in total (square brackets indicate groups of CNOTs in the same layer):
\begin{equation}\label{eq::tricycle_scheduling}\begin{split}
& \text{(1) for } 1\le i \le w_c/2 \text{, apply } [\mathbf{C}_{i+w_c/2}(D_1)\to Z_1,\ \mathbf{C}_{i+w_c/2}(D_2) \to Z_2,\ X\to  \mathbf{C}_i^T(D_3)]; \\
& \text{(2) for } 1\le i \le w_b/2 \text{, apply }  [\mathbf{B}_{i+w_b/2}(D_1)\to Z_3,\ X\to \mathbf{B}_i^T(D_2),\ \mathbf{B}_{i+w_b/2}(D_3)\to Z_2];\\
& \text{(3) for } 1\le i \le w_a \text{, apply } [X \to \mathbf{A}^T_{i}(D_1),\ \mathbf{A}_i(D_2)\to Z_3,\ \mathbf{A}_{i}(D_3)\to Z_1]; \\
& \text{(4) for } 1\le i \le w_b/2 \text{, apply } [\mathbf{B}_{i}(D_1)\to Z_3,\ X \to \mathbf{B}_{i+w_b/2}^T(D_2),\ \mathbf{B}_{i}(D_3)\to Z_2]; \\
& \text{(5) for } 1\le i \le w_c/2 \text{, apply } [\mathbf{C}_{i}(D_1)\to Z_1,\ \mathbf{C}_{i}(D_2) \to Z_2,\ X \to \mathbf{C}_{i+w_c/2}^T(D_3)]. \\
\end{split} \end{equation}

It remains to verify that this schedule measures the stabilizers while preserving logical operators. 
The CNOTs' effects can be tracked using a stabilizer table; since CNOTs do not mix X and Z parts, the action on each part can be analyzed separately. We focus on the X-part; the Z-part is similar.

Consider the stabilizer table extended with an arbitrary logical Pauli X, supported on 1-entries of row vectors $u$, $v$, and $w$.
The circuit is correct if the CNOTs yield the expected outcome:
\begin{equation} \label{eq::tricycle_scheduling_condition}
\begin{array}{c}
\begin{array}{ccccccc} 
X\ &D_1 &D_2 &D_3 &Z_1 &Z_2 &Z_3  \end{array} \\
\left[\begin{array}{c|ccc|ccc} 
\ \mathbf{I}\ &\ 0\ &\ 0\   &\ 0\   &\ 0\   &\ 0\   &\ 0\   \\
0 &\mathbf{A}^T &\mathbf{B}^T &\mathbf{C}^T &0   &0   &0 \\
0 &u   &v   &w   &0   &0   &0
\end{array}\right]
\end{array}
\quad \to \quad
\begin{array}{c}
\begin{array}{ccccccc} 
X\ &D_1 &D_2 &D_3 &Z_1 &Z_2 &Z_3  \end{array} \\
\left[\begin{array}{c|ccc|ccc} 
\ \mathbf{I}\ &\mathbf{A}^T &\mathbf{B}^T &\mathbf{C}^T   &\ 0\   &\ 0\   &\ 0\   \\
\ \mathbf{I}\ &\mathbf{A}^T &\mathbf{B}^T &\mathbf{C}^T &0   &0   &0 \\
0 &u   &v   &w   &0   &0   &0
\end{array}\right]
\end{array}
\end{equation}

At stage 1, the CNOTs from $X$ to $D_3$ add rows from column $X$ to those in column $D_3$, permuted by each $\mathbf{C}^T_i$.
Let $\mathbf{C}^T_< := \sum_{i=1}^{w_c/2}\mathbf{C}^T_i$ and $\mathbf{C}^T_> := \sum_{i=1+w_c/2}^{w_c}\mathbf{C}^T_i$;
these CNOTs correspond to right-multiplying column $X$ by $\mathbf{C}^T_<$ and accumulating into $D_3$.
The other two CNOT groups propagate from data sectors to Z-checks.
As each $\mathbf{C}_i$ is a permutation matrix, $\mathbf{C}_i(D_1) \to Z_1$ is equivalent to $D_1 \to \mathbf{C}_i^T(Z_1)$;
we use this equivalence below.
Thus, the CNOT action right-multiplies $D_1$ and $D_2$ by $\mathbf{C}^T_>$ and accumulates the results to $Z_1$ and $Z_2$, respectively.
After stage 1, the table becomes
\begin{equation*}\left[
\begin{array}{c|ccc|ccc} 
\mathbf{I} &0   &0   &\mathbf{C}^T_<   &0   &0   &0 \\
0 &\mathbf{A}^T &\mathbf{B}^T &\mathbf{C}^T &\mathbf{A}^T\mathbf{C}^T_>   &\mathbf{B}^T\mathbf{C}^T_>   &0 \\
0 &u   &v   &w   &u\mathbf{C}^T_>   &v\mathbf{C}^T_>   &0
\end{array}\right].
\end{equation*}

After stage 2, the table becomes
\begin{equation*}\left[
\begin{array}{c|ccc|ccc} 
\mathbf{I} &0   &\mathbf{B}^T_<   &\mathbf{C}^T_<   &0   &\mathbf{C}^T_<\mathbf{B}^T_>=\mathbf{B}^T_>\mathbf{C}^T_<   &0 \\
0 &\mathbf{A}^T &\mathbf{B}^T &\mathbf{C}^T &\mathbf{A}^T\mathbf{C}^T_>   &\mathbf{B}^T\mathbf{C}^T_> + \mathbf{C}^T\mathbf{B}^T_> = \mathbf{B}^T\mathbf{C}^T_> + \mathbf{B}^T_>\mathbf{C}^T   &\mathbf{A}^T\mathbf{B}^T_> \\
0 &u   &v   &w   &u\mathbf{C}^T_>   &v\mathbf{C}^T_> + w\mathbf{B}^T_>   &u\mathbf{B}^T_>
\end{array}\right],
\end{equation*}
where the equalities hold since all polynomial terms in $\mathbf{A}$, $\mathbf{B}$, and $\mathbf{C}$ commute.
We will use this property throughout and consistently write products of terms in dictionary order.

After stage 3, the table becomes
\begin{equation*}\left[
\begin{array}{c|ccc|ccc} 
\mathbf{I} &\mathbf{A}^T   &\mathbf{B}^T_<   &\mathbf{C}^T_<   &\mathbf{A}^T \mathbf{C}^T_<  &\mathbf{B}^T_>\mathbf{C}^T_<   &\mathbf{A}^T\mathbf{B}^T_< \\
0 &\mathbf{A}^T &\mathbf{B}^T &\mathbf{C}^T &\mathbf{A}^T\mathbf{C}^T_> + \mathbf{A}^T\mathbf{C}^T  &\mathbf{B}^T\mathbf{C}^T_> + \mathbf{B}^T_>\mathbf{C}^T   &\mathbf{A}^T\mathbf{B}^T_> + \mathbf{A}^T\mathbf{B}^T \\
0 &u   &v   &w   &u\mathbf{C}^T_>+w\mathbf{A}^T   &v\mathbf{C}^T_> + w\mathbf{B}^T_>   &u\mathbf{B}^T_> + v\mathbf{A}^T
\end{array}\right].
\end{equation*}

After stage 4, the table becomes
\begin{equation*}\left[
\begin{array}{c|ccc|ccc} 
\mathbf{I} &\mathbf{A}^T   &\mathbf{B}^T   &\mathbf{C}^T_<   &\mathbf{A}^T \mathbf{C}^T_<  &\mathbf{B}^T\mathbf{C}^T_<   &\mathbf{A}^T\mathbf{B}^T_< +\mathbf{A}^T\mathbf{B}^T_<=0 \\
0 &\mathbf{A}^T &\mathbf{B}^T &\mathbf{C}^T &\mathbf{A}^T\mathbf{C}^T_> + \mathbf{A}^T\mathbf{C}^T  &\mathbf{B}^T\mathbf{C}^T_> + \mathbf{B}^T\mathbf{C}^T   &\mathbf{A}^T\mathbf{B}^T + \mathbf{A}^T\mathbf{B}^T = 0 \\
0 &u   &v   &w   &u\mathbf{C}^T_>+w\mathbf{A}^T   &v\mathbf{C}^T_> + w\mathbf{B}^T   &u\mathbf{B}^T + v\mathbf{A}^T
\end{array}\right],
\end{equation*}
where the equalities hold since the addition is modulo 2.

Finally, after stage 5, the table becomes
\begin{equation*}\left[
\begin{array}{c|ccc|ccc} 
\mathbf{I} &\mathbf{A}^T   &\mathbf{B}^T   &\mathbf{C}^T   &0 &0  &0 \\
0 &\mathbf{A}^T &\mathbf{B}^T &\mathbf{C}^T &0  &0  &0 \\
0 &u   &v   &w   &u\mathbf{C}^T+w\mathbf{A}^T=0   &v\mathbf{C}^T + w\mathbf{B}^T=0   &u\mathbf{B}^T + v\mathbf{A}^T=0
\end{array}\right],
\end{equation*}
where the equalities hold since any logical Pauli X commutes with Z-checks, i.e., $[u\ v\ w]H_z^T=0$.
\end{proof}

The construction allows certain degrees of freedom. 
First, the assignment of the matrices $\mathbf{A}$, $\mathbf{B}$, and $\mathbf{C}$ to stages 1+5, 2+4, and 3 can be chosen arbitrarily. 
Second, the ordering of layers within each stage is flexible, provided that stages 1 and 5 contain complementary terms, as do stages 2 and 4.
Let $w_1$, $w_2$, and $w_3$ denote the weights of polynomials assigned to stages 1+5, 2+4, and 3, respectively. (In the proof above, $w_1 = w_c$, $w_2 = w_b$, and $w_3 = w_a$.)
Then, the layers within stage 3 can be ordered $w_3!$ ways.
The order of terms in stages 1 and 5 for the X-checks can be arbitrary, resulting in $w_1!$ possibilities. Since stages 1 and 5 must be complementary for the Z-checks, the assignment of terms to stages 1 and 5 is fixed, but permutations within each stage remain possible, contributing an additional $((w_1/2)!)^2$ arrangements.
Similarly, for stages 2 and 4, the same counting applies with weight $w_2$.
Thus, the total number of distinct circuit configurations is
\begin{equation}
w_1! \times \bigl((w_1/2)!\bigr)^2 \times w_2! \times \bigl((w_2/2)!\bigr)^2 \times w_3!.
\end{equation}
For example, when $w_1 = w_2 = 2$ and $w_3 = 4$, there are $96$ circuits, while for $w_1 = w_3 = 2$ and $w_2 = 4$, there are $384$ circuits.
In addition, the assignments of $\mathbf{A}$, $\mathbf{B}$, and $\mathbf{C}$ to the stages 1+5, 2+4, and 3 themselves can be permuted.
For the $4-2-2$ codes, this yields: two assignments with $w_1 = w_2 = 2, w_3 = 4$, two assignments with $w_1 = w_3 = 2, w_2 = 4$, and two assignments with $w_2 = w_3 = 2, w_1 = 4$.
Altogether, this gives $96 \times 2 + 384 \times 4 = $ 1,728 different circuits for the $4-2-2$ codes.
A similar calculation for the $4-4-2$ and $4-4-4$ codes results in 55,296 and 1,327,104 distinct circuits, respectively.

In the counting above, same polynomial terms are scheduled together if possible, to enable more efficient atom rearrangement (more details in \cref{sec::detail_atom}).
For example, in stage 3, $\mathbf{A}_1$ from $X$ to $D_1$, from $D_2$ to $Z_3$, and from $D_3$ to $Z_1$ are always in the same layer, but they do not have to be.
If we consider such possibilities, the number of different choices for stage 3 is no longer $w_3!$ but $(w_3!)^3$. 
For a specific choice of the weights, the number of different circuits becomes
\begin{equation}
w_1! \times \bigl((w_1/2)!\bigr)^4 \times w_2! \times \bigl((w_2/2)!\bigr)^4 \times (w_3!)^3.
\end{equation}
For $4-2-2$ codes, this amounts to 135,168 circuits.

\cref{thm::tricycle_scheduling}  applies to all codes presented in the main text.
For other cases, let $W = w_a + w_b + w_c$ and $V = \left|\{w \in \{w_a, w_b, w_c\} \mid w \text{ is even}\}\right|$. Our scheduling method produces the following result, which is at most two CNOT layers from optimal.

\begin{corollary}\label{thm::tricycle_schedule_allcase}
For tricycle codes with $V=0$, $1$, $2$, and $3$, there exists syndrome extraction circuits with CNOT depth $W$, $W$, $W+1$, and $W+2$, respectively.
\end{corollary}
\begin{proof}
\cref{thm::tricycle_scheduling} covers $V=0$.
For $V=1$, assign the odd-weight polynomial as $\mathbf{A}$, and the previous proof still applies.
For $V=2$, assign the odd-weight polynomials as $\mathbf{A}$ and $\mathbf{C}$.
Modifying stages 1 and 5 of in \cref{eq::tricycle_scheduling} as follows suffices:
\begin{equation}\label{eq::tricycle_scheduling_odd}\begin{split}
\text{(1) } & \text{apply } X\to  \mathbf{C}_1^T(D_3);\\
&  \text{for } 1\le i \le \lfloor w_c/2 \rfloor \text{, apply } [\mathbf{C}_{i+\lfloor w_c/2 \rfloor}(D_1)\to Z_1,\ \mathbf{C}_{i+\lfloor w_c/2 \rfloor}(D_2) \to Z_2,\ X\to  \mathbf{C}_{1+i}^T(D_3)]; \\
\text{(5) } & \text{for } 1\le i \le \lfloor w_c/2 \rfloor \text{, apply } [\mathbf{C}_{i}(D_1)\to Z_1,\ \mathbf{C}_{i}(D_2) \to Z_2,\ X\to  \mathbf{C}_{1+i+\lfloor w_c/2 \rfloor}^T(D_3)]; \\
& \text{apply } [\mathbf{C}_{w_c}(D_1)\to Z_1,\ \mathbf{C}_{w_c}(D_2) \to Z_2].\\
\end{split} \end{equation}
Thus, stages 1 and 5 now take $w_c+1$ layers.
For $V=3$, modify stages 2 and 4 in a similar way.
\end{proof}

\begin{figure}[t]
\centering
\includegraphics{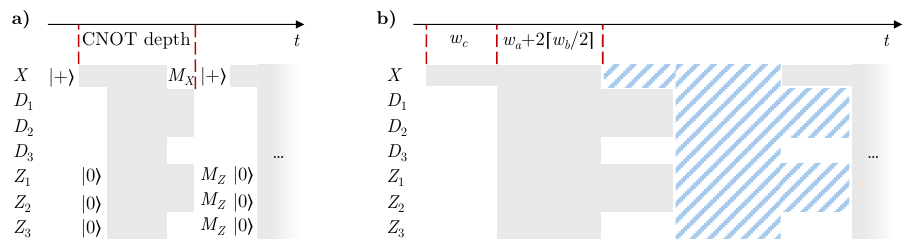}
\caption{\label{fig::tricycle_staggered_schedule}\textbf{Syndrome extraction circuit for multiple cycles:}
(a) When $w_c$ is odd, X-check measurements can overlap with the last CNOT layer, and Z-check initialization can overlap with the first CNOT layer.
(b) Staggered scheduling with separate check qubits for even and odd cycles.
Solid and hatched patterns indicate CNOTs for even and odd cycles, respectively.
Measurement and initialization for one cycle can overlap with the CNOTs of the other cycle.}
\end{figure}

Following Ref.~\cite{bravyi2024high}, counting both initialization and measurement of check qubits as unit depth, our construction gives the following results.

\begin{corollary}\label{thm::tricycle_scheduling_full_depth}
For tricycle codes with $V=0$, $1$, $2$, and $3$, there exists $M$-cycle syndrome extraction circuits with depth $(W+2)M$, $(W+2)M$, $(W+2)M+1$, and $(W+3)M+1$, respectively.
\end{corollary}
\begin{proof}
For $V=0$ and $V=1$, initialization precedes and measurement follows the CNOT layers each cycle, adding two units per cycle.
For $V=2$ and $V=3$, as illustrated in \cref{fig::tricycle_staggered_schedule}a, X-check measurement can be parallelized with the final CNOT layer, and Z-check initialization with the first layer in \cref{eq::tricycle_scheduling_odd}, resulting in the stated depths.
\end{proof}

With two sets of check qubits, we use one set for odd cycles and the other set for even cycles.

\begin{theorem}\label{thm::tricycle_scheduling_multi_cycle}
Given two sets of check qubits, for tricycle codes with $V=0$, $1$, $2$, and $3$, there exists $M$-cycle syndrome extraction circuits with CNOT depth $WM+w_c$, $WM+w_c$, $WM+w_c$, and $(W+1)M+w_c$, respectively.
\end{theorem}
\begin{proof}
Adapt stages 1 and 5 in \cref{eq::tricycle_scheduling} as follows:
\begin{equation}\begin{split}
& \text{(1) for } 1\le i \le w_c \text{, apply } [\mathbf{C}_{i}(D_1)\to Z'_1,\ \mathbf{C}_{i}(D_2) \to Z'_2,\ X\to  \mathbf{C}_i^T(D_3)]; \\
& \text{(5) for } 1\le i \le w_c \text{, apply } [\mathbf{C}_{i}(D_1)\to Z_1,\ \mathbf{C}_{i}(D_2) \to Z_2,\ X' \to \mathbf{C}_{i}^T(D_3)]; \\
\end{split} \end{equation}
where the prime indicates check qubits for the other cycle.
This yields a staggered circuit (\cref{fig::tricycle_staggered_schedule}b) similar to the construction for hypergraph product codes in Ref.~\cite{xu2024constant}.
If at least one polynomial has even weight, assign it as $\mathbf{B}$;
if $V=3$, revise stages 2 and 4 similarly to \cref{eq::tricycle_scheduling_odd}.
\end{proof}

The above results naturally adapt to bicycle codes.
For example, when $w_a = w_b = 3$, the circuits for bivariate bicycle codes in Ref.~\cite{bravyi2024high} are reproduced.

\section{Neutral atom array architecture and rearrangement procedure}\label{sec::detail_atom}

To concretely specify the atom rearrangement procedure, we consider the reference architecture shown in \cref{fig::arch_detail}, primarily based on settings from recent experiments~\cite{bluvstein2024logical, bluvstein2025architectural, chiu2025continuousoperationcoherent3000qubit}.
The neutral atom array features multiple zones for different functions: an entangling zone for CNOT operations, a storage zone for densely packed atoms, a measurement zone for imaging, and a reservoir for fresh atoms.
For our purposes, we focus on the entangling and storage zones.

SLM traps are present throughout these zones, and atoms can be rearranged by AOD between SLM grids, subject to AOD order constraints. 
Parts of the storage zone adjacent to the entangling zone serve as workspaces, where sector permutations occur before each CNOT layer. 
We further divide the workspace into regions $\ts_i$, each supporting permutation of a sector, and each containing two SLM trap arrays: $\ts_{i-}$ (dots) and $\ts_{i+}$ (circles), each shaped like a sector (e.g., $n_y$ rows by $n_x n_z$ columns as in \cref{fig::implementation}b). 
These arrays are just labels; physically, all SLM traps are generated together.

For example, a $y$-cycling permutation of $1\le m < n_y$ units for a sector in $\ts_{i-}$ requires two AOD movements: 
shifting rows $1,\ldots,n_y-m$ in $\ts_{i-}$ down to rows $m+1,\ldots,n_y$ in $\ts_{i+}$;
and the wrap-around, shifting rows $n_y-m+1,\ldots,n_y$ up to rows $1,\ldots,m$ in $\ts_{i+}$. 
After cycling, the sector is entirely in $\ts_{i+}$. 
Minimizing the SLM trap spacing in $\ts_i$ enables faster permutation, which limited by the optical resolving distance, $d_\text{min}$. 
The vertical spacing of $\ts_{i+}$ is set to $2d_\text{min}$, allowing qubits to traverse the region along the dashes;
and the horizontal spacing is $3d_\text{min}$ to additionally ensure sufficient separation between $\ts_{i-}$ and $\ts_{i+}$.

\begin{figure}[t]
\centering
\includegraphics{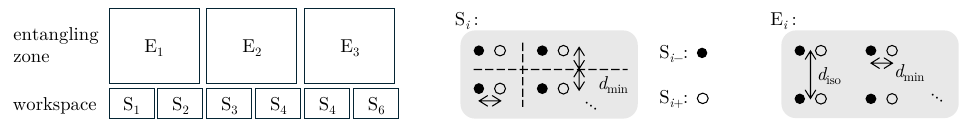}
\caption{\label{fig::arch_detail} \textbf{Reference neutral atom array architecture.}
Rydberg interactions are enabled only within the entangling zone. 
We slice the zones to regions $\te_i$ and $\ts_i$, each one can store a sector. 
Each region contains two trap arrays (dots/$-$ and circles/$+$) to facilitate sector permutation or parallel CNOTs. 
In the workspace, traps are spaced by twice the minimal distance permitted by optical resolution, $d_\text{min}$, allowing qubits to move between traps along dashed paths. 
In the entangling zone, traps are spaced so that qubit pairs involved in parallel CNOTs are separated by a distance to sufficiently isolate the Rydberg interaction, $d_\text{iso}$.}
\end{figure}

Each entangling region $\te_i$ contains two arrays, $\te_{i-}$ and $\te_{i+}$. 
To perform parallel CNOTs, one sector is positioned in $\te_{i-}$, the other in $\te_{i+}$, and a global Rydberg laser pulse is applied. 
We refer to the AOD movement from workspace to the entangling zone as fetching, and the movement back as put-back. 
To ensure that only intended pairs interact, the trap spacing in $\te_i$ is set to at least $d_\text{iso}$, a distance to sufficiently isolate Rydberg interactions.
For example, with $d_\text{min}=$\qty{2}{\micro\meter} and $d_\text{iso}=$\qty{10}{\micro\meter}~\cite{bluvstein2024logical}, each $\te_i$ is about twice as wide as each $\ts_j$.

We now explicitly describe the atom rearrangement procedure outlined in \cref{sec::implementation}, based on the reference architecture.
For initialization, place $D_1$ in $\te_1$, $D_2$ in $\te_2$, $D_3$ in $\te_3$, $Z_3$ in $\ts_2$, $Z_1$ in $\ts_3$, $Z_2$ in $\ts_4$, and $X$ in $\ts_5$.
To implement the first layer, perform:
\begin{itemize}
\item Permute $X$ (relabeling): $\bc_{1}\bi(\ts_5)$;
\item Permute $Z$ (relabeling): $\bc^T_{3} \bi [\ts_2, \ts_3, \ts_4]$;
\item Fetch $X$, $Z_1$ and $Z_2$: $[\ts_5\to\te_3, \ts_3\to\te_1, \ts_4\to\te_2]$;
\item CNOTs: $[\mathbf{C}_{3}(D_1)\to Z_1,\ \mathbf{C}_{3}(D_2) \to Z_2,\ X\to  \mathbf{C}_1^T(D_3)]$;
\item Put back $X$, $Z_1$ and $Z_2$: $[\te_3\to\ts_5, \te_1\to \ts_3, \te_2\to\ts_4]$.
\end{itemize}

Here, `permute $Z$' refers to simultaneous permutation of all three $Z$ sectors.
Permutations are denoted as $\sigma'\sigma^T(\ts_i)$, where $\sigma$ is the current permutation and $\sigma'$ the target, within region $\ts_i$. 
Specifically for this first layer, since the $X$ and $Z$ sectors are still fresh, we can relabel their qubits instead of performing physical permutations.
Fetching and put-back notation indicate initial and final regions for AOD movements, potentially involving sequential moves and grid changes. 
We do not specify which array within a region is used, as either suffices. 
CNOT notation follows previous sections, and parallelized movements are listed in square brackets. 
Notably, fetching and put-back for $X$ and $Z$ sectors can be combined for this layer.

The full implementation of the syndrome extraction circuit in \cref{eq::tricycle_scheduling} proceeds similarly. 
For convenience, the permutation notation is slightly abused: when $i=1$, in expressions like $\bc^T_{i+w_c/2} \bc_{i-1+w_c/2}[\ts_2, \ts_3, \ts_4]$, the second matrix refers to the prior permutation on those sectors (e.g., identity for stage (1)), not literally $\bc_{w_c/2}$. 
The accompanying video demonstrates an instance with $n_x=n_y=n_z=3$ and  $\ba=yz^2+x^2z+x^2y+x^2y^2z$, $\bb=yz^2+xy+x^2y+x^2y^2z$, and $\bc=1+xz^2+x^2y^2+x^2y^2z$.
\begin{enumerate}[label=(\arabic*)]
\item For $1\le i\le w_c/2$, 
permute $X$: $\bc_{i}\bc^T_{i-1}(\ts_5)$; 
permute $Z$: $\bc^T_{i+w_c/2} \bc_{i-1+w_c/2} [\ts_2, \ts_3, \ts_4]$; 
fetch $Z_1$, $Z_2$, and $X$: $[\ts_3\to\te_1, \ts_4\to\te_2, \ts_5\to\te_3]$; 
CNOTs: $[\mathbf{C}_{i+w_c/2}(D_1)\to Z_1,\ \mathbf{C}_{i+w_c/2}(D_2) \to Z_2,\ X\to  \mathbf{C}_i^T(D_3)]$; 
put back $Z_1$, $Z_2$, and $X$: $[\te_1\to \ts_3, \te_2\to\ts_4, \te_3\to\ts_5]$.
\item For $1\le i\le w_b/2$, 
permute $X$: $\bb_{i}\bb^T_{i-1}(\ts_5)$; 
permute $Z$: $\bb^T_{i+w_c/2} \bb_{i-1+w_c/2} [\ts_2, \ts_3, \ts_4]$; 
fetch $X$: $\ts_5\to\te_2$; 
fetch $Z_3$ and $Z_2$: $[\ts_2\to\te_1, \ts_4\to\te_2]$; 
CNOTs: $[\mathbf{B}_{i+w_b/2}(D_1)\to Z_3,\ X\to \mathbf{B}_i^T(D_2),\ \mathbf{B}_{i+w_b/2}(D_3)\to Z_2]$; 
put back $X$: $\te_2\to\ts_5$; 
put back $Z_1$ and $Z_2$: $[\te_1\to\ts_2, \te_2\to\ts_4]$.
\item For $ 1\le i\le w_a$, 
permute $X$: $\ba_{i}\ba^T_{i-1}(\ts_5)$; 
permute $Z$: $\ba^T_{i} \ba_{i-1} [\ts_2, \ts_3, \ts_4]$; 
fetch $X$: $\ts_5\to\te_1$; 
fetch $Z_3$ and $Z_1$: $[\ts_2\to\te_2, \ts_3\to\te_3]$; 
CNOTs: $[X \to \mathbf{A}^T_{i}(D_1),\ \mathbf{A}_i(D_2)\to Z_3,\ ,\mathbf{A}_{i}(D_3)\to Z_1]$; 
put back $X$: $\te_1\to\ts_5$; 
put back $Z_3$ and $Z_1$: $[\te_2\to\ts_2, \te_3\to\ts_3]$.
\item For $1\le i\le w_b/2$, 
permute $X$: $\bb_{i+w_b/2}\bb^T_{i-1+w_b/2}(\ts_5)$; 
permute $Z$: $\bb^T_{i} \bb_{i-1} [\ts_2, \ts_3, \ts_4]$; 
fetching and put-back are the same as (2).
\item For $1\le i\le w_c/2$, 
permute $X$: $\bc_{i+w_c/2}\bc^T_{i-1+w_c/2}(\ts_5)$; 
permute $Z$: $\bc^T_{i} \bc_{i-1} [\ts_2, \ts_3, \ts_4]$; 
fetching and put-back are the same as (1).
\end{enumerate}

For $w_a = w_b = w_c = 4$, the above procedure results in 22 sector permutations (with relabeling used in the first layer) and 40 fetch/put-back operations.
While the scheme is general and structured, inefficiencies may arise from always restoring sectors to their original locations and exclusively permuting check sectors instead of data sectors. 
For example, in (3), since $\ba_i(X) \to D_1$ is equivalent to $X \to \ba_i^T(D_1)$, permuting $D_1$ instead of $X$ may allow both $Z$ sectors and $D_1$ to be permuted in parallel by $\ba^T_i\ba_{i-1}$. 
However, it is the change in permutation, $\sigma'\sigma^T$, and not the polynomial term $\sigma$ itself, that enables such parallelization. 
If two sectors share $\sigma'$ but have different $\sigma$, their permutations still cannot be merged, so a more careful analysis is needed to assess the potential benefits.

There is also flexibility to reorder CNOT layers or combine different CNOT groups into parallel layers, which changes $\sigma'\sigma^T$ between layers. 
This sometimes results in cancellations, reducing the need for certain $x$-, $y$-, or $z$-cycling steps, and can be formulated as a variant of the traveling salesman problem and be solved optimally for small instances.

It is also valuable to investigate `torus-like layout' strategies, where sectors are interspersed not only during entangling gates but also during permutations, provided that there are polynomial terms corresponding to nearest-neighbor interactions~\cite{bravyi2024high, viszlai2024matchinggeneralizedbicyclecodesneutral}. 
Further optimization of sector layout and movement may be possible by mapping the problem to mathematical programming~\cite{tan2024dpqa, tan2022dpqa, tan_depth-optimal_2024, stade25optimal} or graph theory formulations~\cite{Tan2025, lin2024reuseawarecompilationzonedquantum, stade2025routingawareplacementzonedneutral}.

\bibliographystyle{abbrv}
\bibliography{refs}

\clearpage
\newpage

\section*{Supplemental Material}
\phantomsection
\addcontentsline{toc}{section}{Supplementary Material}
\section*{Review of homological formulation of CSS Codes}\label{app:ccs_hom}

\subsection*{Background on cohomology}

A \emph{cochain complex} \( (C^\bullet, d^\bullet) \) is a sequence of vector spaces over a field, or more generally modules over a ring, equipped with linear maps
\[
\cdots \xrightarrow{d^{i-2}} C^{i-1} \xrightarrow{d^{i-1}} C^i \xrightarrow{d^i} C^{i+1} \xrightarrow{d^{i+1}} \cdots
\]
such that
\begin{equation}
d^{i} \circ d^{i-1} = 0 \quad \text{for all } i.
\end{equation}
This condition ensures that the image of one map is always contained in the kernel of the next:
\begin{equation}
\operatorname{im}(d^{i-1}) \subseteq \ker(d^i).
\end{equation}

Elements of \( C^i \) are called \emph{\( i \)-cochains}. Elements of \( \ker(d^i) \subseteq C^i \) are called \emph{\( i \)-cocycles}, and elements of \( \operatorname{im}(d^{i-1}) \subseteq C^i \) are called \emph{\( i \)-coboundaries}. Since every coboundary is automatically a cocycle, we may define the \( i \)-th \emph{cohomology group} (or space) as the space of cocycles modulo coboundaries:
\begin{equation}
H^i(C^\bullet) := \ker(d^i) / \operatorname{im}(d^{i-1}).
\end{equation}
Cohomology detects elements in \( C^i \) that are annihilated by \( d^i \) but are not exact—i.e., not the image of any element under \( d^{i-1} \).

For context, the dual notion is that of a \emph{chain complex}, where the maps go in the opposite direction. The homology of a chain complex is defined analogously, with
\[
H_i(C_\bullet) := \ker(\partial_i) / \operatorname{im}(\partial_{i+1}),
\]

When the \( C^i \) are finite-dimensional vector spaces or free modules (modules that admit a basis), there is a natural identification \( C_i \cong (C^i)^* \), and the differentials satisfy \( d^i = (\partial_{i+1})^T \). In this case, there is a canonical isomorphism between homology and cohomology:
\begin{equation}
H^i(C^\bullet) \cong H_i(C_\bullet).
\end{equation}
This identification is particularly useful when interpreting CSS codes and their logical operators in terms of cohomology.

\subsection*{Classical codes from cochain complexes}

A classical binary linear code \( C \subset \mathbb{F}_2^n \) is a \( k \)-dimensional subspace with minimum distance \( d = \min_{c \in C} |c| \), where \( |\cdot| \) denotes the Hamming weight. The code is specified by a parity check matrix \( H \in \mathbb{F}_2^{n \times r} \) such that \( \ker(H) = C \).

This data can be represented as a 2-term cochain complex:
\begin{equation}
\begin{tikzcd}
C^\bullet: \quad \mathbb{F}_2^r \arrow[r, "H^T"] & \mathbb{F}_2^n
\end{tikzcd}
\label{eq:classical_complex_app}
\end{equation}
where \( \mathbb{F}_2^n \) (in degree 1) represents the bits and \( \mathbb{F}_2^r \) (in degree 0) the checks. The map \( H^T \) is the coboundary operator, and the codewords are given by the degree-1 cohomology: \( H^1(C^\bullet) = \ker(H) = \mathrm{coker}(H^T) \).

A choice of basis for \( \mathbb{F}_2^r \) and \( \mathbb{F}_2^n \) corresponds to choosing generators for the checks and a coordinate system for the message space. Code parameters will be denoted by \( [n, k, d] \).

\subsection*{Quantum CSS codes as 3-term cochain complexes}

A quantum CSS code is defined by a pair of classical binary codes \( C_X, C_Z \subseteq \mathbb{F}_2^n \) satisfying the inclusion \( C_Z^\perp \subseteq C_X \), which is equivalent to \( H_X H_Z^T = 0 \). This ensures commutativity of the corresponding stabilizer operators.

This data defines a 3-term cochain complex:
\begin{equation}
\begin{tikzcd}
C^\bullet: \quad \mathbb{F}_2^{r_Z} \arrow[r, "H_Z^T"] & \mathbb{F}_2^n \arrow[r, "H_X"] & \mathbb{F}_2^{r_X}
\end{tikzcd}
\label{eq:css_complex}
\end{equation}
where \( \mathbb{F}_2^n \) (degree 1) represents the qubits, \( H_Z^T \) defines the \( Z \)-type stabilizers, and \( H_X \) defines the \( X \)-type stabilizers. The degree-1 cohomology \( H^1(C^\bullet) = \ker(H_X)/\operatorname{im}(H_Z^T) \) defines the \( Z \)-logical operators, while the degree-1 homology \( H_1(C_\bullet) = \ker(H_Z)/\operatorname{im}(H_X^T) \) defines the \( X \)-logical operators. The number of encoded qubits is
\[
k = \dim H^1(C^\bullet) = \dim H_1(C_\bullet),
\]
and the distance is \( d = \min(d_X, d_Z) \), where \( d_X \) and \( d_Z \) are the minimum weights of nontrivial logical \( X \) and \( Z \) operators, respectively.

More generally, any \( n \)-term (co)chain complex with \( n \geq 3 \) can define a CSS code by selecting three consecutive terms \( C^{i-1} \rightarrow C^i \rightarrow C^{i+1} \).

\section*{Homological products}\label{app:products}

We describe two constructions for combining classical codes to produce quantum codes: the tensor product and the balanced product. The tensor product generalizes the hypergraph product. The balanced product incorporates group symmetries to reduce size while preserving code properties.

\subsection*{Tensor products of cochain complexes}

Let \( (C^\bullet, \delta^C) \) and \( (D^\bullet, \delta^D) \) be cochain complexes of \( \mathbb{F}_2 \)-vector spaces. Their tensor product is a new cochain complex defined by
\begin{equation}
(C \otimes D)^n = \bigoplus_{i + j = n} C^i \otimes D^j,
\label{eq:tensor_chains}
\end{equation}
with coboundary operator
\begin{equation}
\delta(c \otimes d) = \delta^C(c) \otimes d + (-1)^i c \otimes \delta^D(d), \quad \text{for } c \in C^i, \ d \in D^j.
\label{eq:tensor_boundary}
\end{equation}
When working over \( \mathbb{F}_2 \), the sign \( (-1)^i \) may be omitted.

Tensoring an \( m \)-term and an \( n \)-term cochain complex yields an \( (m+n-1) \)-term complex. This construction applies directly to the cochain complexes associated with classical codes. The hypergraph product of two classical codes is the CSS code associated with the 3-term tensor product of their respective 2-term cochain complexes. In Section~4.3, we show how the Künneth formula recovers the standard dimension and distance bounds for hypergraph product codes.

\subsection*{Balanced products} \label{sec::balanced_app}

The balanced product modifies the tensor product by incorporating a group symmetry. This typically reduces the code size while preserving—or sometimes improving—the relative distance and rate.

\subsubsection*{Balanced Product of Vector Spaces}

Let \( V \) and \( W \) be \( \mathbb{F}_2 \)-vector spaces equipped with compatible linear \( G \)-actions. Then \( V \) and \( W \) are modules over the group algebra \( \mathbb{F}_2[G] \). The \emph{balanced product} is defined as
\begin{equation}
V \otimes_G W := V \otimes_{\mathbb{F}_2[G]} W.
\label{eq:balanced_tensor_def}
\end{equation}

This is the quotient of the standard tensor product by the subspace generated by elements of the form
\begin{equation}
g \cdot v \otimes w - v \otimes g \cdot w, \quad \text{for all } g \in G, \ v \in V, \ w \in W.
\label{eq:balanced_tensor_quotient}
\end{equation}

In terms of bases, let \( X \) and \( Y \) be \( G \)-invariant bases for \( V \) and \( W \). Then a basis for \( V \otimes_G W \) is given by the orbits of the anti-diagonal group action:
\begin{equation}
X \times_G Y := X \times Y / \{ (g^{-1}x, gy) : g \in G \}.
\label{eq:balanced_tensor_basis}
\end{equation}

\subsubsection*{Balanced Product of Cochain Complexes}

Let \( G \) be a finite (Abelian) group, and let \( C^\bullet \), \( D^\bullet \) be cochain complexes of \( \mathbb{F}_2 \)-vector spaces equipped with a linear \( G \)-action on each term. Assume the coboundary maps \( \delta^C \), \( \delta^D \) commute with the group action:
\begin{equation}
g \cdot \delta^C(x) = \delta^C(g \cdot x), \quad g \cdot \delta^D(y) = \delta^D(g \cdot y).
\end{equation}

The balanced product complex is defined by
\begin{equation}
(C \otimes_G D)^n = \bigoplus_{i + j = n} C^i \otimes_G D^j,
\label{eq:balanced_chain}
\end{equation}
with coboundary operator
\begin{equation}
\delta = \delta^C \otimes \mathbb{I} + \mathbb{I} \otimes \delta^D.
\label{eq:balanced_boundary}
\end{equation}

This defines a valid cochain complex. As with the tensor product, quantum CSS codes can be constructed by selecting three consecutive terms of the balanced product complex.

We will particularly be interested in the case where $C$ and $D$ are complexes where the cochain spaces are direct sums of $\FG$ itself, as this is the scenario defining group-algebra quantum codes. In this case, the bases for the cochain spaces are given by copies of $G$. For example, we may consider then consider the balanced product of two $2-$term cochain complexes

\[
\begin{tikzcd}
    \FG \arrow[r,"a"] & \FG \quad \quad \quad \FG \arrow[r,"b"] & \FG
\end{tikzcd}
\]
where the coboundary maps are respectively defined by algebra multiplication by $a\in \FG$ and $b \in \FG$. Using the notation $R\equiv \FG$, the balanced product complex is given by 

\begin{equation}
\begin{tikzcd}[column sep = large]
    R\otimes_R R \arrow[r,"\spmat{a \otimes \mathrm{id}  \\ \mathrm{id}\otimes b}"] & R\otimes_R R \oplus R\otimes_R R \arrow[r,"\spmat{\mathrm{id}\otimes b & a \otimes \mathrm{id}}"] & R\otimes_R R
\end{tikzcd}    
\end{equation}

For module tensor products of $R$ with itself, we have a useful isomorphism $R\otimes_R R \cong R$ given by $r\otimes_Rs \cong rs$ (see Appendix A of \cite{eberhardt2024logical} for example). This yields the bicycle code chain complex in \cref{eq:bicycle_complex}. Iterating the balanced product with a third classical complex defined by $c\in \FG$ yields the complex in \cref{eq:tricycle_complex}.

\subsection*{K\"unneth Theorem}

The K\"unneth theorem provides a formula for computing the cohomology of a tensor product of cochain complexes from the cohomologies of the constituent complexes.

Let \( C^\bullet \) and \( D^\bullet \) be cochain complexes of vector spaces over a field, such as \( \mathbb{F}_2 \). Then the K\"unneth theorem states:
\begin{equation}
H^n(C^\bullet \otimes D^\bullet) \cong \bigoplus_{i + j = n} H^i(C^\bullet) \otimes H^j(D^\bullet).
\label{eq:kunneth_vector_space}
\end{equation}
This expression holds exactly when working over a field and allows one to compute the cohomology of the product complex from the cohomologies of the factors.

More generally, suppose \( C^\bullet \) and \( D^\bullet \) are cochain complexes of modules over a finite-dimensional algebra \( A \) over a field, such as \( \mathbb{F}_2[G] \) for a finite group \( G \). Then the K\"unneth formula similarly holds \cite{loormobius}:
\begin{equation}
H^n(C^\bullet \otimes_A D^\bullet) \cong \bigoplus_{i + j = n} H^i(C^\bullet) \otimes_A H^j(D^\bullet),
\label{eq:kunneth_module_case}
\end{equation}
where \( \otimes_A \) denotes the tensor product over the algebra \( A \) -- i.e the balanced product.

\end{document}